%% file: main_icml2022.tex
\documentclass{article}

\usepackage{microtype}
\usepackage{graphicx}
\usepackage{booktabs} 
\usepackage{longtable}
\usepackage{wrapfig}
\usepackage{rotating}
\usepackage[normalem]{ulem}
\usepackage{amsmath}
\usepackage{textcomp}
\usepackage{amssymb}
\usepackage{capt-of}
\usepackage{hyperref}

\usepackage{dsfont}

\usepackage{graphicx,enumitem}

\usepackage{dsfont}
\usepackage{bm}
\usepackage{caption}
\usepackage{subcaption}
\usepackage[noend]{algorithmic}
\usepackage{amsmath}
\usepackage{amsthm}
\usepackage{amssymb}
\usepackage{xcolor}

\newtheorem{theorem}{Theorem}[section]
\newtheorem{lemma}[theorem]{Lemma}

\newtheorem{definition}[theorem]{Definition}

\newcommand{\E}{\mathbb{E}}
\usepackage{hyperref}

\usepackage[accepted]{icml2022}

\icmltitlerunning{Differentially Private Community Detection for Stochastic Block Models}

\begin{document}

\onecolumn
\icmltitle{Differentially Private Community Detection for Stochastic Block Models}



\icmlsetsymbol{equal}{*}

\begin{icmlauthorlist}
\icmlauthor{Mohamed Seif}{equal,to}
\icmlauthor{Dung Nguyen}{equal,goo,ed}
\icmlauthor{Anil Vullikanti}{goo,ed}
\icmlauthor{Ravi Tandon}{to}
\end{icmlauthorlist}

\icmlaffiliation{to}{Department of Electrical and Computer Engineering, University of Arizona.}
\icmlaffiliation{goo}{Biocomplexity Institute and Initiative, University of Virginia.}
\icmlaffiliation{ed}{Department of Computer Science, University of Virginia}

\icmlcorrespondingauthor{Mohamed Seif}{mseif@email.arizona.edu}
\icmlcorrespondingauthor{Dung Nguyen}{dungn@virginia.edu}
\icmlcorrespondingauthor{Ravi Tandon}{tandonr@email.arizona.edu}

\icmlkeywords{Machine Learning, ICML}

\vskip 0.3in



\printAffiliationsAndNotice{\icmlEqualContribution} 

\begin{abstract}
The goal of community detection over graphs is to recover underlying labels/attributes of users (e.g., political affiliation) given the connectivity between users. There has been significant recent progress on understanding the fundamental limits of community detection when the graph is generated from a stochastic block model (SBM). Specifically, sharp information theoretic limits and efficient algorithms have been obtained for SBMs as a function of $p$ and $q$, which 
represent the intra-community and inter-community connection probabilities. In this paper, we study the community detection problem while preserving the privacy of the individual connections between the vertices. Focusing on the notion of $(\epsilon, \delta)$-edge differential privacy (DP), we seek to understand the fundamental tradeoffs between $(p, q)$, DP budget $(\epsilon, \delta)$, and computational efficiency for exact recovery of community labels. 

To this end, we present and analyze the associated information-theoretic tradeoffs for three differentially private community recovery mechanisms: a) stability based mechanism; b) sampling based mechanisms; and c) graph perturbation mechanisms.
Our main findings are that stability and sampling based mechanisms lead to a superior tradeoff between $(p,q)$ and the privacy budget $(\epsilon, \delta)$; however this comes at the expense of higher computational complexity. On the other hand, albeit low complexity, graph perturbation mechanisms require the privacy budget $\epsilon$ to scale as $\Omega(\log(n))$ for exact recovery. 
\end{abstract}

\section{Introduction}
\label{introduction}

Community detection in networks is a fundamental problem in the area of graph mining and machine learning, with many interesting applications such as social networks, image segmentation, and biological networks (see, e.g., the survey by~\cite{fortunato2010community}). The main goal is to partition the network into communities that are ``well-connected''; no standard definition for communities exists, and a large number of methods have been proposed, e.g.,~\cite{Blondel2008FastUO,girvan2002community,holland1983stochastic}, but, in general, there is a limited theoretical basis for the performance of these methods. One exception is the stochastic block model (SBM)~\cite{holland1983stochastic}, which is a probabilistic generative model for generating networks with underlying communities, providing a rigorous framework for detection algorithms. 
In the simplest canonical form of an SBM, the $n$ vertices are partitioned into $r$ communities, and a pair of vertices connect with probability $p$ within communities and with probability $q$ across communities, where $p > q$. ``Recovering'' communities in a graph generated from an SBM (defined formally in Section~\ref{sec:preliminaries_and_problem_statement}) has been a very active area of research, e.g.,~\cite{condon2001algorithms, arias2014community, abbe2015exact, hajek2016achieving}. The exact conditions for recoverability are well understood in terms of the scaling of $p$ and $q$ (more specifically the difference between $p$ and $q$). In particular, in the dense regime (the focus of this paper), with $p = a \log(n)/n$ and $q = b \log(n)/n$, for some constants $a > b > 0$, it is known that exact recovery is possible \emph{if and only if} $\sqrt{a} - \sqrt{b} > \sqrt{r}$ (see~\cite{abbe2017community} for a comprehensive survey). Efficient algorithms for recovering communities have been developed using spectral methods and semi-definite programming (SDP) ~\cite{boppana1987eigenvalues, mcsherry2001spectral,abbe2015exact,  massoulie2014community,gao2017achieving, hajek2016achieving, abbe2020entrywise, wang2020nearly}.

\begin{table*}[t]
    \centering
    \begin{tabular}{| c |c | c | c | c | c |}
    \hline
   & MLE-Stability & SDP-Stability & Bayesian & Exponential  &  RR + SDP   \\ 
    \hline \hline
      $\epsilon$ & \small{$\mathcal{O}(1)$} & \small{$\mathcal{O}(1)$}  & \small{$\Omega(\log(a/b))$} & \small{$\mathcal{O}(1)$} & \small{$\Omega(\log(n))$} \\ 
        \hline
        $\delta$ & $1/n^{2}$ & $1/n^{2}$ & 0 & 0 & $0$  \\
        \hline
        \small{$\sqrt{a} - \sqrt{b} \geq $} &  \small{$\sqrt{2} \cdot \sqrt{1 + 3/2\epsilon}$}  & \small{$\sqrt{2} \cdot \sqrt{2 + 3/2\epsilon}$}  &  Theorem \ref{thm:bayesian-privacy} &  Theorem \ref{thm:utility_exponential} & Theorem \ref{thm:private_threshold_condition} \\ 
        \hline
        Time complexity  & \small{$\mathcal{O}(\exp(n))$} &  \small{$n^{(\mathcal{O}(\log{(n)}))}$} & \small{$\mathcal{O}(\exp(n))$}  &  \small{$\mathcal{O}(\exp(n))$} &  \small{$\mathcal{O}(\operatorname{poly}(n))$}  \\
        \hline
    \end{tabular}
    \vspace{-5pt}
    \caption{Summary of the recovery threshold(s),  complexity and $(\epsilon, \delta)$-edge DP for $r=2$ communities.}
    \label{table:SBM_summary_results_approach2}
    \vspace{-10pt}
\end{table*}

In many applications, e.g., healthcare, social networks, and finance, network data is often private and sensitive, and there is a risk of revealing private information through adversarial queries. Differential Privacy (DP) \cite{dwork2014algorithmic} is the \textit{de facto} standard notion for providing rigorous privacy guarantees. DP ensures that each user's presence in the dataset has minimal statistical influence (measured by the privacy budget $\epsilon$) on the output of queries. Within the context of network/graph data, two privacy models have been considered--- edge and node privacy, and DP algorithms have also been developed for a few network problems, e.g., the number of subgraphs, such as stars and triangles, cuts, dense subgraphs, and communities, and releasing synthetic graphs~\cite{Kasiviswanathan:2013:AGN:2450206.2450232, blocki:itcs13,mulle2015privacy, nguyen2016detecting, qin2017generating, imola2021locally,blocki:itcs13};  most of them focus on edge privacy models, especially when the output is not a count. Finally, there has been very little work on community detection with privacy. \cite{nguyen2016detecting} consider communities based on the modularity. Very recently,~\cite{hehir2021consistency,ji2019differentially} consider community detection in the SBM models subject to edge privacy constraints (also see related work Section \ref{sec:related_work}); however, neither provides any rigorous bounds on the accuracy or the impact of edge privacy on the recovery threshold.

\vspace{-5pt}
\subsection{Contributions}
In this paper, \emph{we present the first differentially private algorithms for community detection in SBMs with rigorous bounds on recoverability}, under the edge privacy model. Informally, a community recovery algorithm satisfies edge privacy if the output has similar distribution irrespective of the 
presence or absence of an edge between any two vertices in the network (see Definition \ref{def:edgeDP}). Edge DP is the most natural privacy notion for community detection, as it involves outputting the partition of the nodes into communities. Our focus is on characterizing the recoverability threshold under edge DP, i.e., how much does the difference between $p$ and $q$ have to change in order to ensure recoverability with privacy. We analyze three classes of mechanisms for this problem.

\noindent
\emph{1. Stability based mechanisms.}
We show that the stability mechanism~\cite{thakurta2013differentially} gives $(\epsilon, \delta)$-DP algorithms for our problem. The main idea is to determine if a non-private community recovery estimator is stable with respect to graph $G$, i.e., the estimate of community structure does not change if a few edges are perturbed; if the estimator is stable, the non-private estimate of community labels can be released; otherwise, we release a random label. We analyze stability based mechanism for two estimators--- the maximum likelihood estimator (MLE), which involves solving a min-bisection problem, and an SDP based estimator. We also derive sufficient conditions for exact recovery for $r =2$ and $r > 2$ communities for both these types of algorithms---these require a slightly larger separation between $p$ and $q$ as a function of the privacy budget $\epsilon$; further, the threshold converges to the well known non-private bound as $\epsilon$ becomes large. The SDP based stability mechanism can be implemented in quasi-polynomial time. 

Stability based mechanisms are less common in the DP literature, compared to other mechanisms, e.g., exponential or randomized response, since proving stability turns out to be very challenging, in general, and is one of our important technical contributions. Stability of the MLE scheme requires showing that the optimum bisection does not change when $k=\mathcal{O}(\log{n})$ edges are perturbed, with high probability. This becomes even harder for the SDP based algorithm, which doesn't always produce an optimum solution. \cite{hajek2016achieving} construct a ``certificate'' for proving optimality of the SDP solution, with high probability. A technical contribution is to identify a new condition that makes the certificate \emph{deterministic}---this is crucial in our stability analysis.


\noindent
\emph{2. Sampling based mechanisms.}
In the second approach, we design two different sampling based mechanisms: (1) Bayesian Estimation and (2) Exponential mechanism. We show that these algorithms are differentially private (with constant $\epsilon$ for Bayesian Estimation and arbitrary small $\epsilon$ for the Exponential mechanism) and guarantee exact recovery under certain regimes of $\epsilon, a, b$; note that, in contrast to the stability based mechanisms, we have $\delta=0$.


\noindent
\emph{3. Randomized Response (RR) based mechanism.}
We also study and analyze a baseline approach, in which one can use a randomized response (RR) technique to perturb the adjacency matrix, and subsequently run an SDP based algorithm for community recovery on the perturbed graph. Due to the post-processing properties of DP, this mechanism satisfies $\epsilon$-DP for any $\epsilon>0$. We show that in contrast to stability and sampling based methods, the baseline RR  approach requires $\epsilon= \Omega(\log(n))$ for exact recovery. 

\noindent
\emph{4. Empirical evaluation.}
We also present simulation results on both synthetic and real-world graphs to validate our theoretical findings (Section \ref{sec:empirical_results}). We observe that the stability based mechanism generally outperforms the others in terms of the error, which is quite small even for fairly small $\epsilon$. Interestingly, the error is low even in real world networks.

\begin{figure}[!t]
\centering
    \begin{minipage}{.25\textwidth}
	\centering
	 \subcaptionbox{$\epsilon = 2$. (\textit{high} privacy regime)}
	{\includegraphics[width=\linewidth]{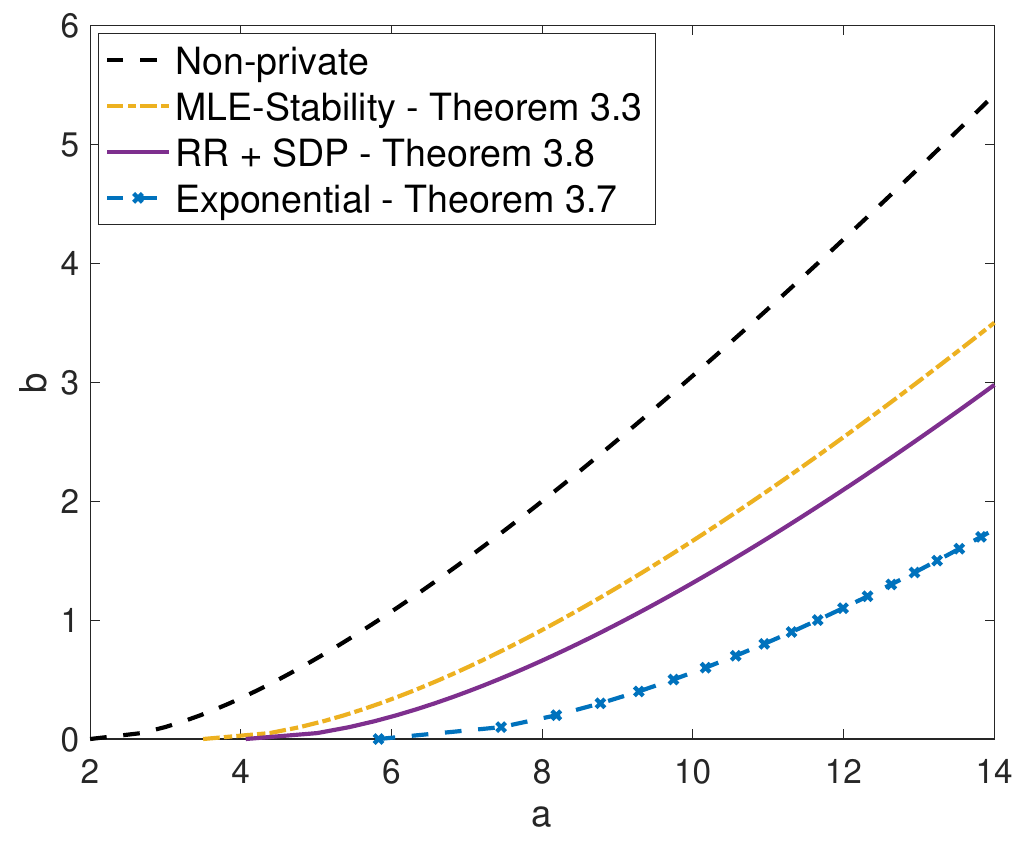}}
	\end{minipage}
	   \begin{minipage}{.25\textwidth}
	\centering
	 \subcaptionbox{$\epsilon = 4$. (\textit{low} privacy regime)}
	{\includegraphics[width=\linewidth]{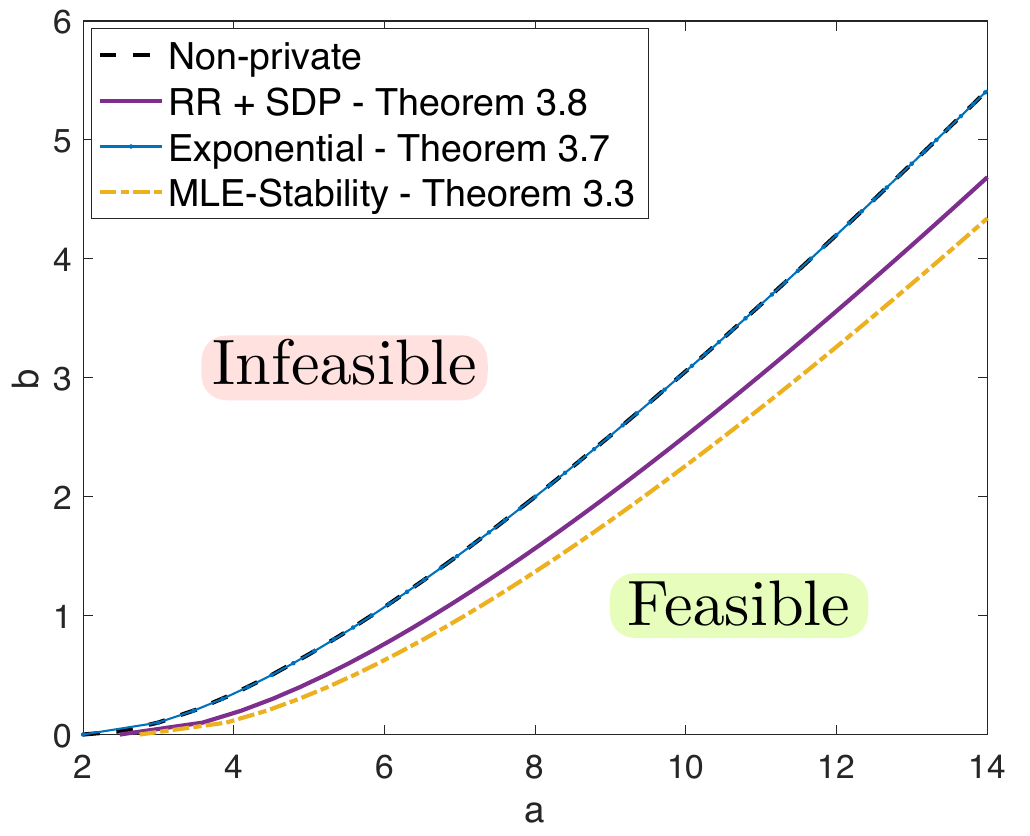}}
	\end{minipage}
	    \caption{Exact Recovery Threshold as a function of $(a, b)$, and the privacy budget $\epsilon$ for $r = 2$ communities.}
	    \label{fig:plot_of_thresholds_private}
	    \vspace{-20pt}
\end{figure}

\noindent
\emph{Comparison between different mechanisms.}
We summarize our theoretical results for differentially private community recovery in Table \ref{table:SBM_summary_results_approach2}, which shows the tradeoffs between $(a, b)$, $(\epsilon, \delta)$ as well as the computational complexity of the mechanisms for $r=2$ communities. Note that none of the mechanisms is redundant--- each is the best in some part of the complex space consisting of the parameters $a, b, \epsilon, \delta$, and the running time.
To further illustrate these tradeoffs, we plot the recovery threshold conditions for these mechanisms in Fig. \ref{fig:plot_of_thresholds_private}. From Fig. \ref{fig:plot_of_thresholds_private}(a), we observe that for the high privacy regime (smaller $\epsilon$), MLE based Stability mechanism requires the least separation between $a$ and $b$ compared to all other algorithms. In the low privacy regime (larger $\epsilon$), as shown in Fig. \ref{fig:plot_of_thresholds_private}(b), we can see that exponential mechanism tends to overlap with the non-private recovery threshold \cite{abbe2015exact}, whereas stability-based and RR based mechanisms require more separation between $a$ and $b$. Complete proofs are presented in the Appendix.

\vspace{-5pt}
\subsection{Related Work}\label{sec:related_work}

We first summarize a few of the main results on the complexity of different recoverability algorithms and then discuss some relevant work on SBMs with DP.
The seminal work of \cite{abbe2015exact} showed that the optimal reconstruction of graph partitions is achieved by the maximum likelihood (ML) estimator, which is computationally intractable.  
\cite{boppana1987eigenvalues, mcsherry2001spectral} designed polynomial time algorithms for exact recovery; however, they did not achieve the optimal information theoretic bound, i.e., $\sqrt{a} - \sqrt{b} > \sqrt{r}$.  \cite{abbe2015exact} showed the first computationally efficient algorithm that achieves the information theoretic limit. This algorithm has two phases:  the first phase performs partial recovery via the algorithm of  \cite{massoulie2014community}. The second phase uses a local improvement to refine the recovery.  \cite{hajek2016achieving} showed that an SDP based rounding algorithm achieves the optimal recovery threshold in polynomial time, and settled the conjecture of \cite{abbe2015exact}. Recently, there have been different computationally efficient recovery algorithms \cite{gao2017achieving, hajek2016achieving, abbe2020entrywise, wang2020nearly}  proposed that achieve the optimal recovery threshold in polynomial time or quasi-linear time for different settings, e.g., multiple communities with different sizes. 

As mentioned earlier, there has been little work on community detection with differential privacy. \cite{nguyen2016detecting} consider the problem of finding communities by modularity maximization.
\cite{qin2017generating} design heuristics for models which are related to SBM. 
Other related work is on estimating parameters of graphons, which are generalizations of SBMs.
\cite{borgs2015private} developed an exponential time algorithm for estimating properties in the node DP model, and derived optimal information theoretic error bounds.  \cite{sealfon2019efficiently} improved this and designed a polynomial time algorithm. 
\cite{hehir2021consistency} study the problem of privacy-preserving community detection on SBMs using a simple spectral method \cite{lei2015consistency} for multiple communities. They generalized the convergence rate analysis of the spectral algorithm and showed the impact of the privacy parameters on the misclassification rate between the ground truth labels and the estimated labels for the algorithm. \cite{ji2019differentially} propose a DP gradient based community detection algorithm. However, neither of these results analyze the thresholds for recoverability, which has remained an open problem (under edge DP constraints) till now.

\vspace{-5pt}
\section{Problem Statement \& Preliminaries} 
\label{sec:preliminaries_and_problem_statement}

 We consider an undirected graph $G = (\mathcal{V}, E)$ consisting of $n$ vertices (vertices), where vertices are divided into $r$ communities  with $\frac{n}{r}$ vertices in each community. The community label for vertex $i$ is denoted by $\sigma^{*}_{i} \in \{1, 2, \cdots, r\}, \forall i \in [n]$. We focus on the setting when the graph $G$ is generated through a Stochastic block model (SBM), where the edges within the classes are generated independently with probability $p$ and the edges between the classes are generated independently with probability $q$. The connections between vertices are represented by an adjacency matrix $\mathbf{A} \in \{0, 1\}^{n \times n}$, where the elements in $\mathbf{A}$ are drawn as:
 \begin{align}
     A_{i,j} \sim \begin{cases} \operatorname{Bern}(p), & i < j, ~~\sigma_{i} = \sigma_{j},  \\
     \operatorname{Bern}(q), & i < j, ~~\sigma_{i} \neq \sigma_{j}.  
     \end{cases}
 \end{align}
with $A_{i,i}=0$ and $A_{i,j}=A_{j,i}$. For the scope of this paper, we focus on the so called ``dense" connectivity regime, where $p = \frac{a \log(n)}{n}$ and $q = \frac{b \log(n)}{n}$, and $a, b\geq 0$ are fixed constants. Note that one can consider other regimes for $p$ and $q$ such as the ``sparse" regime \cite{decelle2011asymptotic}, i.e., $ p = \frac{a}{n}$ and $q = \frac{b}{n}$, however, in this regime exact recovery is not possible since the graph with high probability is not connected. On the other hand, in the dense regime one can still exactly recover the labels of the graph with high probability.
The goal of community detection problem is to design a (stochastic) estimator $\hat{\bm{\sigma}}: \mathbf{A} \rightarrow \{1, 2, \cdots, r\}^{n}$ for community recovery (i.e, the true label vector $\bm{\sigma}^{*} = \{{\sigma}^{*}_{1}, {\sigma}^{*}_{2}, \cdots, {\sigma}^{*}_{n}\}$) upon observing the adjacency matrix. We next define the notion of exact asymptotic recovery as a measure of performance of an estimator. 

 \begin{definition} [Exact Recovery] 
 An estimator $\hat{\bm{\sigma}} = \{\hat{\sigma}_{1}, \hat{\sigma}_{2}, \cdots, \hat{\sigma}_{n}\}$ satisfies exact recovery (upto a global permutation of the community labels) if the probability of error behaves as  
\begin{align}
     \operatorname{Pr}(\hat{\bm{\sigma}} \neq  \bm{\sigma}^{*} ) =  o(1),
\end{align}
where the probability is taken over both the randomness of the graph $G$ as well as the stochastic estimation process. 
\end{definition}
In addition to exact recovery, we require that the recovery algorithm for community detection also protects the individual relationships (i.e., the edges in the graph $G$) in the network. Specifically, we adopt the notion of $(\epsilon, \delta)$-edge differential privacy (DP) \cite{karwa2011private}, defined next.
 \begin{definition} [$(\epsilon, \delta)$-edge DP] \label{def:edgeDP}  An estimator $\hat{\bm{\sigma}}$ satisfies $(\epsilon, \delta)$-edge DP for some $\epsilon \in \mathds{R}^{+}$ and $\delta \in (0, 1]$, if for any pair of adjacency matrices $\mathbf{A}$ and $\mathbf{A}'$ that differ in one edge, we have 
\begin{align}\label{eq:edgeDP}
    \operatorname{Pr}(\hat{\bm{\sigma}}(\mathbf{A}) = \bm{\sigma}) \leq e^{\epsilon}   \operatorname{Pr}(\hat{\bm{\sigma}}(\mathbf{A}') = \bm{\sigma}) + \delta.
\end{align}
For privacy constraints in \eqref{eq:edgeDP}, the probabilities are computed only over the randomness in the estimation process. The case of $\delta = 0$ is called pure $\epsilon$-edge DP.
\end{definition} 
 
\subsection{Prior results on exact recovery without privacy}

The optimal maximum likelihood (ML) estimator for community detection, given by $\hat{\bm{\sigma}}_{\text{ML}} = \arg \max_{\bm{\sigma}} p(\mathbf{A}|\bm{\sigma})$ has been recently analyzed in a series of papers \cite{boppana1987eigenvalues, mcsherry2001spectral, choi2012stochastic, abbe2015exact, mossel2015consistency}. It has been shown that for SBMs with ``dense" regime, i.e.,  $p = \frac{a \log(n)}{n}$ and $q = \frac{b \log(n)}{n}$, exact recovery is possible if and only if $\sqrt{a} - \sqrt{b} > \sqrt{r}$ (often referred to as the phase transition boundary or exact recovery threshold). Even for $r=2$ communities, the ML estimator is equivalent to finding the minimum bisection of the graph, which is known to be NP-hard \cite{abbe2015exact}. Specifically, the ML estimator of $\bm{\sigma}^{*}$  is the solution of the following optimization problem:
\begin{align}
\hat{\bm{\sigma}}_{\text{ML}} = \arg \max_{\bm{\sigma}} \{\bm{\sigma}^{T}\mathbf{A} \bm{\sigma}: \mathbf{1}^{T} \bm{\sigma} =0, \sigma_{i} = \pm 1\}.
\end{align}
Subsequently,
several works have studied if polynomial time algorithms can still achieve the exact recovery threshold. For instance, it has been shown \cite{hajek2016achieving}, \cite{hajek2016achieving_extensions} that SDP relaxation of the ML estimator can also achieve the same recovery threshold. 
Recently, Abbe et.al. \cite{abbe2020entrywise} have analyzed the spectral clustering estimator \cite{lei2015consistency}, and showed that it achieves the same recovery threshold as ML for $r = 2$.


\vspace{-5pt}
\section{Main Results \& Discussions}
\label{sec:main_results}

In this section, we present three different approaches for the design of community detection algorithms for exact recovery while satisfying edge differential privacy.  In the first approach, we analyze the stability property of ML based and SDP based algorithms. For MLE based algorithm, the stability property of the min-bisection hinges on the concentration properties of SBMs in terms of the intra and inter communities edges. For SDP based algorithm, we introduce a concept of concentration that both (1) provides sufficient conditions for the dual certificate of the SDP and (2) persists under certain degrees of connection perturbation. In the second approach, we study and analyze sampling based mechanisms, which release a differentially private estimate of the community labels via sampling.  In the third approach, we perturb the adjacency matrix $\mathbf{A}$ to satisfy DP (using randomized response (RR)), and perform the estimation of community labels using the perturbed graph by using computationally efficient SDP relaxation of the maximum-likelihood estimator. In Table \ref{table:SBM_summary_results_approach2}, we summarize our main results for the case of $r=2$ communities, where we show the constraints on the privacy budget $(\epsilon, \delta)$ and sufficient conditions on $(a,b)$ for exact recovery. 

\vspace{-5pt}
\subsection{Stability-based Mechanisms}
The basic idea behind stability based mechanisms is as follows: Let us consider a non-private estimator for community detection $\hat{\bm{\sigma}}$. We first \emph{privately} compute the stability of this estimator with respect to a graph $G$, which essentially is the \emph{minimum} number of edge modifications on $G$, so that the estimator output on the modified graph $G'$ differs from that on $G$, i.e., $\hat{\bm{\sigma}}(G)\neq \hat{\bm{\sigma}}(G')$. If the graph $G$ is stable enough (i.e., if the estimate of stability is larger than a threshold, which depends on $(\epsilon, \delta)$), then we release the non-private estimate $\hat{\bm{\sigma}}(G)$, otherwise we release a random label vector. The key intuition is that from the output of a stable estimator, one cannot precisely infer the presence or absence of a single edge (thereby providing edge DP guarantee). Before presenting the general stability mechanism, we formally define $d_{\hat{\sigma}}(G)$, which quantifies the stability of an estimator $\hat{\bm{\sigma}}$ with respect to a graph $G$. 
\begin{definition} [Stability of $\hat{\bm{\sigma}}$]
 The stability of an estimator $\hat{\bm{\sigma}}$ with respect to a graph $G$ is defined as follows:

%
%

 \begin{align}
 \label{eqn:stab-sigma2}
    d_{\hat{\bm{\sigma}}}(G) = \{\min_k:  \exists G', \text{dist}(G, G')\leq k+1, \hat{\bm{\sigma}}(G)\neq \hat{\bm{\sigma}}(G')\}. 
\end{align}

\end{definition}
We now present the general stability based mechanism in Algorithm $1$.
\begin{algorithm}
  \caption{$\mathcal{M}^{\hat{\bm{\sigma}}}_{\operatorname{Stability}}(G)$: Stability  Based Mechanism}
  \label{algo:stability}
  \begin{algorithmic}[1]
     \STATE {\bfseries Input:} $G(\mathcal{V}, E) \in \mathcal{G}$
     \STATE {\bfseries Output:} labelling vector $\hat{\bm{\sigma}}_{\text{Private}}$. 
    \STATE $d_{\hat{\bm{\sigma}}}(G) \leftarrow$ stability of $\hat{\bm{\sigma}}$ with respect to graph $G$
    \STATE $\tilde{d} \leftarrow d_{\hat{\bm{\sigma}}}(G) + \operatorname{Lap}(1/\epsilon)$
    \IF{$\tilde{d}  > \frac{\log{1/\delta}}{\epsilon}$}
    \STATE Output $\hat{\bm{\sigma}}(G)$
    \ELSE
    \STATE Output $\perp$ (random label) 
    \ENDIF
  \end{algorithmic}
\end{algorithm}
We first state the following claim about the privacy guarantee of the above mechanism \cite{dwork2014algorithmic}.
\begin{lemma}\label{thm:privacy_guarantee_stability}
  For any community detection algorithm $\hat{\bm{\sigma}}$,  $\mathcal{M}^{\hat{\bm{\sigma}}}_{\operatorname{Stability}}(G)$ satisfies $(\epsilon, \delta)$-edge DP. 
\end{lemma}
In the above algorithm, Step $4$ ensures that the stability is computed privately and Step $5$ ensures that the non-private estimate is released only if the estimator is stable enough (i.e., $\tilde{d}  > \frac{\log{1/\delta}}{\epsilon}$). 

Our first main contribution is to analyze the performance of $\mathcal{M}^{\hat{\bm{\sigma}}}_{\operatorname{Stability}}(G)$ and establish sharp phase transition thresholds for exact recovery as a function of $(p, q)$ and $(\epsilon, \delta)$. Specifically, we focus on two possible choices for $\hat{\bm{\sigma}}$: a) when we use the MLE estimator, i.e., $\hat{\bm{\sigma}}=\hat{\bm{\sigma}}_{\text{MLE}}$, and b) when we use the computationally efficient SDP relaxation, i.e., $\hat{\bm{\sigma}}=\hat{\bm{\sigma}}_{\text{SDP}}$. 

\noindent
\textbf{\emph{Stability of MLE.}}
We start by first presenting the results for MLE based approach for both $r=2$ communities and then for $r>2$ communities. 

\begin{theorem}\label{thm:error_stability_mechanism_two} 


For $r=2$ communities, 
$\mathcal{M}^{\text{MLE}}_{\operatorname{Stability}}(G)$ satisfies exact recovery if 
\begin{align}
\label{eqn:mle-2comm}
 \sqrt{a} - \sqrt{b} > \sqrt{2} \times \sqrt{1 + \frac{t+1}{2\epsilon}}
\end{align}
for any $\epsilon>0$ and $\delta = n^{-t}$, $t > 0$.  


\end{theorem}

We note two important points:
\emph{(1) In contrast to the non-private recovery threshold $\sqrt{a} - \sqrt{b} > \sqrt{2}$, the impact of edge DP shows up explicitly in the threshold condition; (2) As we relax the relax the privacy budget, namely as $\epsilon\rightarrow \infty$, the privacy constrained threshold converges to the non-private threshold.} 
We next generalize our results to $r>2$ equal sized communities and present a sufficient condition on $a$ and $b$ for exact recovery.
\begin{theorem}\label{thm:error_stability_mechanism_general} 


For $r>2$ communities, 
$\mathcal{M}^{\text{MLE}}_{\operatorname{Stability}}(G)$ satisfies exact recovery if 
\begin{align}
\label{eqn:mle-rcomm}
    \sqrt{a} - \sqrt{b} > \sqrt{r} \times \sqrt{ 1 + \frac{t+1}{\epsilon} \times \bigg(1 + \log \sqrt{\frac{a}{b}} \bigg)}
 \end{align}
for any $\epsilon>0$ and $\delta = n^{-t}$, $t>0$.  


\end{theorem}

The result for $r>2$ communities is slightly weaker compared to the case for $r=2$ case. However, \emph{it still converges to the non-private optimal threshold $(\sqrt{a} - \sqrt{b} > \sqrt{r})$ when the privacy budget $\epsilon\rightarrow \infty$}. 

\emph{Main Ideas behind the Proof(s) of Theorems \ref{thm:error_stability_mechanism_two} and \ref{thm:error_stability_mechanism_general} and Intuition behind the private recovery threshold:} Analyzing the error probability for the stability based mechanism for SBM is highly non-trivial. Specifically, there are two types of error events occur in this mechanism when estimating the true labels $\bm{\sigma}^{*}$: $(1)$ When the stability mechanism outputs the ML estimate $\hat{\bm{\sigma}}_{\text{MLE}}$, then we are interested in bounding the $\operatorname{Pr}(\hat{\bm{\sigma}}_{\text{MLE}} \neq \bm{\sigma}^{*})$. This error probability can be analyzed using existing results on exact recovery \cite{abbe2015exact}, and the error vanishes as $o(1)$ if $\sqrt{a} - \sqrt{b} > \sqrt{r}$. 
$(2)$ The second source of error is when the mechanism outputs a random label $\perp$, whose probability is bounded by $\operatorname{Pr}(\tilde{d}  \leq \frac{\log{1/\delta}}{\epsilon})$.  The key technical challenge arises in the analysis of this probability. Specifically, we show that when the graph $G$ is drawn from an SBM, the ML estimator is $\Omega(\log(n))$-stable with high probability. By leveraging this result, we bound the probability
$\operatorname{Pr}(\tilde{d}  \leq \frac{\log{1/\delta}}{\epsilon})$, and in order to make this probability decay as $o(1)$ for exact recovery, we obtain sufficient conditions on $(a, b)$ presented in Theorems \ref{thm:error_stability_mechanism_two} and \ref{thm:error_stability_mechanism_general}. 


\textbf{\emph{Stability of SDP relaxation.}}
We show that the SDP relaxation (SDP for short) method also has the stability property, i.e., a graph $G$ generated by an SBM is $\Omega(\log{n})$-stable with respect to the SDP with high probability, which gives us the following result for both $r=2$ and $r>2$ multiple equal-sized communities.





\begin{theorem} 
\label{theorem:sdp-rcomm}

For $r\geq 2$ communities, 
$\mathcal{M}^{\text{SDP}}_{\operatorname{Stability}}(G)$ satisfies exact recovery if 
\begin{align}
\label{eqn:sdp-rcomm}
\sqrt{a}-\sqrt{b} &> \sqrt{r}\times \sqrt{2 +  \frac{t+1}{\epsilon}\left( 1 + \log{\sqrt{\frac{a}{b}}}\right)}
\end{align}
for any $\epsilon >0$ and $\delta = n^{-t}$, $t > 0$.


\end{theorem}


In contrast with the threshold condition (\ref{eqn:mle-rcomm}), we have a larger constant in  (\ref{eqn:sdp-rcomm}) for $\mathcal{M}^{\text{SDP}}_{\operatorname{Stability}}(G)$, arising out of the concentration bounds for the SDP relaxation algorithm.



\noindent
\emph{Main ideas in the proof of Theorem~\ref{theorem:sdp-rcomm}.}
The proof of the stability of $\mathcal{M}^{\text{SDP}}_{\operatorname{Stability}}(G)$ becomes more complex than that of MLE, because SDP only takes the ground truth label as the optimal solution in some regimes; further, arguing that a solution $\hat{\bm{\sigma}}=\hat{\bm{\sigma}}_{\text{SDP}}$ is not easy (since it may not be the min bisection). 
\cite{hajek2016achieving} design a sophisticated ``certificate'' for proving that the SDP solution is indeed the optimal, and show that the certificate holds with high probability when $\sqrt{a} - \sqrt{b} > \sqrt{r}$ (note that this certificate is much more complex than the primal-dual based certificate used earlier for $r=2$ communities~\cite{lecture7}). The high probability bound for the certificates is unfortunately not sufficient, since we need to argue about the stability for a graph $G$ generated from the SBM deterministically, and there are $n^{\mathcal{O}(\log{n})}$ graphs within distance $\mathcal{O}(\log{n})$ of $G$. Specifically, the high probability bound for the certificate does not hold after flipping $\Omega(\log{n})$ connections, which is required to maintain the stability of the optimal solution. Instead, we define a notion of ``concentration'', and show that if a graph is concentrated, then $SDP(G)$ is optimal at the ground truth label; \emph{note that this holds deterministically, not with high probability}. We then use this notion of concentration to determine stability, by showing that all graphs within $\mathcal{O}(\log{n})$ distance of $G$ are also concentrated. Finally, we derive a lower bound on $\sqrt{a} - \sqrt{b}$ that is both (1) sufficient for concentration and (2) able to preserve concentration after flipping up to $\Omega(\log{n})$ connections. We give more details below.

We say that a graph is  $(c_1, c_2, c_3, c_4)$-concentrated, for constants $c_1,\ldots,c_4$, if the following four conditions hold:
\begin{itemize}[leftmargin=0.2in,topsep=0pt,itemsep=0pt]
  \item $\min_{i\in V(G)} (s_i-r_i) > c_1 \log{n}$, where $s_i$ is the number of same-community neighbors of $i$ and $r_i$ is the maximum number of neighbors of i in one of the other communities.
  \item $\Vert \mathbf{A} - \E[\mathbf{A}] \Vert_2 \leq c_2\sqrt{\log{n}}$
  \item $\max_{k\in[r]}\frac{1}{K}\sum_{i\in C_k}r_i \leq Kq + c_3\sqrt{\log{n}}$, \mbox{where $K=n/r$}.
  \item $e(C_k, C_{k'}) \geq K^2q -3/4K\sqrt{\log{n}}-c_4\log{n}$, where $e(C_k, C_{k'})$ is the number of inter-community edges between communities $k$ and $k'\neq k$
  \end{itemize}
  
Next, we prove that a graph generated by an SBM with appropriate parameters will be $(c_1,\dots, c_4)$-concentrated w.h.p.. The concentration holds with high probability only when $a$ and $b$ satisfies some conditions related to $r$ i.e., $a$ must be large enough (relatively to $b$) and they will determine the exact recovery threshold of the method. Next, we prove that the concentration persists under $\Omega(\log{n})$ edge perturbations, i.e., that if the original graph is concentrated under a tuple $(c_1, \dots, c_4)$, a graph obtained by flipping up to $\Omega(\log{n})$ connections of the original one is also concentrated with slightly different tuple.

We then apply the analyses of~\cite{hajek2016achieving_extensions} to prove that when a graph is $(c_1, \dots, c_4)$-concentrated for some constants $c_i$, the SDP relaxation (SDP for short) outputs the (1) uniquely optimal solution and (2) the optimal solution is also the ground truth community vector. Our proof differs from~\cite{hajek2016achieving_extensions}'s proof in a way that~\cite{hajek2016achieving_extensions}'s conditions holds with high probability and ours holds deterministically. First we note that the SDP can be presented by the following form:

  \begin{align*}
    \mbox{maximize  }&\langle \mathbf{A}, \mathbf{Z} \rangle \\
    \mbox{subsect to  } \mathbf{Z} &\curlyeqsucc 0\\
    Z_{ii} &= 1, \forall i \in [n] \\
    Z_{ij} &\geq 0, \forall i, j \in [n]\\
    \mathbf{Z}\bm{1} &= K\bm{1},
  \end{align*}

\noindent
Then we provide the condition for a dual certificate (deterministically). Intuitively, if we can construct a positive semi-definite matrix $S^*$ by the following formula without violating the two constraints below, the SDP is uniquely optimal at $\mathbf{Z}^*$ constructed by the ground truth community label (We say SDP(G) is optimal at the ground truth community label for short).

  \emph{Lemma 6 of~\cite{hajek2016achieving_extensions}}.
  Suppose there exists $\mathbf{D}^* = diag(d_i^*)$ with $d_i^*>0$  for all $i, \mathbf{B}^*\in\mathcal{S}^n$ with $\mathbf{B}^*\geq \mathbf{0}$ and $B_{ij}>0$ whenever $i$ and $j$ are in distinct clusters, and $\lambda^*\in\mathbb{R}^n$ such that $\mathbf{S}^* \triangleq \mathbf{D}^* - \mathbf{B}^* - \mathbf{A} + \mathbf{\lambda}^*\bm{1}^T + \bm{1}(\mathbf{\lambda}^*)^T$ satisfies $\mathbf{S}^*\curlyeqsucc0$ and
  \begin{align*}
    \mathbf{S}^*\mathbf{\xi}^*_k &= 0, \forall k\in[r]\\
    B_{ij}^*Z_{ij}^* &= 0, \forall i,j \in[n]
  \end{align*}
  Then $\text{SDP}(G) = \mathbf{Z}^*$ is the unique solution for the SDP.

We then prove that the concentration of the input graph implies the existence of a positive semi-definite matrix $S^*$, which satisfies the dual certificate above, i.e., we point out that there's always a way to construct matrices $\mathbf{D}^*, \mathbf{B^*}$ that satisfies above conditions from the concentration's conditions. We note that when such $\mathbf{S^*}$ exists, the SDP will uniquely output the ground truth community vector.

Therefore if a graph $G$ (with size $n$ large enough) is generated by an SBM with the ground truth community vector and $G$ is concentrated, $SDP(G)$ will outputs $\mathbf{Z^*}$. We also know that any $G'$ obtained by flipping up to $c\log{n}/\epsilon$  edges of $G$ is also concentrated (for some constant $c$). It means that $SDP(G')$ also outputs $\mathbf{Z^*}$ and proves that $SDP$ is $c\log{n}/\epsilon$-stable. Compose with the fact that a graph generated by such SBM will be concentrated with high probability, we conclude that $SDP$ is $c\log{n}/\epsilon$-stable with high probability. The threshold for Theorem~\ref{theorem:sdp-rcomm} derives from the conditions of $a, b$ and $r$ for which the concentration holds with high probability and choosing the constant $c$ accordingly to $\delta$.

\noindent
\emph{Complexity of Stability Based Mechanisms.}
A naive implementation of $\mathcal{M}^{\hat{\bm{\sigma}}}_{\operatorname{Stability}}(G)$, which involves computing $d_{\hat{\bm{\sigma}}}(G)$ in Step 3 using (\ref{eqn:stab-sigma}), requires computing $\hat{\bm{\sigma}}(G')$ for all graphs $G'$. It can be shown that the algorithm works if we use $\min\{d_{\hat{\bm{\sigma}}}(G), \mathcal{O}(\log{n})\}$, instead of $d_{\hat{\bm{\sigma}}}(G)$, for which
it  suffices to compute $\hat{\bm{\sigma}}(G')$ for only those graphs $G'$ with $d(G, G')= \mathcal{O}(\log{n})$. The MLE algorithm takes exponential time, so algorithm $\mathcal{M}^{\text{MLE}}_{\operatorname{Stability}}(G)$ still takes exponential time; however, $\mathcal{M}^{\text{SDP}}_{\operatorname{Stability}}(G)$ can be implemented in quasi-polynomial time, i.e., $n^{(\mathcal{O}(\log{(n)}))}$, using the above observation.

\subsection{Sampling  Mechanisms}

We present two sampling based approaches for private community detection. In the first approach of \emph{Bayesian Sampling}, presented in Algorithm \ref{algo:bayesian}, we compute the posterior probability of label vectors given the graph $G$ and release a label estimate by sampling from this posterior distribution. 

\begin{algorithm}
  \caption{$\mathcal{M}_{\text{Bayesian}}(G)$: Bayesian Sampling Mechanism
  }
  \begin{algorithmic}[1]
       \STATE {\bfseries Input:} $G(\mathcal{V}, E) \in \mathcal{G}$
         \STATE {\bfseries Output:}  A labelling vector $\hat{\bm{\sigma}} \in \mathcal{L}$. 
    \STATE For every $\bm{\sigma} \in{\mathcal{L}}$, calculate $p(\bm{\sigma}|G) = \frac{p(\bm{\sigma}) \times p(G|\bm{\sigma})}{p(G)}$ 
    \STATE Sample and output a labelling $\hat{\bm{\sigma}}\in{\mathcal{L}}$ with probability $\Pr(\hat{\bm{\sigma}}|G)$
      \label{algo:bayesian}
  \end{algorithmic}
\end{algorithm}

Surprisingly, we show that this mechanism satisfies pure $\epsilon$-edge DP whenever $\epsilon$ is larger than a threshold, namely, $\epsilon\geq \log(a/b)$. This is in-contrast with Stability mechanisms which achieve approximate $(\epsilon, \delta)$-edge DP, for any $\epsilon>0$ but require $\delta = 1/n^{t}$, for any $t > 0$. Our main result for the Bayesian mechanism is stated in the following theorem along with the corresponding recovery threshold. 

\begin{theorem} 
  \label{thm:bayesian-privacy}
  The mechanism $\mathcal{M}_{\text{Bayesian}}$(G) satisfies $\epsilon$-edge DP,  $\forall \epsilon \geq \epsilon_0 = \log \big(\frac{a}{b}\big)$, and for $r=2$ communities, satisfies exact recovery if 
\begin{align}
     \sqrt{a} - \sqrt{b} & > \max \left[ \sqrt{2}, \frac{{2}}{(\sqrt{2} - 1)(1- e^{- \epsilon_{0}})} \right]. 
 \end{align}
 \end{theorem}
 
Despite the fact that the Bayesian mechanism provides pure edge DP, one disadvantage is that it requires the knowledge of $(a, b)$ for computing the posterior distribution. To this end, we present and analyze the exponential sampling mechanism in Algorithm \ref{algo:expo_1}, where we sample from a distribution over the labels which can be computed directly from the graph and does not require the knowledge of $(a,b)$. Specifically, for any label vector $\bm{\sigma}$ (partition of the graph in two communities), the $\text{score}(\bm{\sigma}) = {E}_{\text{inter}}(G, \bm{\sigma})$ is defined as the set of cross-community edges in the partition $\bm{\sigma}$, the corresponding sampling probability is computed as a function of this score and the privacy budget. 

\begin{algorithm}
  \caption{$\mathcal{M}_{\text{Expo.}}(G)$: Exponential Mechanism
  }
  \begin{algorithmic}[1]
         \STATE {\bfseries Input:} $G(\mathcal{V}, E) \in \mathcal{G}$
         \STATE {\bfseries Output:} A labelling vector $\hat{\bm{\sigma}} \in \mathcal{L}$.
    \STATE For every $\bm{\sigma} \in \mathcal{L}$, calculate $\text{score}(\bm{\sigma}) = {E}_{\text{inter}}(G, \bm{\sigma})$ 
    \STATE Sample and output a labelling $\hat{\bm{\sigma}}\in \mathcal{L}$ with probability $\exp(-\epsilon\times \text{score}(\bm{\sigma}))$
          \label{algo:expo_1}
  \end{algorithmic}
\end{algorithm}

\begin{theorem}\label{thm:utility_exponential}
  The exponential sampling mechanism  $\mathcal{M}_{\text{Expo.}}(G)$ satisfies $\epsilon$-edge DP and for $r=2$ communities, performs exact recovery if 
  \begin{align}
     \sqrt{a} - \sqrt{b} & > \max \left[\sqrt{2}, \frac{{2}}{(\sqrt{2} - 1) \epsilon} \right].
 \end{align}
\end{theorem}


\noindent
\emph{Complexity and comparison with stability based mechanisms.}
A key advantage of the sampling based mechanisms over stability based mechanisms is that they give $\epsilon$-DP solutions. However, implementing the sampling step in these mechanisms takes exponential time, as no efficient algorithm is known for sampling $\bm{\sigma}$ with probability depending on its utility.

\subsection{Graph Perturbation Mechanisms}

In this section, we present and analyze randomized response (RR) based mechanism for private community detection. The basic idea is to perturb the edges of the random graph (i.e., the adjacency matrix $\mathbf{A}$), where each element $A_{i,j}$ is perturbed independently to satisfy $\epsilon$-edge DP. For a graph with an adjacency matrix $\mathbf{A}$, the perturbed matrix is denoted as $\tilde{\mathbf{A}}$, where $\mu = \operatorname{Pr}(\tilde{A}_{i,j} = 1 | A_{i,j} = 0) =  \operatorname{Pr}(\tilde{A}_{i,j} = 0 | A_{i,j} = 1) $. By picking $\mu = \frac{1}{e^{\epsilon} + 1}$, it can be readily shown that the mechanism satisfies $\epsilon$-edge DP. One can then apply any community recovery algorithm (MLE, SDP or spectral methods) on the perturbed matrix $\tilde{\mathbf{A}}$. This mechanism is presented in Algorithm \ref{algo:randomized_response}.

\begin{algorithm}
  \caption{$\mathcal{M}^{\hat{\bm{\sigma}}}_{\text{RR}}(G)$: Graph Perturbation Mechanism via Randomized Response
  }
  \begin{algorithmic}[1]
         \STATE {\bfseries Input:} $G(\mathcal{V}, E) \in \mathcal{G}$
         \STATE {\bfseries Output:} A labelling vector $\hat{\bm{\sigma}} \in \mathcal{L}$.
         \STATE Perturb $\mathbf{A} \rightarrow \tilde{\mathbf{A}} $ via randomized response mechanism
         \STATE Apply community detection algorithm on $\tilde{\mathbf{A}}$ 
         \STATE Output $\hat{\bm{\sigma}}(\tilde{\mathbf{A}})$
          \label{algo:randomized_response}
  \end{algorithmic}
\end{algorithm}

From the perspective of computational complexity, this Algorithm is faster compared to the stability and sampling based approaches. However, in the next Theorem, we state our main result which shows that RR based mechanism achieves exact recovery if $\epsilon = \Omega (\log(n))$, i.e., it requires the privacy leakage to grow with $n$ for exact recovery. 

\begin{theorem} \label{thm:private_threshold_condition} The mechanism $\mathcal{M}^{\text{SDP}}_{\operatorname{RR}}(G)$ satisfies $\epsilon$-edge DP,  $\forall \epsilon \geq \epsilon_{n} = \Omega(\log(n))$, and for $r=2$ communities,  satisfies exact recovery if 
\begin{align}
    \sqrt{a} - \sqrt{b} > \sqrt{2} \times \frac{\sqrt{e^{\epsilon} + 1}}{\sqrt{e^{\epsilon} - 1}} + \frac{1}{\sqrt{e^{\epsilon} - 1}}.
\end{align}
\end{theorem}

In order to understand the intuition behind the worse privacy leakage of RR mechanism for exact recovery, it is instructive to consider the statistics of the perturbed adjacency matrix $\tilde{\mathbf{A}}$ as a function of $\epsilon$. Specifically, the perturbed elements in the adjacency matrix $\tilde{\mathbf{A}}$ are distributed as follows
  $   \tilde{A}_{i,j} \sim  \operatorname{Bern}(\tilde{p}),  i < j, \text{if} \hspace{0.05in} \sigma_{i} = \sigma_{j}$, and $\tilde{A}_{i,j} \sim \operatorname{Bern}(\tilde{q}),  i < j, \text{if} \hspace{0.05in} \sigma_{i} \neq \sigma_{j}$, where
 \begin{align}
     \tilde{p} & =  \underbrace{\left[\frac{n}{(e^{\epsilon} + 1) \times \log(n)} + \frac{e^{\epsilon} -1}{e^{\epsilon} + 1} \times  a \right]}_{a_{n}} \times \frac{\log (n)}{n}, \nonumber \\
     \tilde{q} & =  \underbrace{\left[\frac{n}{(e^{\epsilon} + 1) \times \log(n)} + \frac{e^{\epsilon} -1}{e^{\epsilon} + 1} \times  b \right]}_{b_{n}} \times \frac{\log (n)}{n}.
 \end{align}
Note that $\tilde{p}$ and $\tilde{q}$ are the intra- and inter- community connection probabilities for the perturbed matrix. From the above, we note that if $\epsilon$ is chosen as a constant, and as $n$ grows, then $\lim_{n\rightarrow \infty} \tilde{p} = \lim_{n\rightarrow \infty} \tilde{q}$, i.e., if we insist on constant $\epsilon$, then asymptotically, the statistics of the inter- and intra-community edges are the same and exact recovery is impossible. The result of Theorem \ref{thm:private_threshold_condition} shows that one can indeed get exact recovery by allowing the leakage to grow logarithmically with $n$.

\vspace{-7pt}
\section{Numerical Experiments}
\label{sec:empirical_results}

\begin{figure*}[t]
\centering
    \begin{minipage}{.24\textwidth}
	\centering
	 \subcaptionbox{Impact of changing $a$ where $r = 2$ and $n = 100$ vertices.}
	{\includegraphics[width=\linewidth]{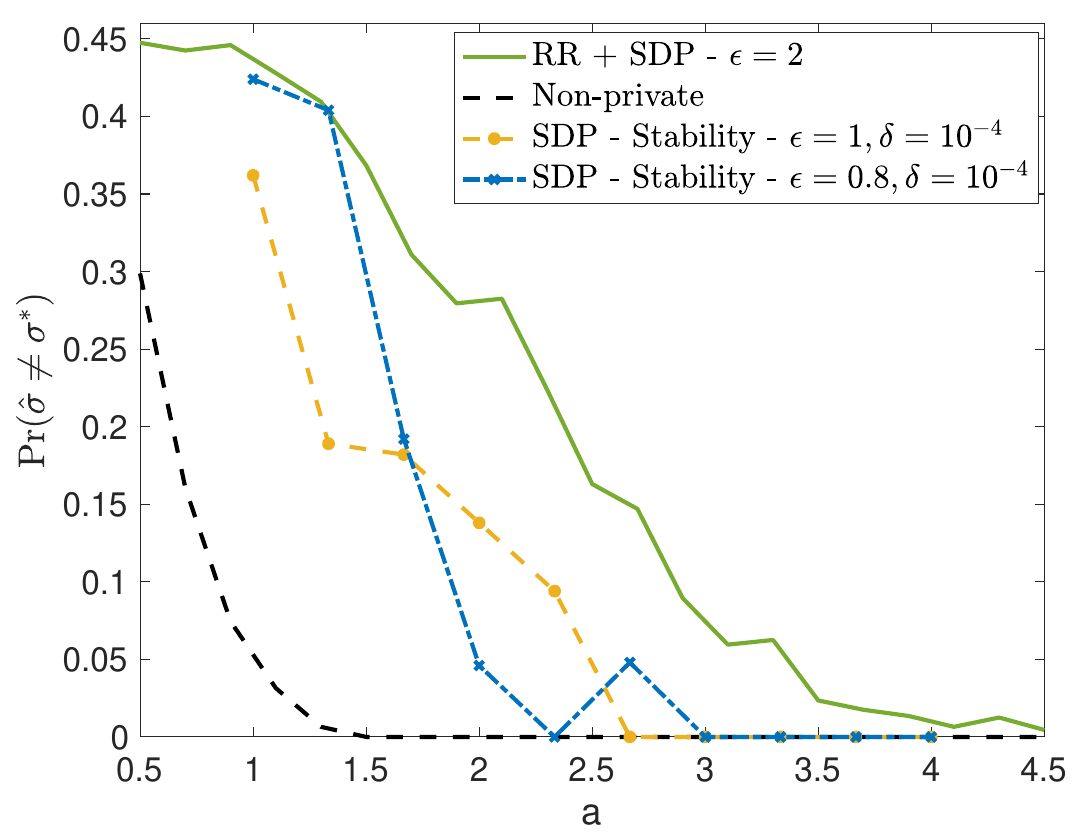}}	
	\end{minipage}
	    \begin{minipage}{.24\textwidth}
	\centering
	 \subcaptionbox{Impact of $\epsilon$ where $r = 2$ and $n = 200$ vertices.}
	{\includegraphics[width=\linewidth]{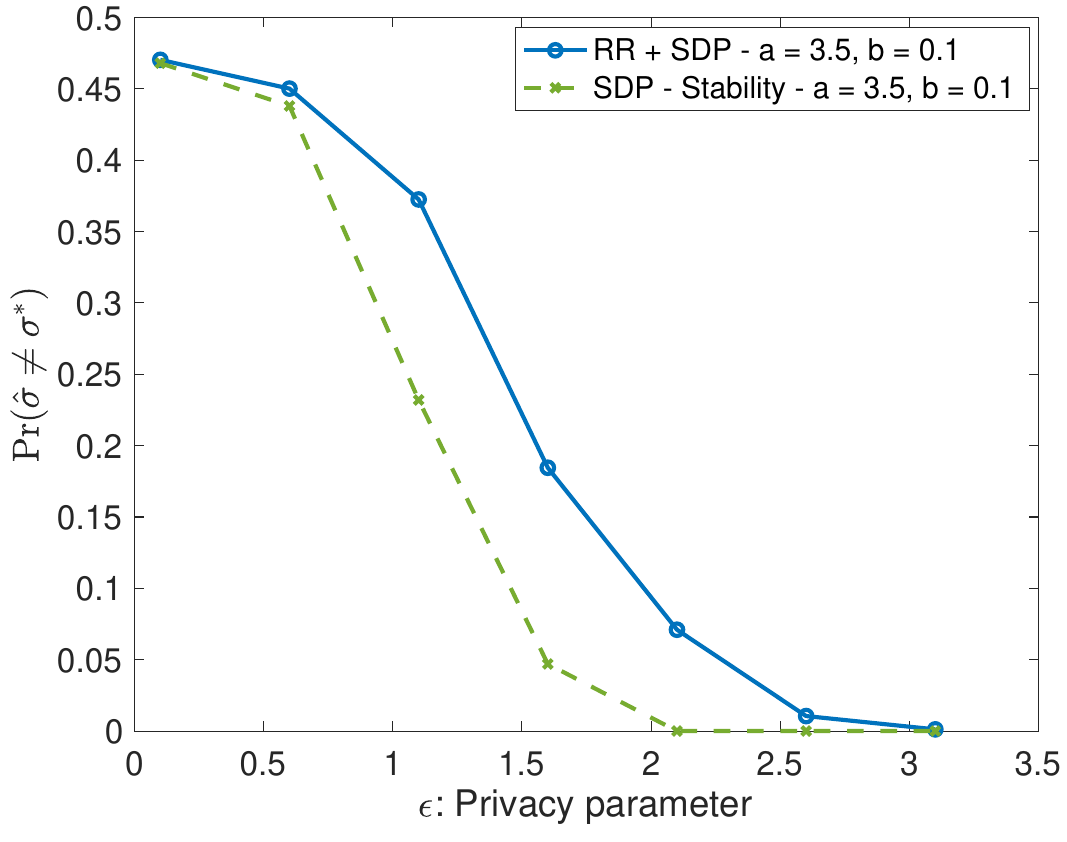}}
	\end{minipage}
	  \begin{minipage}{.24\textwidth}
	\centering
	 \subcaptionbox{Impact of $\epsilon$ where $r = 3$ and $n = 200$ vertices.}
	{\includegraphics[width=\linewidth]{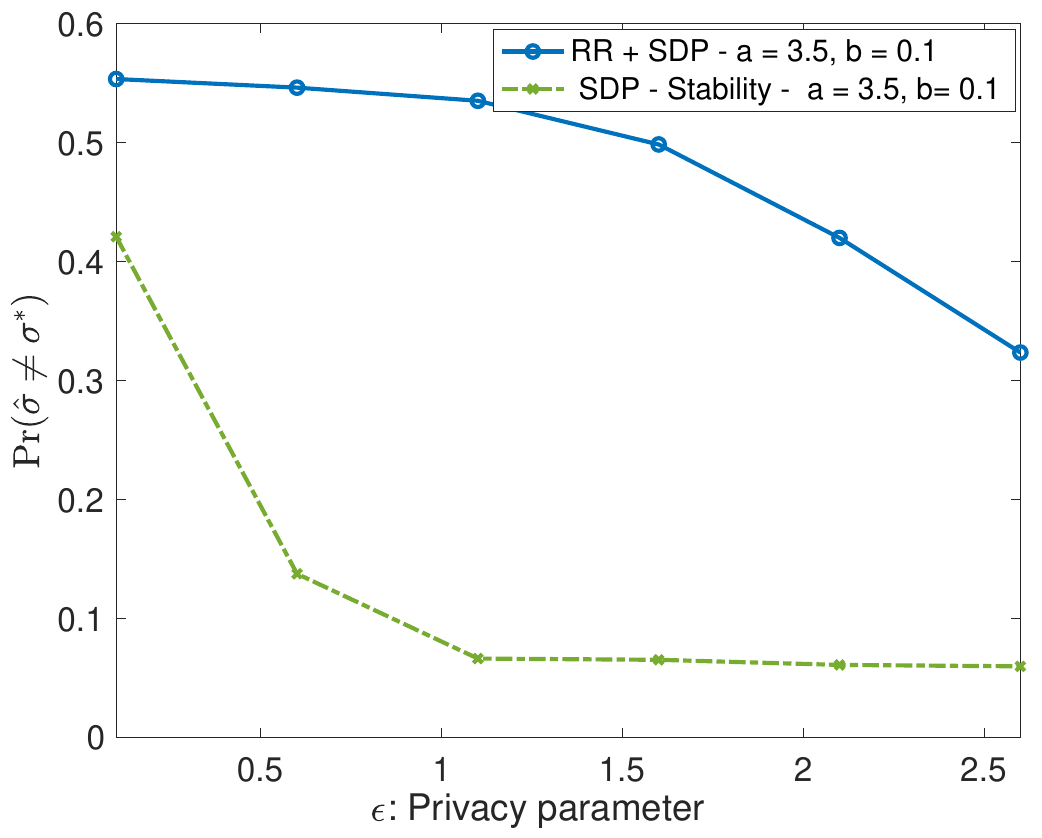}}
	\end{minipage}
	  	  \begin{minipage}{.24\textwidth}
	\centering
	 \subcaptionbox{SDP vs. Spectral method for $r = 2$ communities.}
	{\includegraphics[width=\linewidth]{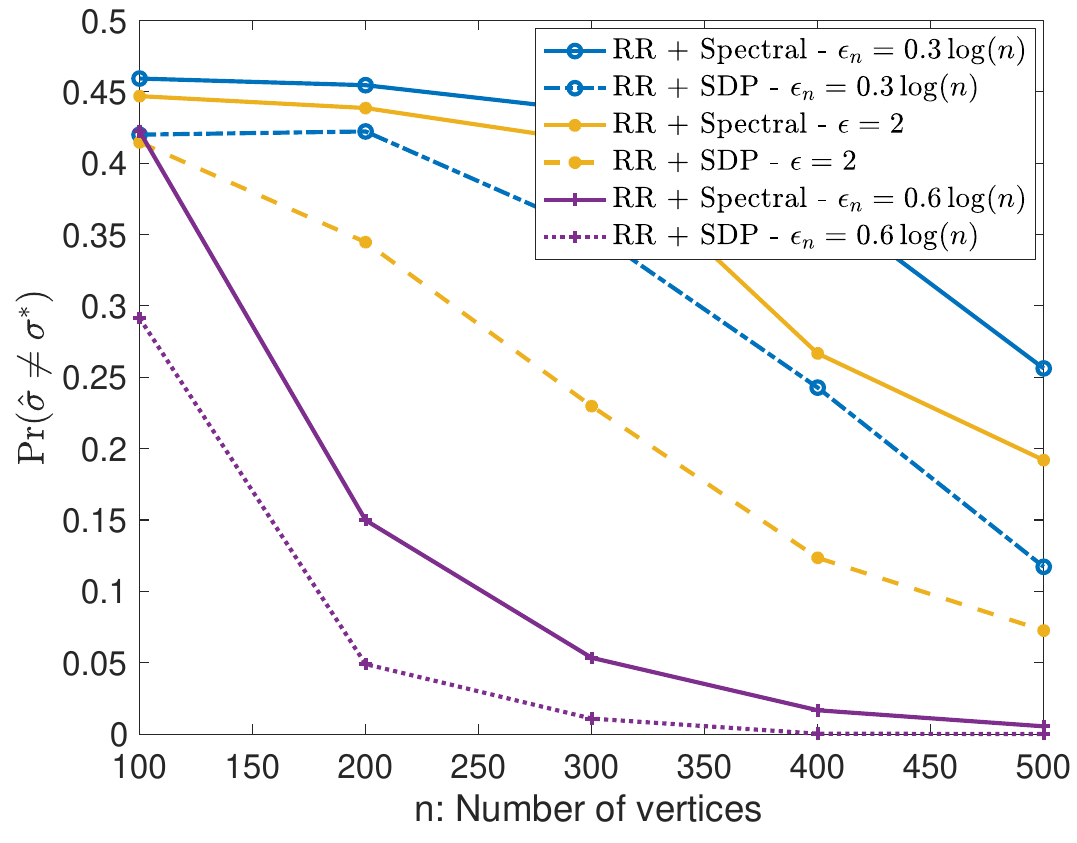}}
	\end{minipage}
 \begin{minipage}{.24\textwidth}
	\centering
	 \subcaptionbox{Comparison between stability and RR + SDP where $r = 2$.}
	{\includegraphics[width=\linewidth]{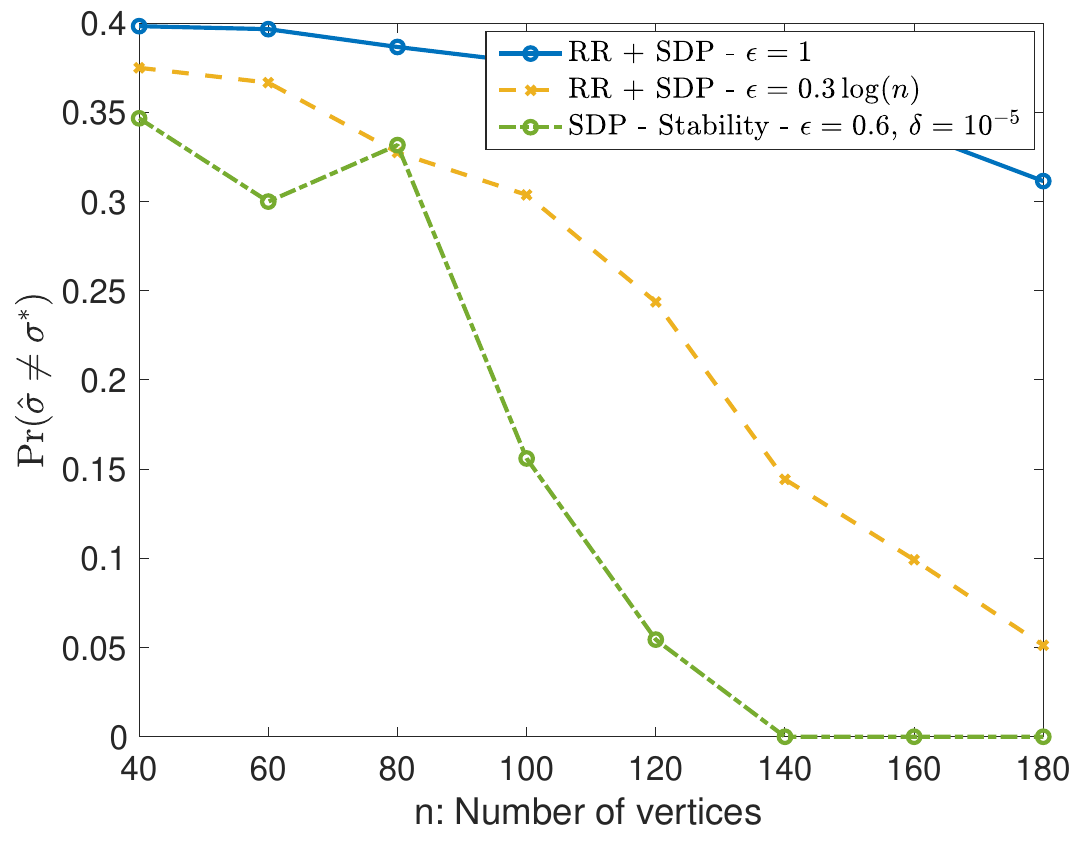}}
	\end{minipage}
	   \begin{minipage}{.24\textwidth}
	\centering
 	 \subcaptionbox{Comparison between stability and RR + SDP where $r = 3$.}
	{\includegraphics[width=\linewidth]{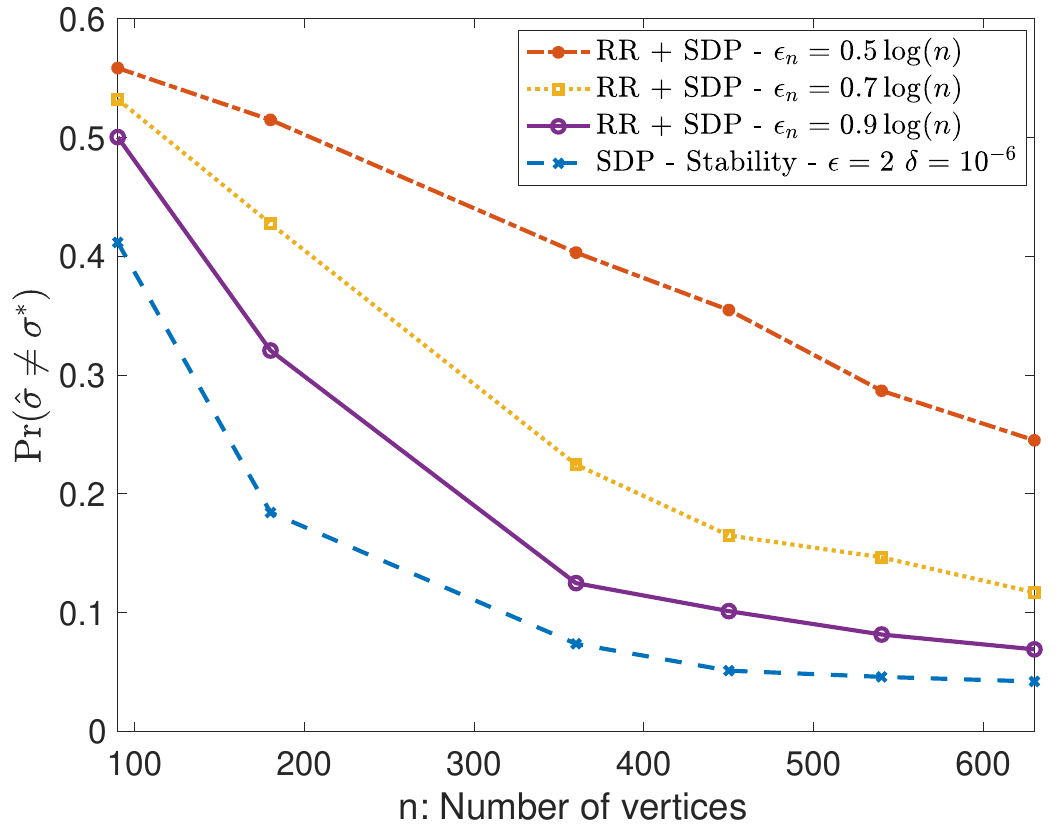}}
 	\end{minipage}
 		  \begin{minipage}{.24\textwidth}
 	\centering
	 \subcaptionbox{Impact of $\epsilon$ for Karate Club dataset where $r = 2$. }
	{\includegraphics[width=\linewidth]{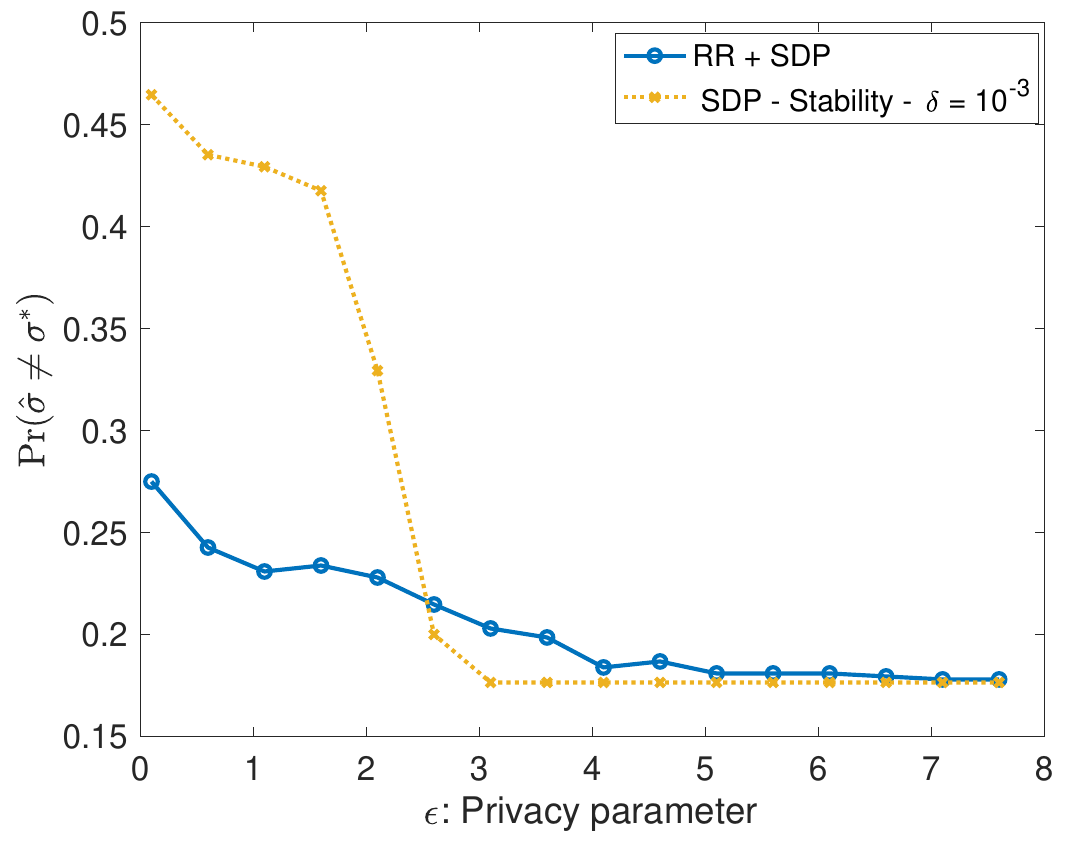}}
	\end{minipage}
 	\begin{minipage}{.24\textwidth}
 	\centering
	 \subcaptionbox{Impact of $\epsilon$ for Political Blogosphere dataset where $r =2$.}
	{\includegraphics[width=\linewidth]{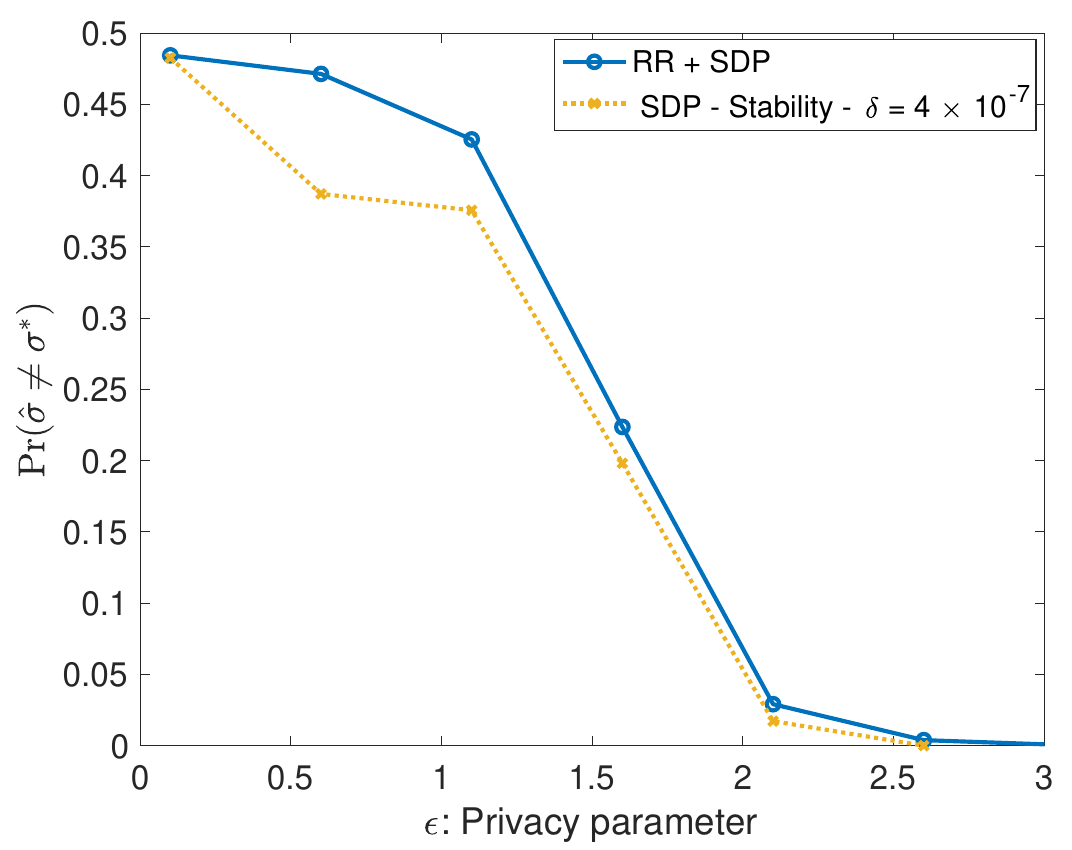}}
	\end{minipage}
	\caption{Synopsis of Numerical results: (a) Shows the impact of changing $a$ for fixed $b, \epsilon$. (b)-(c) Show the impact of $\epsilon$ for $r = 2, 3$, respectively. (d)-(f) Show the error probability as function of $n$. (g)-(h) Show the performance on real-world datasets.}
    \label{fig:impact_changing_parameters}
     \vspace{-10pt}
\end{figure*}

In this section, we present experimental results to assess the performance of our proposed private community detection algorithms, and the associated tradeoffs between privacy and community recovery for both synthetically generated graphs (SBMs) as well as real-world graphs. The proposed mechanisms are implemented in MATLAB 2020b and the optimization (SDP) is done through CVX solver \cite{grant2009cvx}. In the numerical results, we perform Monte Carlo simulations, where in each iteration we compute the normalized hamming distance between $\bm{\sigma}$ and $\hat{\bm{\sigma}}$ as an estimate for the error probability. Our numerical experiments address the following questions:

\textbf{Q1: How does the error probability change with $a$ and $b$?} 
We first study community recovery on synthetic graphs (SBM) with $n = 100$ vertices, $r = 2$ communities, $b = 0.1$ and vary the parameter $a$. Fig. \ref{fig:impact_changing_parameters}(a) shows the impact of increasing $a$ on the error probability of (i) non-private recovery; (ii) SDP-stability mechanism and (iii) randomized-response SDP mechanism. For a fixed privacy budget $\epsilon$, we observe that when the difference between $a$ and $b$ increases, the error probabilities for all private mechanisms decrease but are no better than the non-private case. For a fixed $\epsilon$, the SDP-stability mechanism achieves a smaller error probability compared to RR+SDP mechanism, however, this comes at the expense of approximate edge DP guarantee. 

\textbf{Q2: What is the impact of $\epsilon$ on the error probability?} In Figs. \ref{fig:impact_changing_parameters} (b) and (c), we fix $n=200$, $a=3.5$, $b=0.1$ and study the impact of privacy budget $\epsilon$ on the error probability for the case of $r=2$ and $r=3$ communities. Specifically, for $r=2$, we observe that the SDP-stability mechanism (with $\delta = 10^{-5}$) outperforms RR+SDP; furthermore, as $\epsilon$ increases beyond a certain threshold, error probability for both converge to $0$. For $r=3$ communities, we can observe that the difference in performance between SDP-stability and RR+SDP is even more pronounced. In this setting, however, we do not expect the error probability to converge to $0$ even if $\epsilon\rightarrow \infty$ since the chosen values $(a=3.5, b=0.1)$ do not satisfy the exact recovery threshold $(\sqrt{a}-\sqrt{b}> \sqrt{r})$. 


\textbf{Q3: What is the impact of the problem size on the accuracy (SDP-Stability, RR+SDP, RR+Spectral)?}  In Fig. \ref{fig:impact_changing_parameters}  (d), we compare the performance of SDP relaxation based recovery versus spectral method proposed in \cite{hehir2021consistency}, both under randomized response for $a = 3.5$, $b = 0.1$ and $r = 2$ communities. We can observe that RR+SDP has less probability of error as a function of $n$ compared with the RR-Spectral method; however, RR+SDP has more computational complexity. In Fig. \ref{fig:impact_changing_parameters}(e), we show the error probability behavior as a function of $n$, the number of vertices for $r = 2$ communities and different privacy levels. From the figure, we observe that for the RR based approach, the privacy level should scale as $\Omega(\log(n))$ to achieve exact recovery, which is consistent with our theoretical findings. On the other hand, the stability based mechanisms can still provide exact recovery for finite $\epsilon$. We can draw similar conclusions for the case of $r = 3$ communities in Fig. \ref{fig:impact_changing_parameters}(f). 

\textbf{Q4: How do the \emph{private} community detection mechanisms perform on real-world datasets?}
We now discuss our results for two real-world datasets (shown in Figs. \ref{fig:impact_changing_parameters} (g) \& (h)): $(1)$ Zachary's Karate Club dataset which contains a social network of friendships between $34$ members of a karate club at a US university in the 1970s. \cite{girvan2002community} and $(2)$ The Political Blogosphere dataset \cite{adamic2005political} which consists of $1490$ political blogs captured during $2004$ US elections. Each blog is classified as left/liberal or right/conservative, i.e., $r=2$ and links between blogs were automatically extracted from a crawl of the front page of the blog. 
For the smaller size Karate club dataset ($n=34$), we 
observe from Fig. \ref{fig:impact_changing_parameters}(g), we can observe the impact of choosing $\delta$ on SDP-stability mechanism. Specifically, when $\delta= 34^{-2} \approx 10^{-3}$, then RR+SDP has lower error probability for smaller $\epsilon$ compared to SDP-stability. For the larger Political Blogosphere dataset ($n=1490$), the SDP-stability mechanism outperforms RR+SDP for all values of $\epsilon$ and $\delta = 1490^{-2} \approx 4 \times 10^{-7}$. We can observe that SDP-Stability performs better than RR+SDP for both datasets.

\textbf{Q5: How tight are the obtained bounds?} 
 We have plotted the phase transition behavior for both RR+SDP and SDP-stability mechanisms (see Fig. \ref{fig:phase_transition_mechanisms}). We observe that our theoretical bound (red line) is quite tight, and the threshold region obtained from empirical success probability is close to this bound.
 
 \begin{figure}[t]
\centering
	\centering
	{\includegraphics[width=0.5\columnwidth]{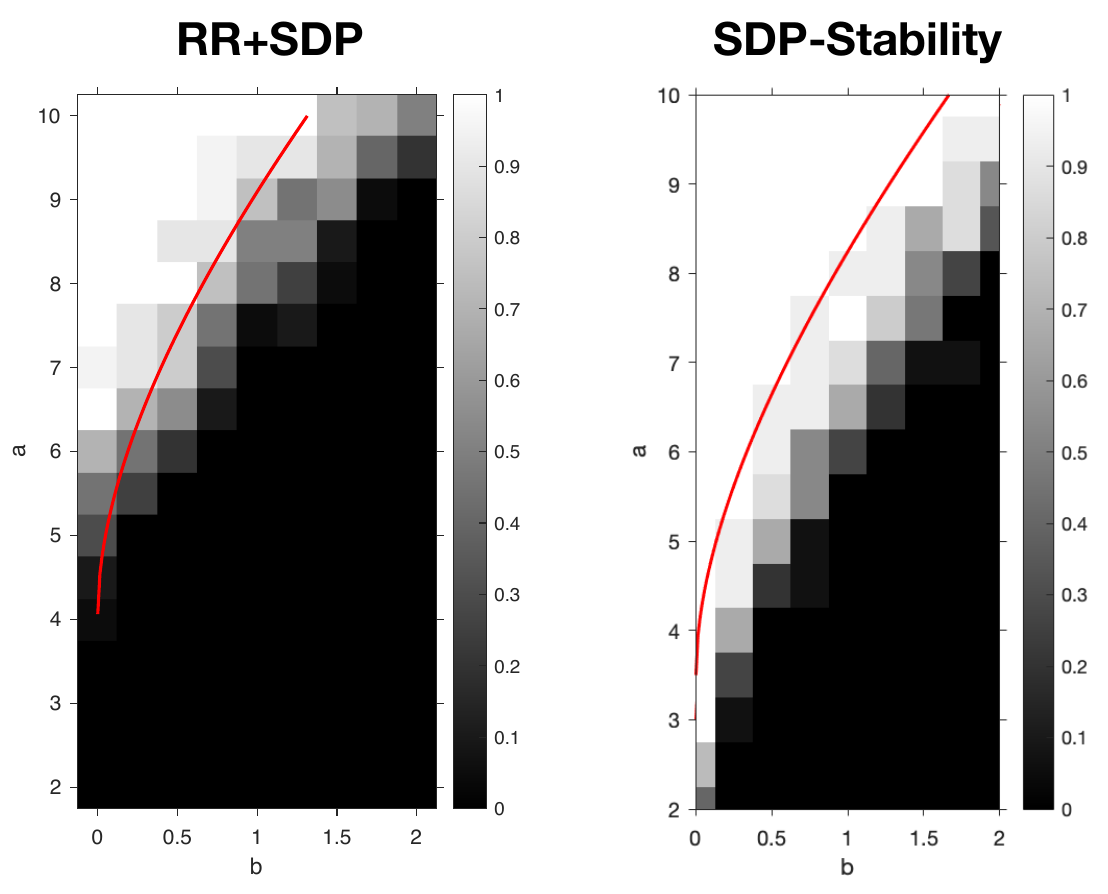}}	
	\caption{\small{Phase transition for Randomized response (RR)-SDP and SDP-Stability mechanisms: darker pixels represent lower empirical probability of success for $n = 50$ nodes, $\epsilon = 2$, $\delta = 4 \times 10^{-4}$, and $r = 2$ communities. The red lines represent the theoretical bounds presented in Theorems $3.3$ and $3.8$.}}
    \label{fig:phase_transition_mechanisms}
    \vspace{-0.2in}
\end{figure}



\vspace{-7pt}
\section{Conclusion}
In this paper, we studied the problem of community detection for SBMs subject to edge differential privacy. We presented and analyzed three classes of privacy-preserving mechanisms (stability, sampling, and randomized response) and studied the tradeoffs for exact recovery as a function of the connectivity $(a, b)$ and privacy parameters $(\epsilon, \delta)$. From our results, we deduce the following conclusions: the stability based mechanisms can achieve $(\epsilon, \delta)$-DP for any $\epsilon>0$ and require $\delta = n^{-t}$, $t > 0$. The sampling based mechanisms can instead achieve pure $(\epsilon, 0)$-DP; however, they require a larger separation between $(a, b)$ for exact recovery compared to the stability based methods. Among the three mechanisms, the randomized response mechanism, while least computationally complex, requires $\epsilon$ to scale as $\Omega(\log(n))$ for exact recovery. We also presented simulation results on both synthetic and real-world graphs to validate our theoretical findings.  There are several interesting open problems: a) obtaining converse results (necessary conditions) to assess the optimality (or gap to optimality) of the proposed mechanisms; b) generalization of the mechanisms and the associated analysis to degree-corrected SBMs; c) studying the impact of privacy on other recovery notions (such as weak recovery \cite{abbe2017community}); and d) design and analysis of efficient algorithms for stability based mechanisms.     





 \section*{Acknowledgements}
 
 We thank the anonymous ICML reviewers for their insightful suggestions. The work of M. Seif and R. Tandon was supported by NSF grants CAREER 1651492, CNS 1715947, CCF 2100013 and CNS 2209951. 
The work of D. Nguyen and A. Vullikanti was partially supported by NSF grants CCF-1918656, IIS-1931628, IIS-1955797, and NIH grant R01GM109718.

\bibliography{icml2022_ref}
\bibliographystyle{icml2022}

\newpage
\appendix
\onecolumn
\section{Appendix}


\input{supp_icml2022}


\end{document}

%% file: supp_icml2022.tex
In this Appendix, we provide the proofs of our results presented in Section 3. We provide auxiliary results that are used throughout the proofs at end of this document.  We summarize the notations and symbols used in Table \ref{table:symbols_table}.

\begin{table}[h]
\centering
\begin{tabular}{l l}
\hline
$n$ & Number of vertices\\
$\mathbf{A}$ & Adjacency matrix\\ 
$\tilde{\mathbf{A}}$ & Perturbed adjacency matrix \\ 
$\epsilon$ & Privacy budget\\ 
$\bm{\sigma}^{*}$ & Ground truth labels \\
$\hat{\bm{\sigma}}$ & Estimated labels \\
$\mathcal{L}$ & Set of all possible labels \\ 
$E_{\text{inter}}(G, \bm{\sigma})$ & set of cross-community edges of graph $G$ \\
$E(C_{1},C_{2})$ & Number of edges between $C_{1}$ and $C_{2}$ \\
$\mathbf{I}$ & Identity matrix \\ 
$\mathbf{J}$ & All ones matrix\\
$\mathbf{1}$ & All ones vector\\
\hline
\end{tabular}
\caption{List of Symbols.}
\label{table:symbols_table}
\end{table}

\input{02_Stability_Privacy}

\input{01_Stability_MLE}

\input{00_Stability_SDP}

 \section{Proof of Theorem 3.6  (Threshold condition for $\mathcal{M}_{\operatorname{Bayesian}}(G)$ for $r = 2$)} \label{appendix:proof_Bayesian_sampling}

We first prove that $\mathcal{M}_{\operatorname{Bayesian}}(G)$ satisfies $\epsilon$-edge DP for $\epsilon \geq \log(a/b)$. For a fixed graph $G$, w.l.o.g., let us assume $G = G' + e$, where $e$ is an edge. We define $E_{\text{intra}}(G, \bm{\sigma})$ as the set of same-community edges in graph $G$ and $E_{\text{inter}}(G, \bm{\sigma})$ is the set of cross-community edges of graph $G$, with respect to a labeling $\bm{\sigma}$. We analyze two cases: $(1)$ $e\in E_{\text{intra}}(G, \bm{\sigma})$ (1) and $e\in E_{\text{inter}}(G, \bm{\sigma})$. For each of them, we analyze the likelihood ratio of releasing a label vector $\bm{\sigma}$ if the input graph is $G$ or $G'$. We start with the first case: 
  
 - \textbf{Case $1$: $e\in E_{\text{intra}}(G, \bm{\sigma})$}
  \begin{align}
    R_{1} &= \frac{\operatorname{Pr}(\bm{\sigma}|G')}{\operatorname{Pr}(\bm{\sigma}|G)}
              = \frac{\operatorname{Pr}(G'|\bm{\mathbf{\sigma}}) }{\operatorname{Pr}(G | \bm{\sigma})} \times \frac{\operatorname{Pr}(G)}{\operatorname{Pr}(G')}  \nonumber \\ 
                &= \frac{1-p}{p} \times \frac{\operatorname{Pr}(G)}{\operatorname{Pr}(G')} \nonumber \\
                &= \frac{1-p}{p} \times \frac{\sum_{\bm{\sigma}'} \operatorname{Pr}(\bm{\sigma}')\operatorname{Pr}(G|\bm{\sigma}')}{\sum_{\bm{\sigma}'} \operatorname{Pr}(\bm{\sigma}')\operatorname{Pr}(G'|\bm{\sigma}')} \nonumber \\
                & \overset{(a)} =  \frac{1-p}{p} \times \frac{\sum_{\bm{\sigma}'}\operatorname{Pr}(\bm{\sigma}')\Pr(G'|\bm{\sigma}') \times \frac{p}{1-p}}{\sum_{\bm{\sigma}'}\Pr(\bm{\sigma}')\Pr(G'|\bm{\sigma}')} = 1,
  \end{align}
  where step (a) follows from the fact that $\operatorname{Pr}(G | \bm{\sigma}) = \frac{p}{1-p} \times \operatorname{Pr}(G'|\bm{\sigma})$. Note that, the distribution $p(G|\sigma)$ is given as 
\begin{align}
 & p(G|\sigma)  = \prod_{i < j} p(A_{i,j}|\sigma)  \overset{(a)} = \prod_{i < j} p(A_{i,j}|\sigma_{i}, \sigma_{j}) \nonumber \\ 
 & \overset{(b)} = \prod_{i < j} \left[ \frac{P(A_{i,j}) + Q(A_{i,j})}{2} + \frac{P(A_{i,j}) - Q(A_{i,j})}{2} \sigma_{i} \sigma_{j} \right], \nonumber 
\end{align}
where in step (a), the random variable $A_{i,j}$ only depends on the labels $\sigma_{i}$ and $\sigma_{j}$. In step (b), we have $P(A_{i,j}) = p^{A_{i,j}} (1-p)^{1-A_{i,j}}$, and $Q(A_{i,j}) = q^{A_{i,j}} (1-q)^{1-A_{i,j}}$, where $A_{i,j} \in \{0,1\}$. The distribution $p(\sigma) = (\frac{1}{r})^{n}$, while $p(G) = \sum_{\sigma'} p(G|\sigma')p(\sigma') = (\frac{1}{r})^{n} \sum_{\sigma'} \prod_{i < j} \left[ \frac{P(A_{i,j}) + Q(A_{i,j})}{2} + \frac{P(A_{i,j}) - Q(A_{i,j})}{2} \sigma_{i}' \sigma_{j}' \right]$.

  We next analyze the other ratio as follows.
  \begin{align}
 R_{2}  &= \frac{\operatorname{Pr}(\bm{\sigma}|G)}{\operatorname{Pr}(\bm{\sigma}|G')}
              = \frac{\operatorname{Pr}(G|\bm{\mathbf{\sigma}}) }{\operatorname{Pr}(G' | \bm{\sigma})} \times \frac{\operatorname{Pr}(G')}{\operatorname{Pr}(G)}  \nonumber \\ 
                         &= \frac{p}{1-p} \times \frac{\operatorname{Pr}(G')}{\operatorname{Pr}(G)} \nonumber \\
                &= \frac{p}{1-p} \times \frac{\sum_{\bm{\sigma}'} \operatorname{Pr}(\bm{\sigma}')\operatorname{Pr}(G'|\bm{\sigma}')}{\sum_{\bm{\sigma}'} \operatorname{Pr}(\bm{\sigma}')\operatorname{Pr}(G|\bm{\sigma}')} \nonumber \\
                     &= \frac{p}{1-p} \times \left[\frac{\sum_{\bm{\sigma}': e\in E_{\text{intra}}(\bm{\sigma}')} \operatorname{Pr}(\bm{\sigma}') \operatorname{Pr}(G'|\bm{\sigma}') + \sum_{\bm{\sigma}': e\in E_{\text{inter}}(L')}\operatorname{Pr}(\bm{\sigma}') \operatorname{Pr}(G'|\bm{\sigma}')}{\sum_{\bm{\sigma}'}\operatorname{Pr}(\bm{\sigma}') \operatorname{Pr}(G|\bm{\sigma}')}\right]  \nonumber \\
                     &= \frac{p}{1-p} \times \left[\frac{\sum_{\bm{\sigma}': e\in E_{\text{intra}}(\bm{\sigma}')} \operatorname{Pr}(\bm{\sigma}') \operatorname{Pr}(G|\bm{\sigma}') \times \frac{1-p}{p} + \sum_{\bm{\sigma}': e\in E_{\text{inter}}(\bm{\sigma}')}\operatorname{Pr}(\bm{\sigma}')\operatorname{Pr}(G|\bm{\sigma}')\times \frac{1-q}{q}}{\sum_{\bm{\sigma}'}\operatorname{Pr}(\bm{\sigma}') \operatorname{Pr}(G|\bm{\sigma}')}\right] \nonumber  \\
                     & \overset{(a)} \leq \frac{p}{1-p} \times \left[\frac{\sum_{\bm{\sigma}': e\in E_{\text{intra}}(\bm{\sigma}')} \operatorname{Pr}(\bm{\sigma}') \operatorname{Pr}(G|\bm{\sigma}')\times \frac{1-q}{q} + \sum_{\bm{\sigma}': e\in E_{\text{inter}}(\bm{\sigma}')}\operatorname{Pr}(\bm{\sigma}')\operatorname{Pr}(G|\bm{\sigma}') \times \frac{1-q}{q}}{\sum_{\bm{\sigma}'}\operatorname{Pr}(\bm{\sigma}') \operatorname{Pr}(G|\bm{\sigma}')}\right] \nonumber  \\
                     &= \frac{p(1-q)}{q(1-p)} = e^{\epsilon_0},
  \end{align}
  where step (a) follows that $\frac{1-q}{q} > \frac{1-p}{p}$ for $p > q$. Now, we analyze the second case as follows.
  
- \textbf{Case 2: $e\in E_{\text{inter}}(G, \bm{\sigma})$}
  
  \begin{align}
    R_{1} &= \frac{\operatorname{Pr}(\bm{\sigma}|G')}{\operatorname{Pr}(\bm{\sigma}|G)} \nonumber \\
                &=\frac{1-q}{q} \times \frac{\sum_{\bm{\sigma}'}\operatorname{Pr}(\bm{\sigma}') \operatorname{Pr}(G|\bm{\sigma}')}{\sum_{\bm{\sigma}'}\operatorname{Pr}(\bm{\sigma}')\operatorname{Pr}(G'|\bm{\sigma}')}\nonumber \\
                &\leq \frac{1-q}{q} \times \frac{\sum_{\bm{\sigma}'} \operatorname{Pr}(\bm{\sigma}')\operatorname{Pr}(G'|\bm{\sigma}') \times \frac{p}{1-p}}{\sum_{\bm{\sigma}'}\operatorname{Pr}(\bm{\sigma}')\operatorname{Pr}(G'|\bm{\sigma}')}\nonumber \\
                &=\frac{p(1-q)}{q(1-p)}.
  \end{align}
We next analyze the other ratio as follows:
\begin{align}
  R_{2} &= \frac{\operatorname{Pr}(\bm{\sigma}|G)}{\operatorname{Pr}(\bm{\sigma}|G')} \nonumber \\
                   &= \frac{q}{1-q} \times \frac{\sum_{\bm{\sigma}'}\operatorname{Pr}(\bm{\sigma}')\operatorname{Pr}(G'|\bm{\sigma}')}{\sum_{\bm{\sigma}'}\operatorname{Pr}(\bm{\sigma}']\operatorname{Pr}(G|\bm{\sigma}')} \nonumber \\
                   &= \frac{q}{1-q} \times \left[ \frac{\sum_{\bm{\sigma}': e\in E_{\text{intra}}(\bm{\sigma}')} \operatorname{Pr}(\bm{\sigma}')\operatorname{Pr}(G'|\bm{\sigma}') + \sum_{\bm{\sigma}': e\in E_{\text{inter}}(\bm{\sigma}')}\operatorname{Pr}(\bm{\sigma}') \operatorname{Pr}(G'|\bm{\sigma}')}{\sum_{\bm{\sigma}'}\operatorname{Pr}(\bm{\sigma}')\operatorname{Pr}(G|\bm{\sigma}')}\right] \nonumber \\
                   &= \frac{q}{1-q} \times \left[\frac{\sum_{\bm{\sigma}': e\in E_{\text{intra}}(\bm{\sigma}')}\operatorname{Pr}(\bm{\sigma}')\operatorname{Pr}(G|\bm{\sigma}') \times \frac{1-p}{p} + \sum_{\bm{\sigma}': e\in E_{\text{inter}}(\bm{\sigma}')}\operatorname{Pr}(\bm{\sigma}')\operatorname{Pr}(G|\bm{\sigma}') \times \frac{1-q}{q}}{\sum_{\bm{\sigma}'}\operatorname{Pr}(\bm{\sigma}')\operatorname{Pr}(G|\bm{\sigma}')}\right] \nonumber \\
                   &\leq \frac{q}{1-q} \times \left[\frac{\sum_{\bm{\sigma}': e\in E_{\text{intra}}(\bm{\sigma}')}\operatorname{Pr}(\bm{\sigma}')\operatorname{Pr}(G|\bm{\sigma}')\times \frac{1-q}{q} + \sum_{\bm{\sigma}': e\in E_{\text{inter}}(\bm{\sigma}')}\operatorname{Pr}(\bm{\sigma}')\operatorname{Pr}(G|\bm{\sigma}') \times \frac{1-q}{q}}{\sum_{\bm{\sigma}'}\operatorname{Pr}(\bm{\sigma}')\operatorname{Pr}(G|\bm{\sigma}')}\right] = 1.
\end{align}
From the above cases,  we conclude that the Bayesian sampling mechanism satisfies $\epsilon$-edge DP for all $\epsilon \geq  \log (\frac{p}{q}) + \log(\frac{1-q}{1-p}) \geq \log (\frac{p}{q}) =  \log (\frac{a}{b}) = \epsilon_{0}$.

We next analyze the error probability of the Bayesian mechanism. For a fixed graph $G$, our goal is to show that 
 \begin{align}
     \frac{\Pr(\hat{\bm{\sigma}}_{\text{Bayesian}}(G) \neq \bm{\sigma}^{*})}{\Pr(\bm{\sigma}^{*}|G)} & = \frac{\sum_{\bm{\sigma} \neq \bm{\sigma}^{*}} \Pr(\bm{\sigma} | G)}{\Pr(\bm{\sigma}^{*} | G)} \nonumber \\ 
     & =  \frac{\sum_{\bm{\sigma} \neq \bm{\sigma}^{*}} \Pr(G | \bm{\sigma} )}{\Pr(G | \bm{\sigma}^{*})} \leq o(1) 
 \end{align}
 which is equivalent to $\Pr(\hat{\bm{\sigma}}_{\text{Bayesian}}(G) \neq \bm{\sigma}^{*}) \leq o(1)$. Given the ground truth partitions $(R^{*}, B^{*})$, let us define the following variables:
\begin{align}
    m_{1} \triangleq E(R^{*}\backslash S_{1}, S_{1}), \nonumber \\ 
    m_{2} \triangleq E(B^{*}\backslash S_{2}, S_{1}), \nonumber \\ 
    m_{3} \triangleq E(B^{*}\backslash S_{2}, S_{2}), \nonumber \\ 
    m_{4} \triangleq E(R^{*} \backslash S_{1}, S_{2}),
\end{align}
where $S_{1} \subseteq R^{*}$, $S_{2} \subseteq B^{*}$ that represent the mis-classified nodes in both communities $(R,B)$. Given $S_{1}$, $S_{2}$ both of size $k$, $k \in [1, \frac{n}{2}]$, we have 
  \begin{align}
  R^{(k)} \triangleq \frac{\operatorname{Pr}(G|\bm{\sigma}, |S_{1}| = |S_{2}| = k )}{\operatorname{Pr}(G|\bm{\sigma^{*}}, |S_{1}| = |S_{2}| = k)}   &=  \underbrace{\bigg(\frac{q}{p}\bigg)^{m_1}\bigg(\frac{1-q}{1-p}\bigg)^{k(n-k)-m_1} \times \bigg(\frac{p}{q}\bigg)^{m_2} \bigg(\frac{1-p}{1-q}\bigg)^{k(n-k)-m_2}}_{=P^{(k)}_{1}} \label{term_1_coeff}\\
                          &\times \underbrace{\bigg(\frac{q}{p} \bigg)^{m_3} \bigg(\frac{1-q}{1-p} \bigg)^{k(n-k)-m_3} \times \bigg(\frac{p}{q} \bigg)^{m_4} \bigg(\frac{1-p}{1-q}\bigg)^{k(n-k)-m_4}}_{=P^{(k)}_{2}}. \label{term_2_coeff}
  \end{align}
  Note that, 
  \begin{align}
       \frac{\Pr(\hat{\bm{\sigma}}_{\text{Bayesian}}(G) \neq \bm{\sigma}^{*})}{\Pr(\bm{\sigma}^{*}|G)} & \leq  \sum_{k=1}^{\frac{n}{2}} {n \choose k}^{2} R^{(k)}.
  \end{align}
In order to bound the above ratio, we will first bound $R^{(k)}$ individually. To this end, let us define $\alpha = \frac{p}{q}, \beta = \frac{1-q}{1-p}$. We now simplify each term individually (i.e., \eqref{term_1_coeff} \& \eqref{term_2_coeff}) as follows:
  \begin{align}
    P_{1}^{(k)} &= \alpha^{-m_1}\times\beta^{k(n-k)-m_1}\times\alpha^{m_2}\times\beta^{-(k(n-k)-m_2)} \nonumber \\
                  &=\alpha^{m_2-m_1} \times \beta^{m_2-m_1} \nonumber \\
                  &={\left(\frac{p}{q} \times \frac{1-q}{1-p} \right)}^{m_2-m_1} \nonumber \\ 
                  & ={\left(\frac{q}{p} \times \frac{1-p}{1-q} \right)}^{m_1-m_2} \nonumber \\
                  &\overset{(a)} =(1- x)^{m_1 - m_2}, \nonumber \\
                  & \overset{(b)} \leq \exp\left[-x(m_1-m_2) \right],
  \end{align}
  where $x \triangleq 1 - \frac{q(1-p)}{p(1-q)}$ where $x \leq 1$. Step (b) follows that $(1-x) \leq e^{-x}, \forall x > 0$. Similarly, we have 
  \begin{align}
      P_{2}^{(k)} & \leq \exp\left[-x(m_3-m_4)\right].
  \end{align}
  Therefore, we have 
  \begin{align}
      R^{(k)} & \leq \exp \left[ -x (m_{1}+m_{3} - (m_{2}+m_{4})) \right] \nonumber \\ 
      &= \exp \left[-x (\tilde{m}_{1}^{(k)} - \tilde{m}_{2}^{(k)}) \right], 
  \end{align}
  where  $\tilde{m}_{1}^{(k)} \sim \operatorname{Bin}(2k(n-k), p)$, and $\tilde{m}_{2}^{(k)} \sim \operatorname{Bin}(2k(n-k), q)$.  For all $k \in [1:\frac{n}{2}]$, we have that w.h.p.
   \begin{align}
       \frac{\Pr(\hat{\bm{\sigma}}_{\text{Bayesian}}(G) \neq \bm{\sigma}^{*})}{\Pr(\bm{\sigma}^{*}|G)}  & = \sum_{k=1}^{\frac{n}{2}} {n \choose k}^{2} R^{(k)} \nonumber \\ 
    & \leq \sum_{k=1}^{\frac{n}{2}} \bigg( \frac{ne}{k}\bigg)^{2k} R^{(k)} \nonumber \\
    & = \sum_{k=1}^{\frac{n}{2}} \exp \left[ 2k (\log(n) - \log(k) + 1 ) - x (\tilde{m}_{1}^{(k)} - \tilde{m}_{2}^{(k)})\right] \nonumber \\ 
    &  = \sum_{k=1}^{\frac{n}{2}} \exp \left[ 2k \bigg(\log(n) - \log(k) + 1  - \frac{x}{2k} (\tilde{m}_{1}^{(k)} - \tilde{m}_{2}^{(k)}) \bigg) \right] \nonumber \\ 
       &  = \sum_{k=1}^{\frac{n}{2}} \exp \left[ - 2k \bigg(-\log(n) + \log(k) - 1  + \frac{x}{2k} (\tilde{m}_{1}^{(k)} - \tilde{m}_{2}^{(k)})\bigg) \right] \nonumber \\
       & = \sum_{k=1}^{\frac{n}{2}} \exp \left[ - 2k \bigg(\log(k) - 1 \bigg) \right] \times \exp \left[ - 2k \bigg(-\log(n)   + \frac{x}{2k} (\tilde{m}_{1}^{(k)} - \tilde{m}_{2}^{(k)})\bigg) \right]  \nonumber \\
       & \leq  \sum_{k=1}^{\frac{n}{2}} \exp \left[ - 2k \bigg(\log(k) - 1 \bigg) \right] \\ & \times \exp \left[ - 2k \bigg(-\log(n)   + \frac{x}{2k} \times (1-\tilde{\delta}) \times 2k (n-k) (a-b) \frac{\log(n)}{n}\right]  \nonumber \\
        & \leq  \sum_{k=1}^{\frac{n}{2}} \exp \left[ - 2k \bigg(\log(k) - 1 \bigg) \right] \times \exp \left[ - 2k \bigg(-\log(n)   + (1-\tilde{\delta})  \times  \frac{x}{2} \times  (a-b) \log(n) \bigg) \right]  \nonumber \\
                & \overset{(a)} =  \sum_{k=1}^{\frac{n}{2}} \exp \left[ - 2k \bigg(\log(k) - 1 \bigg) \right] \times \exp \left[ - 2k \bigg( \frac{x}{2} \times (1-\tilde{\delta}) \times  (a-b) -1  \bigg)\log(n) \right] = o(1),
  \end{align}
  where in step (a), we have 
 \begin{align}
\frac{x}{2} (1-\tilde{\delta}) (a-b) > 1 \Rightarrow \sqrt{a} - \sqrt{b} > \frac{\sqrt{2}}{\sqrt{x (1-\tilde{\delta})}}.     
 \end{align}
  
 To this end, the error probability of the Bayesian mechanism is 
 \begin{align}
     \Pr(\hat{\bm{\sigma}}_{\text{Bayesian}}(G) \neq \bm{\sigma}^{*}) & = \Pr(\hat{\bm{\sigma}}_{\text{Bayesian}}(G) \neq \bm{\sigma}^{*} | E_{S_{1}, S_{2}}) \times \Pr(E_{S_{1}, S_{2}}) \nonumber \\
   &     \hspace{0.2in} + \Pr(\hat{\bm{\sigma}}_{\text{Bayesian}}(G) \neq \bm{\sigma}^{*} | E_{S_{1}, S_{2}}^{c}) \times \Pr(E_{S_{1}, S_{2}}^{c}) \nonumber \\
     & \leq o(1) \times 1 + 1 \times n^{- \delta'},
 \end{align}
 where $E_{S_{1},S_{2}} \triangleq  \{ \tilde{m}_{1} - \tilde{m}_{2} \geq c' \log(n) \}$ and $E_{S_{1},S_{2}}^{c}$ denotes its complement, and $x = 1 - \frac{b (n - a \log(n))}{a (n - b \log(n) )} > 1 - \frac{b}{a}$. In order to make the error probability behave as $o(1)$, we have the following condition on $a$ and $b$:
 \begin{align}
     \sqrt{a} - \sqrt{b} & > \sqrt{2} \times \max \left[ \frac{\sqrt{2}}{{\tilde{\delta}}}, \frac{1}{\sqrt{x (1-\tilde{\delta})}} \right]. \label{eqn:lower_bound_Bayesian} 
 \end{align}
 Note that we showed the Bayesian mechanism satisfies $\epsilon \geq \epsilon_{0} = \log \big(\frac{a}{b} \big)$, therefore, we have 
 \begin{align}
     \frac{b}{a} = e^{-\epsilon_{0}}
 \end{align}
 We pick $\tilde{\delta}$ as $\tilde{\delta} = (\sqrt{2} - 1) \epsilon$ so that the lower bound in \eqref{eqn:lower_bound_Bayesian} is minimized when the two arguments are equal. Plugging the value of $\tilde{\delta}$ into \eqref{eqn:lower_bound_Bayesian} yields the following
 \begin{align}
     \sqrt{a} - \sqrt{b} > \frac{2}{(1-e^{-\epsilon_{0}})(\sqrt{2}-1)}.
 \end{align}
 This completes the proof of Theorem 3.6.

\section{Proof of Theorem 3.7 (Threshold condition for $\mathcal{M}_{\operatorname{Expo.}}(G)$ for $r = 2$)}\label{appendix:proof_exponential}

The privacy analysis of  $\mathcal{M}_{\operatorname{Expo.}}(G)$ is straightforward and follows on similar lines as in \cite{dwork2006calibrating}. 

We next analyze the error probability of $\mathcal{M}_{\operatorname{Expo.}}(G)$. The error probability analysis follows on similar lines as the Bayesian mechanism. 
  \begin{align}
    R^{(k)}  &= \frac{\exp(-\epsilon\times E_{\text{inter}}(G, \bm{\sigma}))}{\exp(-\epsilon\times E_{\text{inter}}(G, \bm{\sigma}^{*}))} \nonumber \\
                          &= \frac{\exp(-\epsilon\times (E_{\text{inter}}(G, \bm{\sigma}^{*}) + m_1 + m_3 - m_2 - m_4))}{\exp(-\epsilon\times E_{\text{inter}}(G, \bm{\sigma}^{*}))} \nonumber \\
                          &= \frac{\exp(-\epsilon\times E_{\text{inter}}(G, \bm{\sigma}^{*}))\times\exp(-\epsilon(m_1 + m_3 - m_2 - m_4))}{\exp(-\epsilon\times E_{\text{inter}}(G, \bm{\sigma}^{*}))} \nonumber \\
                          &= \exp(-\epsilon(m_1 + m_3 - m_2 - m_4)) \nonumber \\
                          & = \exp(-\epsilon (\tilde{m}_{1}^{(k)} - \tilde{m}_{2}^{(k)})),
  \end{align}
    where  $\tilde{m}_{1}^{(k)} \sim \operatorname{Bin}(2k(n-k), p)$, and $\tilde{m}_{2}^{(k)} \sim \operatorname{Bin}(2k(n-k), q)$. Now, we have 
  \begin{align}
         \frac{\Pr(\hat{\bm{\sigma}}_{\text{Expo.}}(G) \neq \bm{\sigma}^{*})}{\Pr(\bm{\sigma}^{*}|G)}  \leq  \sum_{k=1}^{n/2} {n \choose k}^{2}  R^{(k)}.
  \end{align}
Similarly, in order to make the error probability behaves as $o(1)$, we have 
 \begin{align}
     \sqrt{a} - \sqrt{b} & > \sqrt{2} \times \max \left[ \frac{\sqrt{2}}{{\tilde{\delta}}}, \frac{1}{\sqrt{ \epsilon (1-\tilde{\delta})}} \right].
 \end{align}
We pick $\tilde{\delta} = (\sqrt{2}-1) \epsilon$ and this yields 
\begin{align}
       \sqrt{a} - \sqrt{b} > \frac{2}{(\sqrt{2}-1) \epsilon}.
\end{align}
 This completes the proof of Theorem 3.7.

\section{Proof of Theorem 3.8 (Threshold condition for $\mathcal{M}_{\operatorname{RR}}(G)$ for $r = 2$)} \label{proof_for_SDP}


\subsection*{Error probability analysis of SDP recovery algorithm:}

For the ease of exposition, let us consider a graph $G$ with $2n$ vertices. The Lagrangian function is written as follows: 
\begin{align}
    \mathcal{L}(\tilde{\mathbf{A}}, \mathbf{Y}, \mathbf{S}, \mathbf{D}, \lambda) & = \operatorname{tr}(\tilde{\mathbf{A}} \mathbf{Y}) + \operatorname{tr}(\mathbf{S} \mathbf{Y}) - \operatorname{tr}(\mathbf{D} (\mathbf{Y}-\mathbf{I})) - \lambda \operatorname{tr}(\mathbf{J} \mathbf{Y}),  
\end{align}
where $\mathbf{S} \succcurlyeq \mathbf{0}$, $\mathbf{D} = \operatorname{diag}(d_{i})$ and $\lambda \in \mathds{R}$. Then, 
\begin{align}
    \nabla_{\mathbf{Y}} \mathcal{L} = \tilde{\mathbf{A}} + \mathbf{S} - \mathbf{D} - \lambda \mathbf{J} = \mathbf{0}. 
\end{align}
In order to satisfy the first order stationery condition, we have 
\begin{align}
    \mathbf{S}^{*} = \mathbf{D}^{*} - \tilde{\mathbf{A}} + \lambda^{*} \mathbf{J}.  
\end{align}
From the KKT conditions, we have the following: 
\begin{align}
     & \operatorname{tr}(\mathbf{D}, \mathbf{Y}-\mathbf{I}) = 0 \Rightarrow \mathbf{D}^{*} \mathbf{Y}^{*} = \mathbf{D}^{*} \mathbf{I}, \\
     & \lambda^{*} \operatorname{tr}(\mathbf{J} \mathbf{Y}^{*}) = 0, \\
     & \operatorname{tr}(\mathbf{S}^{*} \mathbf{Y}^{*}) = 0 \Rightarrow \mathbf{S} {\bm{\sigma}}^{*} = \mathbf{0}, \label{eqn:null_space_of_S}
\end{align}
where in eqn. \eqref{eqn:null_space_of_S}, ${\bm{\sigma}}^{*}$ is the null space of $\mathbf{S}$. In order to ensure that $\mathbf{Y}^{*}$ is the unique solution, we require that $\lambda_{2}(\mathbf{S}^{*}) > 0$, i.e., the second smallest eigenvalues of $\mathbf{S}$. This comes from the rank-nullity Theorem, i.e., 
\begin{align}
    \text{rank}(\mathbf{S}^{*}) + \text{Null}(\mathbf{S}^{*}) = 2n \Rightarrow \text{rank}(\mathbf{S}^{*}) = 2n -1. 
\end{align}
To this end, we have the following: 
\begin{align}
     \operatorname{tr}(\tilde{\mathbf{A}} \mathbf{Y}) \leq \mathcal{L}(\tilde{\mathbf{A}}, \mathbf{Y}^{*}, \mathbf{S}^{*}, \mathbf{D}^{*}, \lambda^{*}) & =  \operatorname{tr}(\mathbf{S}^{*} - \mathbf{D}^{*} + \tilde{\mathbf{A}} - \lambda^{*} \mathbf{J} Y) +  \operatorname{tr}(\mathbf{D}^{*} \mathbf{I}) \overset{(a)} =  \operatorname{tr}(\mathbf{D}^{*} \mathbf{I}) \nonumber \\
     & = \operatorname{tr}(\mathbf{D}^{*} \mathbf{Y}^{*}) \nonumber \\ 
     & = \operatorname{tr}(\mathbf{S}^{*} + \tilde{\mathbf{A}} - \lambda^{*} \mathbf{J}, \mathbf{Y}^{*}) = \operatorname{tr}(\tilde{\mathbf{A}} \mathbf{Y}^{*}). \label{eqn:lagrangian_function_complete}
\end{align}
Now, our goal is to prove that w.h.p. $\mathbf{S}^{*} \succcurlyeq 0$ with $\lambda_{2}(\mathbf{S}^{*}) > 0$. More specifically, we want to show that 
\begin{align}
    \operatorname{Pr} \left[\inf_{\mathbf{x}: \|\mathbf{x}\| = 1, \mathbf{x} \perp \bm{\sigma}^{*}} \mathbf{x}^{T} \mathbf{S}^{*} \mathbf{x} > 0 \right] \geq 1 - o(1).
\end{align}
Alternatively, 
\begin{align}
    \operatorname{Pr} \left[\inf_{\mathbf{x}: \|\mathbf{x}\| = 1, \mathbf{x} \perp \bm{\sigma}^{*}} \mathbf{x}^{T} \mathbf{S}^{*} \mathbf{x} \leq  0 \right] \leq o(1).
\end{align}

Before we proceed, we note that 
\begin{align}
    \mathds{E}[\tilde{\mathbf{A}}] & = \frac{\tilde{p}- \tilde{q}}{2} \mathbf{Y}^{*} + \frac{\tilde{p}+\tilde{q}}{2} \mathbf{J} -  \tilde{p} \mathbf{I}. 
\end{align}
Now, for any $\mathbf{x}$ such that $\|\mathbf{x}\| = 1, \mathbf{x} \perp \bm{\sigma}^{*}$ (i.e., $\bm{x}^{T} \bm{\sigma}^{*}$ = 0), and $\lambda^{*} \leq (\tilde{p}+\tilde{q})/2$,
\begin{align}
    \mathbf{x}^{T} \mathbf{S}^{*} \mathbf{x} & = \mathbf{x}^{T} \mathbf{D}^{*} \mathbf{x} -  \mathbf{x}^{T} \mathds{E}[\tilde{\mathbf{A}}] \mathbf{x} + \lambda^{*} \mathbf{x}^{T} \mathbf{J} \mathbf{x} - \mathbf{x}^{T} (\tilde{\mathbf{A}} - \mathds{E}[\tilde{\mathbf{A}}]) \mathbf{x} \nonumber \\ 
     & = \mathbf{x}^{T} \mathbf{D}^{*} \mathbf{x} - \frac{\tilde{p}-\tilde{q}}{2} \mathbf{x}^{T} \mathbf{Y}^{*} \mathbf{x} + \bigg(\lambda^{*} - \frac{\tilde{p}+\tilde{q}}{2}\bigg) \mathbf{x}^{T} \mathbf{J} \mathbf{x} + \tilde{p} -  \mathbf{x}^{T} (\tilde{\mathbf{A}} - \mathds{E}[\tilde{\mathbf{A}}]) \mathbf{x}  \nonumber \\
      & \leq  \mathbf{x}^{T} \mathbf{D}^{*} \mathbf{x} - \frac{\tilde{p}-\tilde{q}}{2} \mathbf{x}^{T} \mathbf{Y}^{*} \mathbf{x} + \tilde{p} -   \|\tilde{\mathbf{A}} - \mathds{E}[\tilde{\mathbf{A}}] \| \nonumber \\
       & =  \mathbf{x}^{T} \mathbf{D}^{*} \mathbf{x} - \frac{\tilde{p}-\tilde{q}}{2} \mathbf{x}^{T} \bm{\sigma}^{*} {\bm{\sigma}^{*}}^{T}  \mathbf{x} + \tilde{p} -   \|\tilde{\mathbf{A}} - \mathds{E}[\tilde{\mathbf{A}}] \| \nonumber \\
          & =  \mathbf{x}^{T} \mathbf{D}^{*} \mathbf{x}  + \tilde{p} -   \|\tilde{\mathbf{A}} - \mathds{E}[\tilde{\mathbf{A}}] \| \nonumber \\
          & \leq \mathbf{x}^{T} \mathbf{D}^{*} \mathbf{x} + \tilde{p} = \sum_{i \in [2n]} d_{i}^{*} + \tilde{p}.
\end{align}
Note that from \eqref{eqn:lagrangian_function_complete}, we have $\operatorname{tr}(\mathbf{D}^{*} \mathbf{I}) = \operatorname{tr}(\tilde{\mathbf{A}} \mathbf{Y}^{*})$. Therefore, 
\begin{align}
    d_{i}^{*} = \sum_{j=1}^{2n} \tilde{A}_{i,j} \sigma_{i}^{*} \sigma_{j}^{*}.
\end{align}
Also, note that each $d_{i}^{*}$ is equal in distribution to $X - Y$, where $X \sim \text{Bin}\big(n-1 , a_{n} \frac{\log(n)}{n}\big)$ and $Y \sim \text{Bin} \big(n, b_{n} \frac{\log(n)}{n}\big)$. Applying the union bound, our goal now is to derive conditions on $a_{n}$ and $b_{n}$ such that 
\begin{align}
     \sum_{i = 1}^{2n} \operatorname{Pr}(d_{i}^{*} \leq  0 ) \leq n \times \operatorname{Pr}(Y \geq X) = o(1). 
\end{align}
Similar as before, the above probability will be of order $o(1)$ holds if 
\begin{align}
    & \sqrt{a} - \sqrt{b} - \sqrt{\frac{1}{e^{\epsilon} - 1} } > \sqrt{2} \times \sqrt{\frac{e^{\epsilon} + 1}{e^{\epsilon} - 1}} \nonumber\\
    \Rightarrow & \sqrt{a} - \sqrt{b} > \sqrt{2} \times \sqrt{\frac{e^{\epsilon} + 1}{e^{\epsilon} - 1}} + \frac{1}{\sqrt{e^{\epsilon} - 1}}.
\end{align}
To this end, we conclude that 
\begin{align}
    \operatorname{Pr}(\hat{\bm{\sigma}} \neq \bm{\sigma}) \leq \operatorname{Pr}(\hat{\mathbf{Y}} \neq \mathbf{Y}) = o(1).
\end{align}

We are ready now to analyze the error probability $\operatorname{Pr} (\hat{\bm{\sigma}} \neq \bm{\sigma})$.  Denote $Y \sim \operatorname{Bin}(n, q_{n})$ and  $X \sim \operatorname{Bin}(n, p_{n})$. Using union bound, we have that 
\begin{align}
    \operatorname{Pr} (\hat{\bm{\sigma}} \neq \bm{\sigma}) \leq 2 n \times \operatorname{Pr} (Y \geq X). 
\end{align}
We next expand $\operatorname{Pr} (Y \geq X)$ using law of total probability Theorem. We first define $Z = X + Y$ and $c_{n} = a_{n} + b_{n}$, then we have 
\begin{align}
     \operatorname{Pr}(Y \geq X) & = \sum_{ k = 0}^{2 n } \operatorname{Pr}(Y \geq X | Z = k) \operatorname{Pr}(Z = k) \nonumber \\
    & \leq \sum_{k = 0}^{10 c_{n} \log(n)} \operatorname{Pr}(Y \geq X | Z = k) \operatorname{Pr}(Z = k) + \operatorname{Pr}(Z \geq 10 c_{n} \log(n)) \nonumber \\
        & = \operatorname{Pr}(Y \geq X | Z = 0) \operatorname{Pr}(Z = 0)  + \sum_{k = 1}^{10 c_{n} \log(n)} \operatorname{Pr}(Y \geq X | Z = k) \operatorname{Pr}(Z = k) + \operatorname{Pr}(Z \geq 10 c_{n} \log(n)) \nonumber \\ 
            & = \operatorname{Pr}(Y = X = 0 ) + \sum_{k = 1}^{10 c_{n} \log(n)} \operatorname{Pr}(Y \geq X | Z = k) \operatorname{Pr}(Z = k) + \operatorname{Pr}(Z \geq 10 c_{n} \log(n)) \nonumber \\ 
                & = \operatorname{Pr}(Y = 0) \times \operatorname{Pr}(X = 0)  + \sum_{k = 1}^{10 c_{n} \log(n)} \operatorname{Pr}(Y \geq X | Z = k) \operatorname{Pr}(Z = k) + \operatorname{Pr}(Z \geq 10 c_{n} \log(n)) \nonumber \\
                     & \overset{(a)} \leq n^{-c_{n}} + \sum_{k = 1}^{10 c_{n} \log(n)} \operatorname{Pr}(Y \geq X | Z = k) \operatorname{Pr}(Z = k) +  n^{-10 c_{n}} \nonumber \\ 
        & \leq 2 n^{- c_{n}} + \sum_{k=1}^{10 c_{n} \log(n)} \operatorname{Pr}(Y \geq X | Z = k) \operatorname{Pr}(Z = k),
\end{align}
where in step (a), we have that $\operatorname{Pr}(X = 0) = (1 - p_{n})^{n} \leq e^{ - n p_{n}}  = n^{-a_{n}}$. Similarly, $\operatorname{Pr}(Y = 0) \leq n^{-b_{n}}$. Also, we have  $\operatorname{Pr}(Z \geq 10 c_{n} \log(n)) \leq n^{-10 c_{n}}$ using Bernstein's inequality. We next upper bound $\operatorname{Pr}(Z = k)$ as follows: 
\begin{align}
\operatorname{Pr}(Z = k) & = \sum_{i = 0}^{k} \operatorname{Pr}(Y = i) \times \operatorname{Pr}(X = k -i),   
\end{align}
where, 
\begin{align}
    \operatorname{Pr}(Y = i) & = {n \choose i} \times  p_{n}^{i} \times (1 - p_{n})^{n - i} \nonumber \\ 
    & = \frac{n!}{i! (n-i)!}  \times p_{n}^{i} \times (1 - p_{n})^{n - i} \nonumber \\
    & = \frac{n!}{i! (n-i)!} \times \frac{(b_{n} \log(n))^{i}}{n^{i}} \times  \bigg(1 - \frac{b_{n} \log(n)}{n} \bigg)^{n - i} \nonumber \\ 
    & = \frac{(b_{n} \log(n))^{i}}{i!} \times \frac{n!}{n^{i} \times (n-i)!} \times e^{ - \frac{b_{n} \log(n)}{n} \times (n-i)} \nonumber \\
     & = \frac{(b_{n} \log(n))^{i}}{i!} \times \frac{n!}{n^{i} \times (n-i)!} \times e^{ - b_{n} \log(n) \times (1-i/n)} \nonumber \\
     & = \frac{(b_{n} \log(n))^{i}}{i!} \times \frac{n!}{n^{i} \times (n-i)!} \times n^{ -b_{n}} \times n^{b_{n} \times i/n} \nonumber \\
       & = \frac{(b_{n} \log(n))^{i}}{i!} \times \frac{n \times (n-1) \times \cdots \times (n - i + 1)}{n^{i} } \times n^{ -b_{n}} \times n^{b_{n} \times i/n} \nonumber \\
        & \overset{(a)} \leq  \frac{(b_{n} \log(n))^{i}}{i!} \times 1 \times  n^{ -b_{n}} \times n^{b_{n} \times i/n},
\end{align}
where in step (a) follows that $\prod_{j = 0}^{i-1} \big(1 - \frac{j}{n}\big) \leq 1 $. Similarly, we upper bound $\operatorname{Pr}(X = k - i)$ as 
\begin{align}
    \operatorname{Pr}(X = k - i) \leq \frac{(a_{n} \log(n))^{k-i}}{(k-i)!} \times n^{ -a_{n}} \times n^{a_{n} \times (k-i)/n}.
\end{align}
To this end, we get 
\begin{align}
    \operatorname{Pr}(Z = k) & \overset{(a)} \leq  n^{- c_{n}} \times n^{a_{n} k/n}  \sum_{i = 0}^{k} n^{- (a_{n} - b_{n}) \times i/n} \times \frac{(b_{n} \log(n))^{i}}{i!} \times \frac{(a_{n} \log(n))^{k-i}}{(k-i)!} \nonumber \\
     & \leq  n^{- c_{n}} \times n^{a_{n} k/n}  \sum_{i = 0}^{k} \frac{(b_{n} \log(n))^{i}}{i!} \times \frac{(a_{n} \log(n))^{k-i}}{(k-i)!} \nonumber \\
      & \leq  n^{- c_{n}} \times n^{a_{n} k/n} \times \frac{1}{k!} \sum_{i = 0}^{k} k! \times \frac{(b_{n} \log(n))^{i}}{i!} \times \frac{(a_{n} \log(n))^{k-i}}{(k-i)!} \nonumber \\
      & \leq  n^{- c_{n}} \times n^{a_{n} k/n} \times \frac{1}{k!} \sum_{i = 0}^{k} {k \choose i} \times (b_{n} \log(n))^{i} \times (a_{n} \log(n))^{k-i} \nonumber \\
            & =  n^{- c_{n}} \times n^{a_{n} k/n} \times \frac{(c_{n} \log(n))^{k}}{k!},
\end{align}
where in step (a),  we have that $a_{n} \geq b_{n}$ and $n^{- (a_{n} - b_{n}) \times i/n} \leq 1$. Similarly, we upper bound $\operatorname{Pr}(Y \geq X | Z = k)$ as follows: 
\begin{align}
\operatorname{Pr}(Y \geq X | Z = k) & = \sum_{i = \frac{k}{2}}^{k} \operatorname{Pr}(Y = i) \times \operatorname{Pr}(X = k -i) \nonumber \\ 
& \leq  n^{- c_{n}} \times n^{a_{n} k/n} \times \frac{1}{k!} \sum_{i = \frac{k}{2}}^{k} {k \choose i} \times (b_{n} \log(n))^{i} \times (a_{n} \log(n))^{k-i} \nonumber \\
& =  n^{- c_{n}} \times n^{a_{n} k/n} \times \frac{(c_{n} \log(n))^{k}}{k!} \sum_{i = \frac{k}{2}}^{k} {k \choose i} \times \eta_{n}^{i} \times (1 - \eta_{n})^{k-i},
\end{align}
where $\eta_{n} = \frac{b_{n}}{a_{n} + b_{n}} < 1/2 $. For a fixed $k$ where $k \leq 10 c_{n} \log(n)$, we have 
\begin{align}
      \operatorname{Pr}(Z = k)  \times \operatorname{Pr}(Y \geq X | Z = k) & \leq n^{- 2 c_{n}} \times n^{2 a_{n} k/n} \times \frac{(c_{n} \log(n))^{2k}}{(k!)^{2}} \times \operatorname{Pr}(\operatorname{Bin}(k, \eta_{n}) \geq \frac{k}{2}) \nonumber \\ 
      & \overset{(a)} \leq  \frac{1 - \eta_{n}}{1 - 2 \eta_{n}} \times  n^{- 2 c_{n}} \times   \frac{(c_{n} n^{a_{n}/n} \log(n))^{2k}}{(k!)^{2}} \times \operatorname{Pr}(\operatorname{Bin}(k, \eta_{n}) = \frac{k}{2}) \nonumber \\
      & =  \frac{1 - \eta_{n}}{1 - 2 \eta_{n}} \times  n^{- 2 c_{n}} \times   \frac{(c_{n} n^{a_{n}/n} \log(n))^{2k}}{(k!)^{2}} \times \frac{k!}{\big(\frac{k}{2}!\big)^{2}} \times \theta_{n}^{k}  \nonumber\\
      & =  \frac{1 - \eta_{n}}{1 - 2 \eta_{n}} \times  n^{- 2 c_{n}} \times   \frac{(c_{n} n^{a_{n}/n} \log(n))^{2k}}{k!} \times \frac{1}{\big(\frac{k}{2}!\big)^{2}} \times \theta_{n}^{k}  \nonumber\\
            & =  \frac{1 - \eta_{n}}{1 - 2 \eta_{n}} \times  n^{- 2 c_{n}} \times   \frac{(c_{n}\log(n))^{k}}{k!} \times (c_{n} n^{2a_{n}/n} \log(n))^{k} \times \frac{1}{\big(\frac{k}{2}!\big)^{2}} \times \theta_{n}^{k}  \nonumber\\
             & \overset{(b)} \leq  \frac{1}{\sqrt{2\pi}} \times  \frac{1 - \eta_{n}}{1 - 2 \eta_{n}} \times  n^{- 2 c_{n}} \times   \frac{(c_{n} e \log(n))^{k}}{k^{k+1/2}} \times (c_{n} n^{2a_{n}/n} \log(n))^{k} \times \frac{1}{\big(\frac{k}{2}!\big)^{2}} \times \theta_{n}^{k}  \nonumber\\
                &  \overset{(c)} \leq  \frac{1}{\sqrt{2\pi}} \times  \frac{1 - \eta_{n}}{1 - 2 \eta_{n}} \times  n^{- 2 c_{n}} \times n^{c_{n}} \times (c_{n} n^{2a_{n}/n} \log(n))^{k} \times \frac{1}{\big(\frac{k}{2}!\big)^{2}} \times \frac{\theta_{n}^{k}}{\sqrt{k}}  \nonumber\\
                & = \frac{1}{\sqrt{2\pi}} \times  \frac{1 - \eta_{n}}{1 - 2 \eta_{n}} \times  n^{-  c_{n}} \times  \frac{(\theta_{n} c_{n} n^{2a_{n}/n} \log(n))^{k}}{\big(\frac{k}{2}!\big)^{2} \sqrt{k}} \nonumber \\ 
      & \overset{(d)} \leq \frac{2}{(2 \pi^{2})^{3/2}} \times \frac{1 - \eta_{n}}{1 - 2 \eta_{n}} \times  n^{-  c_{n}}  \times  \frac{(c_{n}  n^{2a_{n}/n} e \log(n))^{k}}{k^{k+1}} \times  \frac{(2 \theta_{n})^{k}}{\sqrt{k}} \nonumber \\ 
      & =  \frac{1}{\sqrt{2} \pi^{3}} \times \frac{1 - \eta_{n}}{1 - 2 \eta_{n}} \times  n^{- c_{n}}  \times  \frac{(2 \theta_{n} c_{n} n^{2a_{n}/n} e \log(n))^{k}}{k^{k+3/2}} \nonumber \\ 
       & =  \frac{1}{\sqrt{2} \pi^{3}} \times n^{2a_{n}k/n}  \times \frac{1 - \eta_{n}}{1 - 2 \eta_{n}} \times  n^{- c_{n}}  \times  \frac{(2 \theta_{n} c_{n} e \log(n))^{k}}{k^{k+3/2}} \nonumber \\ 
         & \leq    \frac{1}{\sqrt{2} \pi^{3}} \times n^{20 a_{n} c_{n} \log (n)/n}  \times \frac{1 - \eta_{n}}{1 - 2 \eta_{n}} \times  n^{- c_{n}}  \times  \frac{(2 \theta_{n} c_{n} e \log(n))^{k}}{k^{k+3/2}},
\end{align}
where  $\theta_{n} = \sqrt{\eta_{n} (1 -\eta_{n})} = \frac{\sqrt{a_{n} b_{n}}}{a_{n} + b_{n}}$. In step (a), we used Lemma \ref{lemma:binomoal_probability},  while in steps (c) and (d) we used the following lower bound on $k!$, i.e.,  $k! \geq \sqrt{2 \pi} k^{k+1/2} e^{-k}$. In step (b), it can be readily shown that 
\begin{align}
    \frac{(c_{n} e \log(n))^{k}}{k^{k}} \leq n^{c_{n}}. \nonumber 
\end{align}
We next upper bound $\frac{(2 \theta_{n} c_{n} e \log(n))^{k}}{k^{k+3/2}} \triangleq e^{f(k)}$. By taking $\log(\cdot)$ for this term, we have 
\begin{align}
    f(k) = k \times \log(2 \theta_{n} c_{n}  e \log(n)) - (k+3/2) \times \log(k).
\end{align}
We next take the derivative of $f(k)$, 
\begin{align}
     & f'(k) = \log({2 \theta_{n}} c_{n} e \log(n)) - (1 + 3/2k) -  \log(k) = 0. \\
     \Rightarrow & \log(2 \theta_{n} c_{n} e \log(n)) - (1 + 3/2k) - \log(k) = 0, \nonumber \\
   \Rightarrow & \log( e^{-1} 2 \theta_{n} c_{n}  e \log(n)) = 3/2k + \log(k), \nonumber \\ 
     \Rightarrow & e^{-1} 2 \theta_{n} c_{n} e \log(n) = k \times e^{3/2k}, \nonumber \\
   \Rightarrow & 2 \theta_{n} c_{n}   \log(n) =  k \times e^{3/2k}.
 \end{align}
 \noindent Therefore, the optimal solution can be written as 
 \begin{align}
     k^{*} = {2 \theta_{n}} c_{n}  \log(n) e^{-3/2k^{*}}.
 \end{align}
 To this end, we have 
 \begin{align}
     e^{f(k^{*})} & = \frac{n^{2 \theta_{n} c_{n}  e^{-3/{2k^{*}}}}}{2 \theta_{n} c_{n} \log(n) e^{-3/{2k^{*}}}} \nonumber \\
     & \leq \frac{n^{2 \theta_{n} c_{n} e^{-3/{20c_{n} \log(n)}} }}{2 \theta_{n} c_{n} \log(n) } \times e^{3/{2k^{*}}} \nonumber \\ 
     & \leq \frac{n^{2 \theta_{n} c_{n}  }}{2 \theta_{n} c_{n}  \log(n) } \times e^{3/2}.
 \end{align}
The second term will be upper bounded by 
\begin{align}
   n^{- c_{n}} \times  10 c_{n} \log(n) \times \frac{n^{2 \theta_{n} c_{n}  }}{2 \theta_{n} c_{n} \log(n) } \times e^{3/2} & = \frac{5e^{3/2}}{\theta_{n}}  \times n^{-c_{n}  + 2 \theta_{n} c_{n}  } \nonumber \\
   & = \frac{5e^{3/2}}{\theta_{n}} \times n^{-2\left[c_{n}/2  - \theta_{n} c_{n} \right]}.
\end{align}
In order to achieve exact recovery, we require that $a_{n}$ and $b_{n}$: 
\begin{align}
    \frac{c_{n}}{2} - \theta_{n} c_{n}  > 1 \Rightarrow \frac{a_{n} + b_{n}}{2} - {\sqrt{a_{n} b_{n}}} > 1.
\end{align}
Plugging the values of $\theta_{n}$, $c_{n}$, we get 
\begin{align} 
    \operatorname{Pr}(\hat{\bm{\sigma}} \neq  \bm{\sigma} )  & \leq 2 n \times \left[ 2 n^{- (a_{n} + b_{n})} +    \zeta_{n} \times n^{-2\left[(a_{n}+b_{n})/2  - \sqrt{a_{n} b_{n}}  \right]} \right], \label{eqn:upperbound_FER} 
\end{align}
where $\zeta_{n} = \frac{5 \times e^{3/2}}{\sqrt{2} \pi^{3}} \times  e^{20 a_{n} (a_{n} + b_{n}) \log^{2}(n)/n} \times  \frac{a_{n}}{a_{n} - b_{n}} \times \frac{a_{n} + b_{n}}{\sqrt{a_{n} b_{n}}}$. We upper bound the term $e^{20 a_{n} (a_{n} + b_{n}) \log^{2}(n)/n} $ such that
\begin{align}
    & e^{20 a_{n} (a_{n} + b_{n}) \log^{2}(n)/n} \leq     e^{40 a_{n}^{2} \log^{2}(n)/n}  \leq n^{\alpha} = e^{\alpha \log(n)} \nonumber \\ 
    \Rightarrow &  a_{n}^{2}  \leq \frac{\alpha}{40}\times \frac{n}{\log(n)},
\end{align}
where $\alpha < a_{n}+b_{n} - 2 \sqrt{a_{n} b_{n}} - 1$. Plugging the value of $a_{n}$, we get 
\begin{align}
    e^{\epsilon} > \frac{n}{\log(n)} \times \frac{1}{\sqrt{\frac{\alpha n}{40 \log(n)}} - a}. 
\end{align}
Based on the previous condition, we have two cases:  When (1) $a > \sqrt{\frac{\alpha n}{40 \log(n)}} $, in this case, a sufficient condition will be $\epsilon > \log(n) - \log(\log(n))$, and (2) $a < \sqrt{\frac{\alpha n}{40 \log(n)}} $, in this case we require that $\epsilon > \log(n) - \log(\log(n)) - \log \bigg( \sqrt{\frac{\alpha n}{40 \log(n)}} -a \bigg)$.

\subsection*{Derivation of Recovery Threshold Condition:}

The randomized response mechanism $\mathcal{M}_{\text{RR}}(G)$ can be expressed as 
\begin{align}
    \tilde{A}_{i,j} = ({A}_{i,j} + N_{i, j}) \operatorname{mod} 2, \forall i \neq j,
\end{align}
where $N_{i, j} \sim \operatorname{Bern}(1 - \mu )$, $\mu = \frac{1}{e^{\epsilon} + 1}$, $N_{i,j} = N_{j,i}$, and the operation $\operatorname{mod} 2$ ensures that the released output is bounded, i.e., $\tilde{A}_{i,j}  \in \{0, 1\}$. 
If nodes $i$ and $j$ belong to the same community, we have the following:
\begin{align}
   \tilde{p} & =  \operatorname{Pr}(\tilde{A}_{i,j} = 1)  = \operatorname{Pr}(A_{i,j} = 1) \times \operatorname{Pr}(N_{i,j} = 0 | A_{i, j} = 1) + \operatorname{Pr}(A_{i,j} = 0) \times \operatorname{Pr}(N_{i,j} = 0 | A_{i, j} = 1) \nonumber \\ 
    & \overset{(a)} = \operatorname{Pr}(A_{i,j} = 1) \times \operatorname{Pr}(N_{i,j} = 0) + \operatorname{Pr}(A_{i,j} = 0) \times \operatorname{Pr}(A_{i,j} = 1)  \nonumber \\ 
    & = p \times (1 - \mu) + (1 - p) \times \mu \triangleq p \circledast \mu, \label{eqn:perturbed_p}
\end{align}
where in step (a), the perturbation mechanism is independent of the $A_{i,j}$'s. Similarly, if nodes $i$ and $j$ belong to different communities, we have 
\begin{align}
    \tilde{q} & = \operatorname{Pr}(\tilde{A}_{i,j} = 1) 
     =  q \circledast \mu. \label{eqn:perturbed_q}
\end{align}
Plugging the expression of $\mu$, $p$ and $q$ into the previous equations \eqref{eqn:perturbed_p} and \eqref{eqn:perturbed_q}, we get
\begin{align}
\tilde{p} 
& = \left[\frac{n}{(e^{\epsilon} + 1) \times \log(n)} + \frac{e^{\epsilon} -1}{e^{\epsilon} + 1} \times  a \right] \times \frac{\log (n)}{n} \triangleq a_{n} \times \frac{\log (n)}{n}.
\end{align}
Similarly,
\begin{align}
    \tilde{q} &= \left[\frac{n}{(e^{\epsilon} + 1) \times \log(n)} + \frac{e^{\epsilon} -1}{e^{\epsilon} + 1} \times b \right] \times \frac{\log (n)}{n} \triangleq b_{n} \times \frac{\log (n)}{n}.
\end{align}

We next derive a necessary threshold condition for randomized response mechanism. From eqn. \eqref{eqn:upperbound_FER}, in order to ensure exact recovery we require that 
\begin{align}
    \frac{a_{n} + b_{n}}{2} - \sqrt{a_{n} b_{n}} > 1. \label{eqn:basic_threshold_condition}
\end{align}
Plugging the expressions for $a_{n}$ and $b_{n}$ into \eqref{eqn:basic_threshold_condition}, we get the following: 
\begin{align}
    & \frac{1}{2} \left[\frac{2}{e^{\epsilon} + 1} + \frac{e^{\epsilon} - 1}{e^{\epsilon} + 1} \times  (a + b) \right] - \sqrt{\frac{1}{e^{\epsilon} + 1} + \frac{e^{\epsilon} - 1}{e^{\epsilon} + 1} \times a } \times \sqrt{\frac{1}{e^{\epsilon} + 1} + \frac{e^{\epsilon} - 1}{e^{\epsilon} + 1} \times b} > 1 \nonumber \\ 
     \Rightarrow & \frac{1}{e^{\epsilon} + 1} + \frac{e^{\epsilon} - 1}{e^{\epsilon} + 1} \times \frac{a+b}{2} - \frac{1}{e^{\epsilon} + 1}  \times \sqrt{1 + (e^{\epsilon} - 1) \times a} \times  \sqrt{1 + (e^{\epsilon} - 1) \times b} > 1 \nonumber \\ 
     \Rightarrow & \frac{e^{\epsilon} - 1}{e^{\epsilon} + 1} \times \frac{a + b}{2} - \frac{1}{e^{\epsilon} + 1}  \times \sqrt{1 + (e^{\epsilon} - 1) \times a} \times  \sqrt{1 + (e^{\epsilon} - 1) \times b} > \frac{e^{\epsilon}}{e^{\epsilon} + 1} \nonumber \\ 
         \Rightarrow & \frac{a + b}{2} - \frac{1}{e^{\epsilon} - 1}  \times \sqrt{1 + (e^{\epsilon} - 1) \times a} \times  \sqrt{1 + (e^{\epsilon} - 1) \times b} > \frac{e^{\epsilon}}{e^{\epsilon} - 1} \nonumber \\ 
    \Rightarrow & \frac{a+b}{2} - \sqrt{\bigg( \frac{1}{e^{\epsilon} - 1} + a \bigg) \bigg( \frac{1}{e^{\epsilon} - 1} + b\bigg) } > \frac{e^{\epsilon}}{e^{\epsilon} - 1}. 
\end{align}
We can further simplify the above equation as follows:
\begin{align}
    & a + b - 2  \sqrt{\bigg( \frac{1}{e^{\epsilon} - 1} + a \bigg) \bigg( \frac{1}{e^{\epsilon} - 1} + b\bigg) } > \frac{2 e^{\epsilon}}{e^{\epsilon} - 1} \nonumber \\ 
    \Rightarrow & - \frac{2}{e^{\epsilon} - 1} + \frac{1}{e^{\epsilon} - 1} + a + \frac{1}{e^{\epsilon} - 1} + b  - 2  \sqrt{\bigg( \frac{1}{e^{\epsilon} - 1} + a \bigg) \bigg( \frac{1}{e^{\epsilon} - 1} + b\bigg) }  > \frac{2 e^{\epsilon}}{e^{\epsilon} - 1} \nonumber \\
    \Rightarrow & \bigg(\sqrt{\frac{1}{e^{\epsilon} - 1} + a} - \sqrt{\frac{1}{e^{\epsilon} - 1} + b}  \bigg)^{2} > \frac{2 e^{\epsilon}}{e^{\epsilon} - 1}.
\end{align}
To this end, we get 
\begin{align}
    \sqrt{\frac{1}{e^{\epsilon} - 1} + a} - \sqrt{\frac{1}{e^{\epsilon} - 1} + b} > \sqrt{2} \times \sqrt{\frac{e^{\epsilon} + 1}{e^{\epsilon} - 1}}.
\end{align}
A more stringent condition is 
\begin{align}
    & \sqrt{a} - \sqrt{b} - \sqrt{\frac{1}{e^{\epsilon} - 1} } > \sqrt{2} \times \sqrt{\frac{e^{\epsilon} + 1}{e^{\epsilon} - 1}} \nonumber\\
    \Rightarrow & \sqrt{a} - \sqrt{b} > \sqrt{2} \times \sqrt{\frac{e^{\epsilon} + 1}{e^{\epsilon} - 1}} + \frac{1}{\sqrt{e^{\epsilon} - 1}}.
\end{align}
This completes the proof of Theorem 3.8. Note that the threshold condition matches the non-private case when $\epsilon = \infty$.

\section*{Auxiliary Results:}

\begin{lemma} \label{lemma:binomoal_probability} Suppose $X \sim \operatorname{Bin}(k, p) $, then for $ p < 1/2$, we have 
\begin{align}
    \operatorname{Pr}(X \geq \frac{k}{2}) \leq \frac{1 - p }{1 - 2 p}  \times  \operatorname{Pr}(X = \frac{k}{2}).
\end{align}
\end{lemma}
\begin{proof}
  Our goal is to upper bound the following ratio:
\begin{align}
    \frac{\operatorname{Pr}(X \geq \frac{k}{2})}{\operatorname{Pr}(X = \frac{k}{2})} \leq \frac{\operatorname{Pr}(X = \frac{k}{2}) + \operatorname{Pr}(X = \frac{k}{2} + 1) + \cdots + \operatorname{Pr}(X = k)}{\operatorname{Pr}(X = \frac{k}{2})}
\end{align}
For any $m$ and $i$, we have the following: 
\begin{align}
    \frac{\operatorname{Pr}(X = m + i)}{\operatorname{Pr}(X = m)} & = \frac{{k \choose m+i } p^{m+i} (1 - p)^{k-m-i}}{{k \choose m } p^{m} (1 - p)^{k-m}} \nonumber\\
    & = \frac{{k \choose m+i } }{{k \choose m } } \times \bigg(\frac{p}{1-p}\bigg)^{i} \nonumber \\ 
    & = \frac{m!}{(m+i)!} \times \frac{(k-m)!}{(k-m-i)!} \times \bigg(\frac{p}{1-p}\bigg)^{i} \nonumber \\ 
    & \leq \bigg(\frac{k-m}{m+1} \bigg)^{i} \delta^{i},
\end{align}
where $ \delta = \frac{p}{1-p} < 1, \forall p < 1/2$. For $m = \frac{k}{2}$, we have
\begin{align}
    \frac{\operatorname{Pr}(X = \frac{k}{2} + i)}{\operatorname{Pr}(X = \frac{k}{2})} \leq \bigg( \frac{k/2}{k/2+1} \bigg)^{i} \delta^{i} \leq \delta^{i}. 
\end{align}
To this end, 
\begin{align}
     \frac{\operatorname{Pr}(X \geq \frac{k}{2})}{\operatorname{Pr}(X = \frac{k}{2})} & \leq \sum_{i=0}^{k/2} \delta^{i} \leq \sum_{i=0}^{\infty} \delta^{i} \nonumber \\ 
     & = \frac{1}{1 - \delta} = \frac{1}{1 - \frac{p}{1-p}} = \frac{1 - p}{1 - 2p}.
\end{align}
\end{proof}

\begin{lemma} \label{lemma:bernstein's inequality} Consider a random variable $Z$ as a sum of $2 n$ independent random variables, i.e., $Z = \sum_{i = 1}^{n} (X_{i} + Y_{i} )$, where $X_{i} \sim \operatorname{Bern}(p_{n})$, $Y_{i} \sim \operatorname{Bern}(q_{n})$ and $\mu_{Z} = c_{n} \log(n)$. Then, we have 
\begin{align}
    \operatorname{Pr} (Z \geq 10 c_{n} \log(n)) \leq n^{-10 c_{n}}.
\end{align}
\end{lemma}
\begin{proof}
   By direct application of Bernstein's inequality, it is straight forward to show that $|X_{i}| \leq 1$, $|Y_{i}| \leq 1$ and $\sum_{i=1}^{n} \mathds{E} \left[(X_{i} - \mu_{X_{i}})^{2} \right] + \sum_{i=1}^{n} \mathds{E} \left[(Y_{i} - \mu_{Y_{i}})^{2} \right] \leq n (p_{n} + q_{n}) =  c_{n} \log(n)$. Then, we have the following:
\begin{align}
\operatorname{Pr}(Z \geq 10  c_{n} \log(n)) & = \operatorname{Pr}(Z -  c_{n} \log(n) \geq 9 c_{n} \log(n)) \nonumber \\ 
    & = \operatorname{Pr}(Z - \mu_{Z}  \geq 9 c_{n} \log(n)) \nonumber \\ 
    & \leq  \exp{\left[- \frac{\frac{1}{2} \times 81 \times  \times c_{n}^{2} \log^{2}(n)}{  c_{n} \log(n) + \frac{1}{3} \times 9 c_{n} \log(n)  } \right]} \nonumber \\
    & \leq \exp{\left[- \frac{40 c_{n}^{2}  \log^{2}(n)}{  c_{n} \log(n) + 3 c_{n} \log(n)  } \right]} \nonumber \\ 
    & = \exp{(-  10  c_{n} \log(n))} = n^{- 10  c_{n}}.
\end{align}
\end{proof}

\begin{lemma}\label{lemma:chernoff_hoeffding_bound} (Chernoeff-Hoeffding bound) Let $X = \sum_{i \in [n]} X_{i} $, where $X_{i}$'s are indentically and independently distributed over the support $\{0, 1\}$. Then, for any $\gamma \in (0, 1]$, we have 
\begin{align}
    \operatorname{Pr} \left[X \notin \left[ (1-\gamma) \mathds{E}(X), (1+\gamma) \mathds{E}(X) \right] \right] \leq 2 e^{- \frac{\gamma^{2} \mathds{E}(X)}{3}}.
\end{align}
\end{lemma}

\begin{lemma} \label{lemma:optimized_tail_bound} (Tail bounds on the difference of two Binomial R.V.s \cite{hajek2016achieving_extensions}) Let $X$ and $R$ be independent R.V.s with $X \sim \operatorname{Bin} (m_{1}, \frac{a \log(n)}{n})$ and $R \sim \operatorname{Bin}(m_{2}, \frac{b \log(n)}{n})$, where $m_{1}, m_{2} \in \mathds{N}$, such that $f_{n,\epsilon} \leq (m_{1} a -m_{2} b) \frac{\log(n)}{n}$, then 
\begin{align}
    \Pr(X - R \leq f_{n,\epsilon}) \leq n^{-g(m_{1}/n, m_{2}/n, a, b, f_{n,\epsilon}/\log(n))},
\end{align}
where,
\begin{align}
g(m_{1}/n, m_{2}/n, a, b, f_{n,\epsilon}/\log(n)) & = a \times \frac{m_{1}}{n} + b\times \frac{m_{2}}{n} - \gamma - \frac{\alpha}{2} \times \log \left[ \frac{(\gamma - \alpha) a m_{1}}{(\gamma+\alpha) b m_{2}} \right],
\end{align}
where $\alpha = f_{n, \epsilon}/\log(n)$, and $\gamma = \sqrt{\alpha^{2} + 4 \frac{m_{1} m_{2}}{n^{2}} a b}$.
\end{lemma}

 \begin{definition} [Multiplicative Chernoff Bound] 
    Given $\tilde{m}_{1}^{(k)} \sim \operatorname{Bin}(2k(n-k), p)$, and $\tilde{m}_{2}^{(k)} \sim \operatorname{Bin}(2k(n-k), q)$, we have
 \begin{align}
     \operatorname{Pr} \left[ \tilde{m}_{1}^{(k)} - \tilde{m}_{2}^{(k)} < (1-\tilde{\delta}) \mu^{(k)} | |S_{1}| = k,   |S_{2}| = k \right] \leq \exp (- \tilde{\delta}^{2} \mu^{(k)}/2), \label{eqn:multi_chernoff_bound}
 \end{align}
 where, 
 \begin{align}
     \mu^{(k)} & = 2 k (n-k) (a - b) \frac{\log(n)}{n}.
     \end{align}
 \end{definition}
Applying the union bound for possible values of $k \in [1: \frac{n}{2}]$, it yields 
\begin{align}
          \operatorname{Pr} \left[ E_{S_{1}, S_{2}}^{c} \right] & = \sum_{k=1}^{\frac{n}{2}} {n \choose k}^{2} \times \exp \left[ - \frac{\tilde{\delta}^{2}}{2} \times 2 k(n-k) \times (a-b) \times \frac{\log(n)}{n}  \right] \nonumber \\ 
          & \leq \sum_{k=1}^{\frac{n}{2}} \bigg(\frac{ne}{k}\bigg)^{2k} \times   \exp \left[ - \frac{\tilde{\delta}^{2}}{2} \times k \times (a-b)\times {\log(n)}  \right] \nonumber \\
            & = \sum_{k=1}^{\frac{n}{2}} \exp \left[2k \bigg( \log(n) - \log(k) + 1   - \frac{\tilde{\delta}^{2}}{4}  \times (a-b)\times {\log(n)} \bigg) \right]  \nonumber \\
             & = \sum_{k=1}^{\frac{n}{2}} \exp \left[- 2k \bigg(\log(k) - 1   + \bigg(\frac{\tilde{\delta}^{2}}{4}  \times (a-b) - 1\bigg) {\log(n)} \bigg) \right].
\end{align}
In order to make the probability decays with $n$, we require that 
\begin{align}
    \tilde{\delta}^{2}(a-b) > 4 \Rightarrow a-b > \frac{4}{\tilde{\delta}^{2}} \Rightarrow \sqrt{a} - \sqrt{b} > \frac{\sqrt{2}}{\tilde{\delta}/\sqrt{2}}.
\end{align}





%

%% file: 02_Stability_Privacy.tex
\section{Proof of Lemma 3.2 ($\mathcal{M}^{\text{MLE}}_{\operatorname{stability}}(G)$ satisfies $(\epsilon, \delta)$-edge DP)} \label{appendix:privacy_guarantee_stability}

\begin{algorithm}
  \caption{$\mathcal{M}^{\hat{\bm{\sigma}}}_{\operatorname{Stability}}(G)$: Stability  Based Mechanism}
  \label{algo:stability}
  \begin{algorithmic}[1]
     \STATE {\bfseries Input:} $G(\mathcal{V}, E) \in \mathcal{G}$
     \STATE {\bfseries Output:} labelling vector $\hat{\bm{\sigma}}_{\text{Private}}$. 
    \STATE $d_{\hat{\bm{\sigma}}}(G) \leftarrow$ stability of $\hat{\bm{\sigma}}$ with respect to graph $G$
    \STATE $\tilde{d}(G) \leftarrow d_{\hat{\bm{\sigma}}}(G) + \operatorname{Lap}(1/\epsilon)$
    \IF{$\tilde{d}(G)  > \frac{\log{1/\delta}}{\epsilon}$}
    \STATE Output $\hat{\bm{\sigma}}(G)$
    \ELSE
    \STATE Output $\perp$ (random label) 
    \ENDIF
  \end{algorithmic}
\end{algorithm}

The proof that the stability based mechanism satisfies $(\epsilon, \delta))$-edge DP follows directly from \cite{dwork2014algorithmic}, and we include the proof here for the sake of completeness by adapting it to the community detection problem. In this analysis, we drop the subscript $\hat{\sigma}$ when the context is clear. Given a pair of neighbor graphs $G\sim G'$, $d(G)$ denotes the distance from $G$ to its nearest unstable instance and $d(G')$ is the distance from $G'$ to its nearest unstable instance. Due to the triangle inequality, $| d(G)-d(G')| \leq 1$, hence the sensitivity of $d$: $\Delta_d = 1$. Adding a Laplacian noise of magnitude of $1/\epsilon$ guarantees $\epsilon$-differential privacy for $\tilde{d}$. In order to verify $(\epsilon, \delta)$-edge DP for the overall mechanism, we consider two scenarios: the first one, when the output of the mechanism is $\perp$. In this case, we have:


 \begin{align}
    \operatorname{Pr}[\mathcal{M}_{\operatorname{Stability}}(G) = \perp] &= \operatorname{Pr} \left[\tilde{d} (G) \leq  \frac{\log{1/\delta}}{\epsilon} \right] \nonumber \\
                               &\overset{(a)}{\leq} e^\epsilon \operatorname{Pr} \left[\tilde{d}(G') \leq \frac{\log{1/\delta}}{\epsilon} \right] \nonumber \\
    &=e^\epsilon \operatorname{Pr}[\mathcal{M}_{\operatorname{Stability}}(G') = \perp].
  \end{align}
where $(a)$ follows from the fact that $\tilde{d}$ satisfies $\epsilon$-DP. For the second scenario, when the output of the mechanism is some label vector $\bm{\sigma}$, we have to analyze two cases. 

   The remaining part of the proof, we prove that output $\hat{\bm{\sigma}}(G)$ in line $6$ satisfies $(\epsilon, \delta)$-differential privacy to fullfill the proof of the theorem. We analyze two cases (1) $d(G) = 0$ and (2) $d(G) > 0$.

   \textbf{Case 1. $d(G)=0$}, we have $\Pr[\tilde{d}(G) > \frac{\log{1/\delta}}{\epsilon}] = \Pr[Lap(1/\epsilon) > \frac{\log{1/\delta}}{\epsilon}] \leq e^{-\log{1/\delta}} = \delta$. For any set of output $S \subseteq (Range(\hat{\sigma}) \cup \{\perp\})$, we have
   \begin{align*}
     \Pr[\mathcal{M}_{Stability}(G) \in S] &\leq \Pr[\mathcal{M}_{Stability}(G) \in (S\cup \{\perp\})] \\
     &\leq \Pr[\mathcal{M}_{Stability}(G) \in (S \cap \{\perp\})] + \Pr[\mathcal{M}_{Stability}(G) \neq \perp] \\
     &\leq \Pr[\mathcal{M}_{Stability}(G) \in (S \cap \{\perp\})] + \delta \\
     &\leq e^\epsilon\Pr[\mathcal{M}_{Stability}(G') \in (S \cap \{\perp\})] + \delta \\
     &\leq e^\epsilon\Pr[\mathcal{M}_{Stability}(G') \in S] + \delta, 
   \end{align*}
   where the third inequality is because $\Pr[M_{Stability}(G) \neq \perp] = \Pr[\tilde{d}(G) > \frac{\log{1/\delta}}{\epsilon}] \leq \delta$ and the fourth inequality is because $S\cap \{\perp\}$ is \textbf{(a)} $\emptyset$ or \textbf{(b)} $\{\perp\}$. When $\textbf{(a)}$ happens, $\Pr[M_{Stability}(G) \in \emptyset] = \Pr[M_{Stability}(G') \in \emptyset] = 0$ and when $\textbf{(b)}$ happens, it follows above proof that $\Pr[M_{Stability}(G) = \perp] \leq e^\epsilon\Pr[M_{Stability}(G') = \perp]$.
 
   \textbf{Case 2. $d(G)>0$}. In this case, $G$ is at least $1$-stable, which means: $\bm{\sigma}(G) = \bm{\sigma}(G') = \sigma$, we have:

   \begin{align*}
     \Pr[M_{Stability}(G) = \sigma] &= \Pr[\tilde{d}(G) > \frac{\log{1/\delta}}{\epsilon}] \\
     &\leq e^\epsilon \Pr[\tilde{d}(G') > \frac{\log{1/\delta}}{\epsilon}] \\
                                     &= e^\epsilon\Pr[M_{Stability}(G') = \sigma],
   \end{align*}

   and the Lemma follows.
  

%% file: 01_Stability_MLE.tex

\section{Proof of Theorem 3.3 (Exact Recovery Threshold for $\mathcal{M}^{\text{MLE}}_{\operatorname{Stability}}(G)$ for $r = 2$)}  \label{appendix:error_probability_proof_stability}
The error probability for the stability mechanism can be expressed as 
\begin{align}
 \operatorname{Pr}(\mathcal{M}_{\operatorname{Stability}}(G)\neq \bm{\sigma}^{*})  & = \operatorname{Pr}(\mathcal{M}_{\operatorname{Stability}}(G) = \perp) \times \operatorname{Pr}(\mathcal{M}_{\operatorname{Stability}}(G) \neq \bm{\sigma}^{*} | \mathcal{M}_{\operatorname{Stability}}(G) = \perp) \nonumber\\
 & \hspace{0.1in} + \operatorname{Pr}(\mathcal{M}_{\operatorname{Stability}}(G) = \hat{\bm{\sigma}}_{\text{ML}}) \times \operatorname{Pr}(\mathcal{M}_{\operatorname{Stability}}(G)\neq \bm{\sigma}^{*} | \hat{\bm{\sigma}}_{\text{stb.}} = \hat{\bm{\sigma}}_{\text{ML}}) \nonumber \\ 
 & \leq \operatorname{Pr}(\hat{\bm{\sigma}}_{\operatorname{Stability}} = \perp) \times 1 + 1 \times \operatorname{Pr}(\hat{\bm{\sigma}}_{\operatorname{Stability}} \neq \bm{\sigma}^{*} | \mathcal{M}_{\operatorname{Stability}}(G) = \hat{\bm{\sigma}}_{\text{ML}}) \nonumber \\ 
 & \leq \operatorname{Pr}(\mathcal{M}_{\operatorname{Stability}}(G) = \perp) + \operatorname{Pr}(\hat{\bm{\sigma}}_{\text{ML}} \neq \bm{\sigma}^{*}),
\end{align}
where the probability is taken over the randomness of the Laplacian mechanism and over the randomness graph generated from SBM. We further upper bound the first term in the above equation as follows:
\begin{align}
    \operatorname{Pr}(\mathcal{M}_{\operatorname{Stability}}(G)= \perp) & = \operatorname{Pr}\left[\tilde{d}(G) < \frac{\log(1/\delta)}{\epsilon}\right] \nonumber \\
    & = \operatorname{Pr}\left[d(G) + \text{Lap}(1/\epsilon) < \frac{\log(1/\delta)}{\epsilon} \right] \nonumber \\ 
    & = \operatorname{Pr}\left[\text{Lap}(1/\epsilon) < \frac{\log(1/\delta)}{\epsilon} - d(G)\right] \nonumber \\
     & \leq \operatorname{Pr}\left[d(G) <   f_{n,\epsilon} \right] \times \operatorname{Pr}\left[\text{Lap}(1/\epsilon) < \frac{\log(1/\delta)}{\epsilon} \right] \nonumber \\
     & \hspace{0.1in} +  \operatorname{Pr}\left[d(G) \geq  f_{n,\epsilon} \right] \times \operatorname{Pr}\left[\text{Lap}(1/\epsilon) < \frac{\log(1/\delta)}{\epsilon} -   f_{n,\epsilon}\right] \nonumber \\ 
     & \leq \underbrace{\operatorname{Pr}\left[d(G) <   f_{n,\epsilon} \right]}_{\text{Term 1}} + \underbrace{\operatorname{Pr}\left[\text{Lap}(1/\epsilon) < \frac{\log(1/\delta)}{\epsilon} -   f_{n,\epsilon}\right]}_{\text{Term 2}}. \label{eqn:prob_outputting_null}
\end{align}

\textbf{Bounding Term 2:} By picking $f_{n,\epsilon} = (t+1) \log(n)/\epsilon$ and $\delta = n^{-t}$ for any positive $t$, it can be readily shown that $\operatorname{Pr}\left[\text{Lap}(1/\epsilon) < \frac{\log(1/\delta)}{\epsilon} -   f_{n,\epsilon} = -\log{(n)}/\epsilon\right] = o(1)$.
To upper bound Term $1$, we introduce an intermediate lemma which gives a lower bound on $d(G)$ as follows.


\begin{lemma}
  \label{lemma:mle-stable_lower_bound} 
  Let $(R^{\text{ML}}, B^{\text{ML}})$ be the output of $\text{MLE}(G)$. Let $d(G)$ be the distance from $G$ to the nearest unstable instance, then $d(G)$ is lower bounded by
\begin{align}
     d(G) \geq  \min_{(R, B)\neq (R^{\text{ML}}, B^{\text{ML}})}{{E}^{(G)}(R, B) - {E}^{(G)}(R^{\text{ML}}, B^{\text{ML}})}.
\end{align}
where ${E}^{(G)}(R, B)$ denotes the number of edges between partitions $R$ and $B$ of the graph $G$.
\end{lemma}
\begin{proof}   W.l.o.g, let's consider two equal sized communities $R$ and $B$. We prove this Lemma by contradiction. For a fixed graph $G$, let us assume that we have $\min_{(R, B)\neq (R^{\text{ML}}, B^{\text{ML}})}{E}^{(G)}(R, B)
- {E}^{(G)}(R^{\text{ML}}, B^{\text{ML}}) = \tilde{d} > d$. Let $\tilde{G}$ is the nearest graph from $G$ that we have $\text{MLE}(G) \neq \text{MLE}(\tilde{G})$. We know that $\text{dist}(G, \tilde{G}) = d $ so we can obtain $\tilde{G}$ by adding $m_{add}$ edges to and removing $m_{\text{remove}}$ edges from $G$ and $m_{\text{add}} + m_{\text{remove}} = d$. For any labelling $(R, B) \neq (R^{\text{ML}}, B^{\text{ML}})$, we have
   \begin{align}
     {E}^{(\tilde{G})}(R^{\text{ML}}, B^{\text{ML}}) &\leq {E}^{(G)}(R^{\text{ML}}, B^{\text{ML}}) + m_{\text{add}} \nonumber \\
                                &\leq \min_{(R, B)\neq (R^{\text{ML}}, B^{\text{ML}})}{E}^{(G)}(R, B) - \tilde{d} + m_{\text{add}} \nonumber  \\
                                &\leq \min_{(R, B)\neq (R^{\text{ML}}, B^{\text{ML}})}{E}^{(G)}(R, B) + m_{\text{remove}} + m_{\text{add}} - \tilde{d} \nonumber  \\
                                &\leq \min_{(R, B)\neq (R^{\text{ML}}, B^{\text{ML}})}{E}^{(G)}(R, B) + d - \tilde{d} \nonumber \\
                                 & < \min_{(R, B)\neq (R^{\text{ML}}, B^{\text{ML}})}{E}^{(G)}(R, B) + d
   \end{align}
   which implies $(R^{ML}, B^{ML}) = \text{MLE}(\tilde{G})$, which contradicts that $\text{MLE}(G) \neq \text{MLE}(\tilde{G})$.
\end{proof}

With this Lemma, we now return to analyze the first term in \eqref{eqn:prob_outputting_null} as follows. \\

\textbf{Bounding Term 1:} By expanding the probability by law of total of probability theorem, we get the following sequences of steps:
\begin{align}
    & \operatorname{Pr} \left[ d(G) < f_{n, \epsilon}\right]   
\nonumber \\
    &\overset{(a)}{\leq}  \operatorname{Pr } \left[ E(R^{\min},B^{\min}) - E(R^{\text{ML}}, B^{\text{ML}}) < f_{n,\epsilon}\right] \nonumber \\
    & = \operatorname{Pr } \left[ E(R^{\min},B^{\min}) - E(R^{\text{ML}}, B^{\text{ML}}) < f_{n,\epsilon}| E(R^{\text{ML}}, B^{\text{ML}}) = E(R^{*}, B^{*})\right] \times  \operatorname{Pr} \left[E(R^{\text{ML}}, B^{\text{ML}}) = E(R^{*}, B^{*}) \right]  \nonumber\\
    & \hspace{0.1in} +  \Pr \left[ E(R^{\min},B^{\min}) - E(R^{\text{ML}}, B^{\text{ML}}) < f_{n,\epsilon} | E(R^{\text{ML}}, B^{\text{ML}}) \neq E(R^{*}, B^{*})\right] \times \operatorname{Pr} \left[E(R^{\text{ML}}, B^{\text{ML}}) \neq  E(R^{*}, B^{*}) \right] \nonumber \\ 
       & \overset{(b)} \leq \operatorname{Pr } \left[  E(R^{\min},B^{\min}) - E(R^{*}, B^{*}) < f_{n,\epsilon}\right] \times 1 + 1 \times o(1) \nonumber\\
       & = \operatorname{Pr } \left[  E(R^{\min},B^{\min}) - E(R^{*}, B^{*}) < f_{n,\epsilon}\right]  + o(1),
\end{align}
where (a) follows from Lemma $2.1$, and we have defined $(R^{\min},B^{\min})$ as the solution of the minimization $\min_{(R, B)\neq (R^{\text{ML}}, B^{\text{ML}})}{E}^{(G)}(R, B)$; in step (b), $\operatorname{Pr} \left[E(R^{\text{ML}}, B^{\text{ML}}) \neq  E(R^{*}, B^{*}) \right] = o(1)$, when $\sqrt{a} - \sqrt{b} > \sqrt{2}$ \cite{abbe2015exact}.  Note that any two communities $(R,B) \neq (R^{*}, B^{*})$, $(R,B)$ can be expressed as 
\begin{align}
    B & = B^{*} - S_{2} + S_{1}, \nonumber \\ 
    R & = R^{*} - S_{1} + S_{2},
\end{align}
where $S_{1}$ and $S_{2}$ are the set of mis-classified labels with respect to $R^{*}$ and $B^{*}$, respectively. By the construction of symmetric communities, we have $|S_{1}| = |S_{2}|$. We can further write the $E(R^{\min}, B^{\min})- E(R^{*}, B^{*})$ as 
  \begin{align}
E(R^{\min}, B^{\min}) - E(R^*, B^*)  &= E(S_{1}, R^{*} \backslash S_{1}) + E(S_{2}, B^{*} \backslash S_{2}) \nonumber \\ 
& \hspace{0.1in} - E(S_{1}, B^{*} \backslash S_{2})  - E(S_{2}, R^{*} \backslash S_{1}).
 \end{align} 
 Given $S_{1}$, $S_{2}$ both of size $k$, $k \in [1, \frac{n}{2}]$, we have 
 \begin{align}
 \operatorname{Pr}(E(R^{\min}, B^{\min}) - E(R^*, B^*) \leq f_{n, \epsilon} | |S_{1}| = k, |S_{2}| = k) = \operatorname{Pr}(\tilde{m}_{1}^{(k)} - \tilde{m}_{2}^{(k)} < f_{n,\epsilon}), \label{eqn:difference_binomial_for_specific_k}
 \end{align}
 where $\tilde{m}_{1}^{(k)} \sim \operatorname{Bin}(2k(n-k), p)$, and $\tilde{m}_{2}^{(k)} \sim \operatorname{Bin}(2k(n-k), q)$. Applying Chernoff's bounds, we get 
 \begin{align}
     \operatorname{Pr}(\tilde{m}_{1}^{(k)} - \tilde{m}_{2}^{(k)} < f_{n,\epsilon}) & \leq \min_{\lambda > 0} e^{\lambda f_{n,\epsilon}} \times \mathds{E} \left[e^{-\lambda(m_{1}^{(k)} - m_{2}^{(k)} )} \right] \nonumber \\ 
     & = \min_{\lambda > 0} e^{\lambda f_{n,\epsilon}} \times \mathds{E} \left[e^{-\lambda m_{1}^{(k)} } \right] \times \mathds{E} \left[e^{\lambda m_{2}^{(k)}} \right] \nonumber \\
     & = \min_{\lambda > 0} e^{\lambda f_{n,\epsilon}} \times (1- p (1- e^{-\lambda}))^{2k(n-k)} \times  (1- q (1- e^{\lambda}))^{2k(n-k)} \nonumber \\
     & \leq n^{-g}, 
     \label{eqn:chernoff_bound_given_k}
 \end{align}
 where $   g = (a+b) \times \frac{2 k (n-k)}{n} - \gamma - \frac{\alpha}{2} \log \left[ \frac{(\gamma - \alpha)a}{(\gamma + \alpha) b} \right] $ and  $\gamma = \sqrt{\frac{(t+1)^2}{\epsilon^{2}} + 4 \times \frac{4 k^{2}(n-k)^{2}}{n^{2}} ab}$ and $\alpha=(t+1)/\epsilon$. The upper bound in \eqref{eqn:chernoff_bound_given_k} is invoked from an existing result \cite{hajek2016achieving_extensions} (stated in Lemma  \ref{lemma:optimized_tail_bound}).  We further lower bound $g$ as follows:
\begin{align}
    g & = (a+b) \times \frac{2 k (n-k)}{n} - \gamma - \frac{\alpha}{2} \log \left[ \frac{(\gamma - \alpha)a}{(\gamma + \alpha) b} \right] \nonumber \\
    & \geq (a+b) \times \frac{2 k (n-k)}{n} - \gamma - \frac{\alpha}{2} \log \bigg(\frac{a}{b} \bigg) \nonumber \\ 
    & = (a+b) \times \frac{2 k (n-k)}{n} - \sqrt{\frac{(t+1)^2}{\epsilon^{2}} + 4 \times \frac{4 k^{2}(n-k)^{2}}{n^{2}} ab} - \frac{(t+1)}{2\epsilon} \log \bigg(\frac{a}{b} \bigg) \nonumber \\ 
    & = (a+b) \times \frac{2 k (n-k)}{n} - \frac{2 k (n-k)}{n}  \times \sqrt{\frac{n^{2}}{4k^{2}(n-k)^{2}} \times  \frac{(t+1)^2}{\epsilon^{2}} + 4 a b} - \frac{t+1}{2\epsilon} \log \bigg(\frac{a}{b} \bigg) \nonumber \\ 
     & \overset{(a)} \geq  (a+b) \times \frac{2 k (n-k)}{n} - \frac{2 k (n-k)}{n}  \times \sqrt{\frac{(t+1)^2}{\epsilon^{2}} + 4 a b} - \frac{t+1}{2\epsilon} \log \bigg(\frac{a}{b} \bigg) \nonumber \\ 
        & = (a+b) \times \frac{2 k (n-k)}{n} - \frac{2 k (n-k)}{n} \times 2  \times \sqrt{\frac{(t+1)^2}{4\epsilon^{2}} + a b} - \frac{(t+1)}{2\epsilon} \log \bigg(\frac{a}{b} \bigg) \nonumber \\ 
           & = \frac{4 k (n-k)}{n} \times \underbrace{\left[ \frac{a+b}{2} - \sqrt{\frac{(t+1)^2}{4\epsilon^{2}} + a b} \right]}_{C_{a,b}} - \frac{t+1}{2\epsilon} \log \bigg(\frac{a}{b} \bigg),
\end{align}
where in step (a), we have the term $\frac{n^{2}}{k^{2} (n-k)^{2}} \leq  4, \forall k \in [1:\frac{n}{2}]$. Applying the union bound and assuming that $C_{a,b} > 1$, we have 
\begin{align}
    & \Pr (E(R^{\min}, B^{\min}) - E(R^*, B^*) < f_{n, \epsilon} )  \leq \sum_{k=1}^{\frac{n}{2}} {n \choose k}^{2} \times \exp \left[ - \frac{\log(n)}{n} \times 4 k (n-k) C_{a,b} + \frac{t+1}{2\epsilon} \log \bigg(\frac{a}{b} \bigg) \right] \nonumber \\ 
    & \leq \sum_{k=1}^{\frac{n}{2}} \bigg( \frac{n e}{k}\bigg)^{2k} \times \exp \left[ - \frac{\log(n)}{n} \times 4 k (n-k) C_{a,b} + \frac{t+1}{2\epsilon} \log \bigg(\frac{a}{b} \bigg) \right] \nonumber \\ 
    & = \sum_{k=1}^{\frac{n}{2}} \exp \left[2 k \bigg(\log\bigg(\frac{n}{k}\bigg) +1 \bigg) - \frac{\log(n)}{n} \times 4 k (n-k) C_{a,b} + \frac{t+1}{2\epsilon} \log \bigg(\frac{a}{b} \bigg) \right] \nonumber \\ 
    & = \sum_{k=1}^{\frac{n}{2}} \exp \left[ 2k \bigg(\log(n) - \log(k) + 1 - \bigg(2 - \frac{2k}{n}\bigg) C_{a,b} \log(n)\bigg) + \frac{t+1}{2\epsilon} \log \bigg(\frac{a}{b} \bigg) \right] \nonumber 
     \end{align}
    \begin{align}
    & \leq \sum_{k=1}^{\frac{n}{2}} \exp \left[ 2k \bigg(\log(n) - \log(k) + 1 - 2\times \bigg(1 - \frac{k}{n} \bigg) \times 1 \times  \log(n)\bigg) + \frac{t+1}{2\epsilon} \log \bigg(\frac{a}{b} \bigg) \ \right] \nonumber \\ 
            & = \sum_{k=1}^{\frac{n}{2}} \exp \left[2k \bigg(- \log(n) - \log(k) + 1 + \frac{2k}{n} \log(n) \bigg) + \frac{t+1}{2\epsilon} \log \bigg(\frac{a}{b} \bigg)  \right] \nonumber \\ 
    & =  \sum_{k=1}^{\frac{n}{2}} \exp \left[2k \bigg(- \log(k) + 1 + \frac{2k}{n} \log(n) \bigg) \right] \times \exp (-2k \log(n)) \times \exp \left[ \frac{t+1}{2\epsilon} \log \bigg( \frac{a}{b}\bigg) \right] \nonumber \\ 
    & < n^{-2} \times \bigg( \frac{a}{b}\bigg)^{(t+1)/2\epsilon}  \times \sum_{k=1}^{\frac{n}{2}} \exp \left[2k \bigg(- \log(k) + 1 + \frac{2k}{n} \log(n) \bigg) \right] \nonumber \\ 
    & = n^{-2} \times \bigg( \frac{a}{b}\bigg)^{(t+1)/2\epsilon} \times \sum_{k=1}^{\frac{n}{2}} \exp \left[- 2k \bigg( \log(k)  - \frac{2k}{n} \log(n) - 1 \bigg) \right] \nonumber \\ 
    & \overset{(a)} \leq  n^{-2} \times \bigg( \frac{a}{b}\bigg)^{(t+1)/2\epsilon}  \times \sum_{k=1}^{\frac{n}{2}} \exp \left[- 2k \bigg( \frac{1}{3} \log(k) - 1 \bigg) \right] \nonumber \\ 
    & \overset{(b)} = o(1), \nonumber 
\end{align}
where in step (a), we have that $\log(k)  - \frac{2k}{n} \log(n) > \frac{1}{3} \log(k)$ for sufficiently large $n$ \cite{abbe2015exact}. In step (b),  we have that $\big( \frac{a}{b}\big)^{(t+1)/2\epsilon}  \times \sum_{k=1}^{\frac{n}{2}} \exp \left[- 2k \bigg( \frac{1}{3} \log(k) - 1 \bigg) \right] = O(1)$. To this end, we have the following conditions on $a$ and $b$: 
\begin{align}
    & a - b > \frac{t+1}{\epsilon} \Rightarrow \sqrt{a - b} > \frac{\sqrt{t+1}}{\sqrt{\epsilon}} \overset{(a)} \Rightarrow \sqrt{a} - \sqrt{b} > \frac{\sqrt{t+1}}{\sqrt{\epsilon}},\\
    & \frac{a+b}{2} - \sqrt{\frac{(t+1)^2}{4\epsilon^{2}} + ab} > 1 \overset{(b)} \Rightarrow \sqrt{a} - \sqrt{b} > \sqrt{2} \times  \sqrt{1 + \frac{t+1}{2\epsilon}},
\end{align}
where in (a), we have that $\sqrt{a -b} > \sqrt{a} - \sqrt{b}$ where $a > b$, while in (b), we used the fact that $\sqrt{a + b} < \sqrt{a} + \sqrt{b}$. Therefore, a sufficient condition to make Term 1 behave as $o(1)$ will be
\begin{align}
    \sqrt{a} - \sqrt{b} & > \sqrt{2} \times \max \left[ \frac{\sqrt{t+1}}{\sqrt{2\epsilon}}, \sqrt{1 + \frac{t+1}{2\epsilon}} \right] \nonumber \\ 
    & = \sqrt{2} \times \sqrt{1 + \frac{t+1}{2\epsilon}}. 
\end{align}

This completes the proof of Theorem 3.3.

\section{Proof of Theorem 3.4  (Threshold condition for $\mathcal{M}^{\text{MLE}}_{\operatorname{Stability}}(G)$ for $r > 2$)} \label{appendix:proof_multiple_communities}


The proof steps follow on similar lines as the $r=2$ case. Specifically, the error probability boils down to establishing an upper bound on $  \operatorname{Pr}\big(d(G)\leq \frac{(t+1)\log(n)}{\epsilon} \big) $ (similar to Term $1$ in the proof of Theorem $3.3$) as follows.
\begin{align}
      \operatorname{Pr}\big(d(G)\leq \frac{(t+1)\log(n)}{\epsilon} \big) 
    & \overset{(a)} \leq  rn \times \operatorname{Pr}\bigg(\operatorname{Bin}\bigg(\frac{n}{r}, p\bigg) -\operatorname{Bin}\bigg(\frac{n}{r}, q\bigg)   \leq \frac{(t+1)\log(n)}{\epsilon}\bigg) \leq n^{-g}, \label{eqn:tail_difference_multiple}
\end{align}
where step (a) follows from applying the union bound. In order to further upper bound \eqref{eqn:tail_difference_multiple}, we invoke Lemma \ref{lemma:optimized_tail_bound}. Define $\gamma = \sqrt{\frac{(t+1)^2}{\epsilon^{2}} + 4 \times \frac{ab}{r^{2}} }$ and $\alpha = \frac{t+1}{\epsilon}$. The function $g$ is lower bounded as follows:
\begin{align}
    g & = \frac{a+b}{r} - \gamma - \frac{\alpha}{2} \log \left[ \frac{(\gamma - \alpha)a}{(\gamma + \alpha) b} \right] \nonumber \\
    & \geq \frac{a+b}{r}  - \gamma - \frac{\alpha}{2} \log \bigg(\frac{a}{b} \bigg) \nonumber \\ 
    & = \frac{a+b}{r} - \sqrt{\frac{(t+1)^2}{\epsilon^{2}} + 4 \times \frac{ab}{r^{2}}} - \frac{t+1}{2\epsilon} \log \bigg(\frac{a}{b} \bigg) \nonumber \\ 
    & = \frac{a+b}{r} - \frac{2}{r} \sqrt{\frac{r^2(t+1)^2}{4\epsilon^{2}} +  ab} - \frac{t+1}{2\epsilon} \log \bigg(\frac{a}{b} \bigg) \nonumber \\ 
      & = \frac{1}{r} \times  \left[{a+b} - 2 \sqrt{\frac{r^2(t+1)^2}{4\epsilon^{2}} +  ab} - \frac{r(t+1)}{2\epsilon} \log \bigg(\frac{a}{b} \bigg) \right]\nonumber \\
        & \geq \frac{1}{r} \times  \left[{a+b} - 2 \sqrt{ab} -  \frac{(t+1)r}{\epsilon}  - \frac{r(t+1)}{2\epsilon} \log \bigg(\frac{a}{b} \bigg) \right]\nonumber \\
         & \geq \frac{1}{r} \times  \left[{a+b} - 2 \sqrt{ab} -  \frac{r}{\epsilon} \times \bigg(t+1  +  \frac{t+1}{2}\log \bigg(\frac{a}{b} \bigg)\bigg) \right]\nonumber \\
           & = \frac{1}{r} \times  \left[(\sqrt{a} - \sqrt{b})^{2} -  \frac{r}{\epsilon} \times \bigg(t+1  + \frac{t+1}{2} \log \bigg(\frac{a}{b} \bigg)\bigg) \right]\nonumber \\
           & = \frac{(\sqrt{a} - \sqrt{b})^{2} }{r} - \frac{1}{\epsilon} \times \bigg(t+1  + \frac{t+1}{2} \log \bigg(\frac{a}{b} \bigg)\bigg)
\end{align}
We have the following conditions on $a$ and $b$: 
\begin{align}
   & \frac{(t+1) \log(n)}{\epsilon} \leq \frac{n}{r} \times (a-b) \times \frac{\log(n)}{n} \nonumber \\ 
    & \Rightarrow a - b \geq \frac{(t+1) r}{\epsilon} \nonumber \\ 
    & \Rightarrow \sqrt{a} - \sqrt{b} \geq \frac{\sqrt{t+1}}{\sqrt{\epsilon}} \times \sqrt{r}.
\end{align}
Also, we require that 
\begin{align}
   & 1 - \frac{(\sqrt{a} - \sqrt{b})^{2} }{r} + \frac{1}{\epsilon} \times \bigg(t+1  +  \frac{t+1}{2}\log \bigg(\frac{a}{b} \bigg)\bigg) < 0 \nonumber \\ 
   & \Rightarrow \frac{(\sqrt{a} - \sqrt{b})^{2} }{r}  > 1 + \frac{1}{\epsilon} \times \bigg(t+1  + \frac{t+1}{2} \log \bigg(\frac{a}{b} \bigg)\bigg) \nonumber \\ 
   & \Rightarrow \sqrt{a} - \sqrt{b} > \sqrt{r} \times \sqrt{ 1 + \frac{1}{\epsilon} \times \bigg(t+1  + \frac{t+1}{2} \log \bigg(\frac{a}{b} \bigg)\bigg)}, \label{eqn:condition_multiple_communities}
\end{align}
where \eqref{eqn:condition_multiple_communities} leads to a sufficient condition on $a$ and $b$ for exact recovery. This completes the proof of Theorem 3.4.

Before delving into proving Theorem 3.5, we first present the optimization problem for SDP relaxation for $r = 2$ and $r > 2$ communities as follows.

%% file: 00_Stability_SDP.tex
\section*{SDP Relaxation Recovery Algorithm:}
 Let us first define $\mathbf{Y} = {\bm{\sigma}} {\bm{\sigma}}^{T}$, where $Y_{i,i} = 1, \forall i \in [n]$, and $\mathbf{J}$ as all ones matrix. Our goal is to solve the following optimization problem:
\begin{align}
    \hat{\mathbf{Y}}_{\text{SDP}} = \max_{\mathbf{Y}} &  \hspace{0.1in}\text{tr}({\mathbf{A}} \mathbf{Y}) \nonumber \\ 
     \text{s.t.} & \hspace{0.1in}  \mathbf{Y} \succcurlyeq \mathbf{0} \nonumber \\ 
    & Y_{i, i} = 1, \forall i \in [n] \nonumber \\ 
    & \text{tr}(\mathbf{J} \mathbf{Y}) = 0. \label{eqn:SDP_relaxation}
\end{align}
It has been shown that if $\sqrt{a} - \sqrt{b} > \sqrt{2}$, then $\operatorname{Pr}(\hat{\mathbf{Y}}_{\text{SDP}} = \mathbf{Y}^{*}) = 1 - o(1)$.  For $r$ communities each of size $\frac{n}{r}$, the ML estimator \cite{hajek2016achieving} is given as: 
\begin{align}
       \hat{\mathbf{Z}}_{\text{SDP}} = \max_{\mathbf{Z}} &  \hspace{0.1in}\text{tr}({\mathbf{A}} \mathbf{Z}) \nonumber \\ 
     \text{s.t.} & \hspace{0.1in}  \mathbf{Z} \succcurlyeq \mathbf{0} \nonumber \\ 
    & Z_{i, i} = 1, \forall i \in [n] \nonumber \\
    & Z_{i,j} \geq 0, i,j \in [n] \nonumber \\ 
    & \mathbf{Z} \mathbf{1} = \frac{n}{r} \mathbf{1}, \label{eqn:SDP_relaxation_r_communities}
\end{align}
where $\mathbf{Z}^{*} = \sum_{k=1}^{r} \xi_{k}^{*} (\xi_{k}^{*})^{T}$, and $\xi_{k}^{*}$ is a binary vector that is an indicator function for community $k$, such that $\xi_{k}(i) = 1$ if vertex $i$ is in community $k$ and $\xi_{k}(i) = 0$, otherwise. It has been shown that if $\sqrt{a} - \sqrt{b} > \sqrt{r}$, then $\operatorname{Pr}(\hat{\mathbf{Z}}_{\text{SDP}} = \mathbf{Z}^{*}) = 1 - o(1)$.


\section{Proof of Theorem 3.5  (Threshold condition for $\mathcal{M}^{\text{SDP}}_{\operatorname{Stability}}(G)$ for $r \geq 2$)}



\begin{lemma}
  \label{lemma:f-stable}
  Given any function $f: \mathcal{G} \rightarrow \mathcal{R}$, the $f$-based Stability algorithm $\mathcal{M}^f_{\operatorname{Stability}}$ with $\delta = n^{-t}$ for any positive $t$ outputs $f(G)$ with probability at least $1-\mathcal{O}(n^{-1})$, if $G$ is $\frac{(t+1)\log{n}}{\epsilon}$-stable under $f$ with probability at least $1-\mathcal{O}(n^{-1})$.
\end{lemma}

\begin{proof} 


  Because $G$ is $\frac{(t+1)\log{n}}{\epsilon}$-stable under $f$ with probability at least $1-\mathcal{O}(n^{-1})$, $d(G) \geq \frac{(t+1)\log{n}}{\epsilon}$ with probability at least $1-\mathcal{O}(n^{-1})$. We drop the parameter $(G)$ when the context is clear. The probability that $\mathcal{M}^f_{Stability}(G)$ does not output $f(G)$ is:
  \begin{align}
    \operatorname{Pr}\left[\mathcal{M}^{f}_{\operatorname{Stability}}(G) \neq f(G) \right] &= \operatorname{Pr} \left[\tilde{d} \leq \frac{\log{1/\delta}}{\epsilon} \right]\\
                                       &= \operatorname{Pr} \left[d + \operatorname{Lap}(1/\epsilon) \leq  \frac{\log{1/\delta}}{\epsilon} \right]  \\
                                       &= \operatorname{Pr}\left[\operatorname{Lap}(1/\epsilon) \leq  \frac{\log{1/\delta}}{\epsilon} - d \right]  \\
                                       &= \operatorname{Pr}\left[\operatorname{Lap}(1/\epsilon) \leq  \frac{\log{1/\delta}}{\epsilon} - d \bigg|d>\frac{(t+1)\log{n}}{\epsilon}\right]\operatorname{Pr}\left[d>\frac{(t+1)\log{n}}{\epsilon}\right] \nonumber \\
                                       & \hspace{0.1in} + \operatorname{Pr}\left[\operatorname{Lap}(1/\epsilon) \leq  \frac{\log{1/\delta}}{\epsilon} - d \bigg|d\leq\frac{(t+1)\log{n}}{\epsilon}\right] \operatorname{Pr}\left[d \leq\frac{(t+1)\log{n}}{\epsilon}\right] \\
                                       &\leq \operatorname{Pr}\left[\operatorname{Lap}(1/\epsilon) \leq  \frac{\log{1/\delta} - (t+1)\log{n}}{\epsilon}\right] \operatorname{Pr}\left[d>\frac{(t+1)\log{n}}{\epsilon} \right] + \operatorname{Pr}\left[d <\frac{(t+1)\log{n}}{\epsilon}\right] \\
                                           \end{align}
                                       \begin{align}
                                       &\leq \operatorname{Pr}\left[\operatorname{Lap}(1/\epsilon) \leq  \frac{t\log{n} - (t+1)\log{n}}{\epsilon}\right] + \mathcal{O}(n^{-1}) \\ 
                                       &\leq \operatorname{Pr}\left[\operatorname{Lap}(1/\epsilon) \leq  \frac{-\log{n}}{\epsilon} \right] + \mathcal{O}(n^{-1}) \\
                                       &\leq \operatorname{Pr}\left[|\operatorname{Lap}(1/\epsilon)| \geq  \frac{\log{n}}{\epsilon} \right] + \mathcal{O}(n^{-1}) \\
                                       &\leq n^{-1}+ \mathcal{O}(n^{-1}) \\
                                       &= \mathcal{O}(n^{-1})
  \end{align}
  
  \noindent
  Finally, we have that $\mathcal{M}^f_{\operatorname{Stability}}(G)$ outputs $f(G)$ with probability at least $1-\mathcal{O}(n^{-1})$.
\end{proof}

Proving the optimality of the SDP based algorithm is very challenging; \cite{hajek2016achieving} use a sophisticated dual certificate, and use it to show that the SDP solution is optimal, with high probability (Theorem 4 of~\cite{hajek2016achieving_extensions}). However, this probability (which is $1-1/n^{O(1)}$) is not high enough for a union bound to ensure stability for all graphs within $\mathcal{O}(\log{n})$ distance. Our main technical contribution in this analysis is a slightly different certificate, which ensures that the SDP solution is optimal with probability $\bm{1}$; we refer to this certificate as  ``concentration''.

\begin{definition} A graph $G$ is called $(c_1, c_2, c_3, c_4)$-concentrated if it satisfies all four (4) conditions below
  
  \begin{itemize}
  \item $min_{i\in V(G)} (s_i-r_i) > c_1 \log{n}$
  \item $\Vert A - \E[A] \Vert_2 \leq c_2\sqrt{\log{n}}$
  \item $max_{k\in[r]}\frac{1}{K}\sum_{i\in C_k}r_i \leq Kq + c_3\sqrt{\log{n}}$, \mbox{where $K=n/r$},
  \item $e(C_k, C_{k'}) \geq K^2q -(3/4)K\sqrt{\log{n}}-c_4\log{n}$
  \end{itemize}
  
in which $s_i$ is the number of same-community neighbors of node $i$ and $r_i$ is the maximum number of neighbors of i in one of the other communities; $e(C_k, C_{k'})$ is the number of inter-community edges between communities $k$ and $k' \neq k$.
\end{definition}

We note that the bound in the first condition is stronger than the one in Lemma 4 of~\cite{hajek2016achieving_extensions}.

Next we prove that in some regimes of the SBM, the concentration of a graph $G$ generated by the SBM holds with high probability.~\cite{hajek2016achieving_extensions} proves the second and the third conditions holds with probability at least $1-1/poly(n)$. We will prove that the first and the last condition will hold with probability at least $1-O(1/n)$ to complete the Lemma.

\begin{lemma}
  \label{lemma:sdp-concen-high-prob}
  A graph $G$ generated by an SBM is $(c_1, c_2, c_3, c_4)$-concentrated with probability at least $1-\mathcal{O}(n^{-1})$  for some constants $(c_1, c_2, c_3, c_4)$ with $0 < c_1 \leq \frac{(\sqrt{a}-\sqrt{b})^2/r - 2}{1 + \frac{1}{2}\log{\frac{a}{b}}}$ and $\sqrt{a} - \sqrt{b} >\sqrt{2r}$. 
\end{lemma}

\begin{proof}
The second condition follows from Theorem 5 of \cite{hajek2016achieving}.
The third condition has been shown to hold with high probability in Lemma 5 of~\cite{hajek2016achieving_extensions}. 
Therefore, we only need to prove the first and last conditions will hold with high probability to complete the lemma.

  \emph{The first condition's proof.}
  Let $r_i(k)$ be the number of cross-community neighbors of node $i$ in community $k$ ($i$ is not in community $k$). We have $\forall k: r_i(k) \leq r_i$.
  
  
  

  
  
  Fix a node $i$. First we notice that $s_i \sim Binom(K, a\log{n})$ and $r_i(k) \sim Binom(K, b\log{n})$. Applying the result of Lemma~\ref{lemma:optimized_tail_bound}, substituting $m_1 = m_2 = K = n/r$, $0 < f_n = c_1\log{n} \leq \frac{a-b}{r}\log{n}$ we have $\Pr[s_i - r_i(k) < c_1\log{n}] \leq n^{-g(1/r, 1/r, a, b, c_1)}$. Our goal is to find $c_1$ such that $g(\ldots) \geq 2$, so that we can apply union bound over $r$ communities and $n$ nodes to complete the statement.
  
  We have for any $0 < c_1 \leq \frac{(\sqrt{a}-\sqrt{b})^2/r - 2}{1 + \frac{1}{2}\log{\frac{a}{b}}}$:
  
  \begin{align}
    g(1/r,1/r,a, b, c_1) &= \frac{a+b}{r}- \sqrt{c_1^2 + 4(K^2/n^2)ab} - (c_1/2)\log{\frac{(\gamma-\alpha)a}{(\gamma+\alpha)b}} \\
    &\geq \frac{a+b}{r} - c_1 -\frac{2\sqrt{ab}}{r} - \frac{c_1}{2}\log{\frac{a}{b}} \\
    &= \frac{(\sqrt{a} - \sqrt{b})^2}{r}  - c_1\left(1 + \frac{1}{2}\log{\frac{a}{b}}\right) \\
    &\geq  \frac{(\sqrt{a} - \sqrt{b})^2}{r}  - \frac{(\sqrt{a}-\sqrt{b})^2/r - 2}{1 + \frac{1}{2}\log{\frac{a}{b}}}\left(1 + \frac{1}{2}\log{\frac{a}{b}}\right) \\
    &=2,
  \end{align}

 where the first inequality holds because of the inequality $\sqrt{x} + \sqrt{y} > \sqrt{x + y}$ and $\log{\frac{(\gamma-\alpha)a}{(\gamma+\alpha)b}}\leq \log{\frac{a}{b}}$, and the second inequality holds because $c_1 \leq \frac{(\sqrt{a}-\sqrt{b})^2/r - 2}{1 + \frac{1}{2}\log{\frac{a}{b}}}$.
 Now we have that $\Pr[s_i - r_i(k) < c_1\log{n}] \leq n^{-2}$. We also note that since $(\sqrt{a}-\sqrt{b})^2 < a-b$ for $a > b > 0$, $\frac{(\sqrt{a}-\sqrt{b})^2/r - 2}{1 + \frac{1}{2}\log{\frac{a}{b}}} < \frac{a-b}{r}$, hence $c_1 < \frac{a-b}{r}$ which satisfies the condition of Lemma~\ref{lemma:optimized_tail_bound}.
  

  Taking the union bound over all $r$ communities, we have:
  \begin{align}
  \Pr\left[ s_i - r_i < c_1\log{n}\right] &\leq rn^{-2}
  \end{align}
  
  Taking the union bound on all node $i$, the lemma follows that $\Pr\left[ \forall i \in V: s_i - r_i < c_1\log{n}\right] < rn^{-1}$ where $0 < c_1 \leq \frac{(\sqrt{a}-\sqrt{b})^2/r - 2}{1 + \frac{1}{2}\log{\frac{a}{b}}}$.

  \emph{The last condition's proof.} We first notice that $e(C_k, C_{k'}) \sim Binom (K^2, q)$, since there are $K^2$ pairs of nodes and the probability of edges between each pair is $q$. By Chernoff's bound, we have:
  
  \begin{align}
      \Pr\left[e(C_k, C_{k'}) < (1-\alpha)K^2q \right] &\leq e^{-\alpha^2K^2q/2}
  \end{align}
  
  Set $\alpha = \frac{(3/4)K\sqrt{\log{n}} + c_4\log{n}}{K^2q}$. We notice that $|c_4\log{n} | \ll (K/4)\sqrt{\log{n}}$, therefore $\alpha \geq \frac{(K/2)\sqrt{\log{n}}}{K^2q} = \frac{\sqrt{\log{n}}}{2Kq}$
  
  Therefore we have:
  \begin{align}
      \Pr\left[e(C_k, C_{k'}) < \left(1 - \frac{(3/4)K\sqrt{\log{n}} + c_4\log{n}}{K^2q}\right)K^2q\right] &\leq e^{-\left(\frac{(3/4)K\sqrt{\log{n}} + c_4\log{n}}{K^2q}\right)^2K^2q/2} \\
      \Pr\left[ e(C_k, C_{k'}) < K^2q - (3/4)K\sqrt{\log{n}} - c_4\log{n}\right] &\leq e^{-(\frac{\sqrt{\log{n}}}{2Kq})^2K^2q/2} \\
      &= e^{-\frac{\log{n}}{8q}} \\
      &= n^{-\frac{1}{8q}}
  \end{align}
  
  Taking the union bound over all $k$ and $k'$, we have that $e(C_k, C_{k'}) \geq K^2q -(3/4)K\sqrt{\log{n}}-c_4\log{n}$ with probability at least $1-r^2n^{-\mathcal{O}(\frac{n}{\log{n}})} > 1-\mathcal{O}(n^{-1})$.
  
  Taking the union bound over all four conditions, the Lemma follows.
\end{proof}

Next we prove that the concentration persists under $\Omega(\log{n})$ edge perturbations. Specifically, we prove that if a graph is concentrated, a graph obtained by flipping up to $\Omega(\log{n})$ connections of the original one is also concentrated, albeit with slightly different tuple of constants.

\begin{lemma}
  \label{lemma:sdp-concen-r}
  If a graph $G$ is $(c_1, c_2, c_3, c_4)$-concentrated, all graphs $G'$ at distance at most $c\log{n}/\epsilon$ are $(c_1', c_2', c_3', c_4')$-concentrated with $c_1' = c_1 - c/\epsilon, c_2' = c_2 + \sqrt{2c/\epsilon}, c_3' = c_3 + c/\epsilon, c_4' = c_4 + c/\epsilon$.
\end{lemma}

\begin{proof}
  \noindent
  For the first condition:
  
  \begin{align}
    \min_{i\in V(G)} (s_i - r_i) &\geq \min_{i\in V(G)}\left( s_i - r_i - \frac{c\log{n}}{\epsilon}\right) \\
                                &\geq \min_{i\in V(G)} (s_i - r_i) - \frac{c\log{n}}{\epsilon} \\
                                &\geq (c_1 - c/\epsilon)\log{n}\\
                                &= c_1' \log{n}
  \end{align}
  
  \noindent
  Second, let $\mathbf{A}'$ be the adjacency matrix of graph $G'$. We have $\E[\mathbf{A}'] = \E[\mathbf{A}] = \bar{\mathbf{A}}$ with the assumption that both $G$ and $G'$ are generated by the same SBM. We have:
  
  \begin{align}
    \Vert \mathbf{A}' -\bar{\mathbf{A}} \Vert_2 &\leq \Vert \mathbf{A} + (\mathbf{A}' - \mathbf{A}) -\bar{\mathbf{A}} \Vert_2 \\
                              &\leq \Vert \mathbf{A} - \bar{\mathbf{A}} \Vert_2 + \Vert \mathbf{A}' - \mathbf{A} \Vert_2 \\
                              &\leq c_2\sqrt{\log{n}} + \Vert \mathbf{A}' - \mathbf{A} \Vert_F \\
                              &\leq c_2\sqrt{\log{n}} + \sqrt{\frac{2c\log{n}}{\epsilon}} \\
                              &= (c_2 + \sqrt{2c/\epsilon})\sqrt{\log{n}} \\ 
                              &= c_2'\sqrt{\log{n}}
  \end{align}
  
  \noindent
  Third, the third condition is:
  \begin{align}
    \max_{k\in[r]} \frac{1}{K}\sum_{i\in C_k} r_i' &\leq \max_{k\in[r]} \frac{1}{K}(\sum_{i\in C_k} r_i + \frac{c\log{n}}{\epsilon}) \\
                                                  &\leq \max_{k\in[r]} \frac{1}{K}\sum_{i\in C_k} r_i + \frac{1}{K}\frac{c\log{n}}{\epsilon} \\
                                                  &\leq Kq + c_3\sqrt{\log{n}} + \frac{cr\log{n}}{n\epsilon} \\
                                                  &= Kq + (c_3 + c\frac{r\sqrt{\log{n}}}{\epsilon n})\sqrt{\log{n}}\\
                                                  &\leq Kq + (c_3 + c/\epsilon)\sqrt{\log{n}} \\
                                                  &= Kq + c_3'\sqrt{\log{n}}
  \end{align}
  
  Finally, the last condition is:
  
  \begin{align}
      e(C'_k, C'_{k'}) &\geq e(C_k, C_{k'}) - \frac{c\log{n}}{\epsilon} \\
      &= K^2q -(3/4)K\sqrt{\log{n}}-c_4\log{n} -\frac{c\log{n}}{\epsilon} \\
      &= K^2q -(3/4)K\sqrt{\log{n}} - (c_4 + c/\epsilon) \log{n} \\
      &= K^2q -(3/4)K\sqrt{\log{n}} - c_4'\log{n}
  \end{align}
\end{proof}

\noindent
  We follow~\cite{hajek2016achieving_extensions} to prove that when a graph is $(c_1, c_2, c_3, c_4)$-concentrated for some constants $c_i$, the SDP relaxation (SDP for short) outputs the optimal ground truth community vector. First we note that the SDP can be presented by the following form:

  \begin{align}
    \mbox{maximize  }&\langle \mathbf{A}, \mathbf{Z} \rangle \\
    \mbox{subsect to  } \mathbf{Z} &\curlyeqsucc 0\\
    Z_{ii} &= 1, \forall i \in [n] \\
    Z_{ij} &\geq 0, \forall i, j \in [n]\\
    \mathbf{Z}\bm{1} &= K\bm{1},
  \end{align}
with $K = n/r$ is the size of each community.

\noindent
Then the following Lemma provides the condition for a dual certificate (deterministically). Intuitively, if we can construct a positive semi-definite matrix $S^*$ by the following formula without violating the two constraints below, the SDP is uniquely optimal at $\mathbf{Z}^*$ constructed by the ground truth community label (We say SDP(G) is optimal at the ground truth community label for short).

\begin{lemma}
  \label{lemma:sdp-r-optimal-condition}
  \textbf{Lemma 6 of~\cite{hajek2016achieving_extensions}}.
  Suppose there exists $\mathbf{D}^* = diag(d_i^*)$ with $d_i^*>0$  for all $i, \mathbf{B}^*\in\mathcal{S}^n$ with $\mathbf{B}^*\geq \mathbf{0}$ and $\mathbf{B^*}_{ij}>0$ whenever $i$ and $j$ are in distinct clusters, and $\lambda^*\in\mathbb{R}^n$ such that $\mathbf{S}^* \triangleq \mathbf{D}^* - \mathbf{B}^* - \mathbf{A} + \mathbf{\lambda}^*\bm{1}^T + \bm{1}(\mathbf{\lambda}^*)^T$ satisfies $\mathbf{S}^*\curlyeqsucc0$ and

  \begin{align}
    \mathbf{S}^*\mathbf{\xi}^*_k &= 0, \forall k\in[r]\\
    B_{ij}^*Z_{ij}^* &= 0, \forall i,j \in[n]
  \end{align}

  Then $\text{SDP}(G) = \mathbf{Z}^*$ is the unique solution for the SDP.
\end{lemma}

In the following statement, we claim that when the concentration holds, the SDP outputs the uniquely optimal ground truth community deterministically. We closely follow the proof of Theorem 4 of~\cite{hajek2016achieving_extensions} but replacing their high probability bounds by our concentration conditions. We prove that the concentration of the input graph implies the existence of a positive semi-definite matrix $S^*$ satisfies Lemma~\ref{lemma:sdp-r-optimal-condition}.

\begin{lemma}
  \label{lemma:sdp-optimal-r}
  When a graph $G$ is $(c_1, c_2, c_3, c_4)$-concentrated for some constants $c_j$ and $c_1>0$, the SDP outputs optimal ground truth community. 
\end{lemma}


\begin{proof}
  \noindent
  By the result of Lemma~\ref{lemma:sdp-r-optimal-condition}, we will prove that there exists  $\mathbf{S}^*\curlyeqsucc0$ which satisfies the two constraints above, and the Lemma follows. The main idea is to specify a way to construct $\mathbf{B^*}$, $\mathbf{D^*}$, and $\lambda^*$ that satisfy all properties above. Theorem 4 of~\cite{hajek2016achieving_extensions} defines $\mathbf{B^*}$, $\mathbf{D^*}$, and $\lambda^*$ as follows, for $i\in C_k$ and $j\in C_{k'}$:
  
  \begin{align}
    u_{kk'} &= \frac{1}{2K}\left(\frac{e(C_k, C_{k'})}{K} - Kq + \sqrt{\log{n}}\right) \\
    y^*_{kk'}(i) &= \frac{1}{K}(r_i - e(i, C_{k'})) + u_{kk'} \\
    z^*_{kk'}(j) &= \frac{1}{K}(r_j - e(j, C_{k})) + u_{kk'} \\
    \mathbf{B^*_{C_k\times C_{k'}}}(i,j) &= y^*_{kk'} (i) + z^*_{kk'}(j), \forall, 1 \leq k \leq k' \leq r \\
    \alpha_k &= \frac{1}{2}(Kq - \sqrt{\log{n}}) \\
    d^*_i &= e(i, C_k) - r_i + 2\alpha_k - \frac{1}{K}\sum_{i\in C_k} r_i \\
    \lambda^*_i &= \frac{1}{K}(r_i - \alpha_k) \mbox{  , for $i \in C_k$} 
  \end{align}
  
  \noindent
  We follow the proof of Theorem 4 of~\cite{hajek2016achieving_extensions} to prove $\mathbf{S}^*\curlyeqsucc0$. We note that the main difference between our proof and Theorem 4 of~\cite{hajek2016achieving_extensions} is that our proof proves the statement (about $\mathbf{S}^*$) is always true under the concentration condition while~\cite{hajek2016achieving_extensions} proves the statement is true with high probability over the SBM. We also need a different bound on $\min_{i\in V}(s_i-r_i)$ from~\cite{hajek2016achieving_extensions} to tolerate the change of up to $c\log{n}/\epsilon$ connections later.

  \noindent
  Let $E$ be the subspace spanned by vectors $\{\xi_k^*\}_k\in[r]$, i.e., $E = \operatorname{span}(\mathbf{\xi}_k^*: k\in[r])$. We show that
  \begin{align}
    x^T \mathbf{S}^* x > 0\mbox{ } {\forall x \perp E, \Vert x \Vert_2 = 1}
  \end{align}

  when the input graph $G$ is $(c_1, c_2, c_3, c_4)$-concentrated and $n$ is large enough. Note that this is sufficient to imply that $x^TS^*x\geq 0$ for all $x$ because $S^*\xi^*_k=0$ for all $k$ (as shown in the proof of Theorem 4 of~\cite{hajek2016achieving_extensions}), which implies $S^*x=0$ for all $x\in E$. 

  \noindent
  Note that $\E[\mathbf{A}] = (p-q) \mathbf{Z}^* + q\bm{J} - p\bm{I}$ and $\mathbf{Z}^* = \sum_{k\in[r]}\xi_k^*(\xi_k^*)^T$. For any $x$ such that $x \perp E$ and $\Vert x \Vert_2 = 1$,

  \begin{align}
    x^T \mathbf{S}^* x &= x^T \mathbf{D}^* x - x^T \E[\mathbf{A}] x - x^T \mathbf{B}^*x + 2x^T\lambda^*\bm{1}^Tx - x^T(\mathbf{A} - \E[\mathbf{A}])x \\
    &= x^T \mathbf{D}^* x - (p-q)x^T \mathbf{Z^*} x - qx^T \mathbf{J}x + p - x^T \mathbf{B}^*x - x^T(\mathbf{A} - \E[\mathbf{A}])x \\
            &= x^T\mathbf{D}^*x + p - x^T\mathbf{B}^*x - x^T(\mathbf{A} - \E[\mathbf{A}])x \\
            &\geq x^T\mathbf{D}^*x + p - x^T\mathbf{B}^*x - \Vert \mathbf{A} - \E[\mathbf{A}]\Vert_2 \\
            &\geq x^T\mathbf{D}^*x + p - x^T\mathbf{B}^*x - c_2\sqrt{\log{n}},
  \end{align}
  
  where the second equality holds because $\E[\mathbf{A}] = (p-q) \mathbf{Z}^* + q\bm{J} - p\bm{I}$ and $x \perp \mathbf{1}$; and the third equality holds because $\langle x, \xi^*_k \rangle = 0$ and $x \perp \mathbf{1}$.

\noindent
Theorem 4 of~\cite{hajek2016achieving_extensions} shows that  $\mathbf{B}^*$ can be chosen such that for any $x \perp E$, we have $x^T \mathbf{B}^* x = 0$ where both constraints of Lemma~\ref{lemma:sdp-r-optimal-condition} are satisfied. 


From the definition of $D^*$, we have
  \begin{align}
    \lambda_{\min}(D^*) &\geq \min_i d_i^* \\
                        &= \min_i e(i, C_k) - r_i + 2\alpha_k - \frac{1}{K}\sum_{i\in C_k} r_i \\
                        &= \min_i e(i, C_k) - r_i + Kq - \sqrt{\log{n}} - \frac{1}{K}\sum_{i\in C_k} r_i \\
                        &= \min_i e(i, C_k) - r_i - \sqrt{\log{n}} - \left(\frac{1}{K}\sum_{i\in C_k} r_i - Kq\right)\\
                       &\geq \min_i (e(i, C_k) - r_i) - (c_3+1)\sqrt{\log{n}} \\
    &\geq c_1\log{n} - (c_3+1)\sqrt{\log{n}},
  \end{align}
where the second inequality holds because from the third condition of concentration, $\frac{1}{K}\sum_{i\in C_k} r_i - Kq \leq c_3\sqrt{\log{n}}$; and the last inequality holds because from the first condition of concentration, $e(i, C_k) - r_i) \geq c_1\log{n}$ .
  We then have:

  \begin{align}
    x^T \mathbf{D}^* x &\geq \lambda_{\min}(\mathbf{D}^*)\Vert x \Vert^2_2 \\
            &\geq \lambda_{\min}(\mathbf{D}^*)\\
            &\geq \min_i d_i^* \\
    &\geq c_1\log{n} - (c_3+1)\sqrt{\log{n}}.
  \end{align}

  \noindent

  With that, we simplify $x^T\mathbf{S}^*x$ with the assumption that the graph is $(c_1, c_2, c_3, c_4)$-concentrated:

  \begin{align}
    x^T\mathbf{S}^*x &\geq c_1\log{n} - (c_2+c_3+1)\sqrt{\log{n}} \\
    &\geq 0 \mbox{  when $n$ is large enough.}
  \end{align}
  
  Finally, to guarantee that $\mathbf{B^*}_{ij}>0$ whenever $i$ and $j$ are in distinct clusters, we will prove $y^*_{kk'}(i) \text{ and }z^*_{kk}(j) > 0$. From their definitions, we see that $r_i - e(i, C_{k'}) \geq 0$ and $r_j - e(j, C_{k}) \geq 0$, therefore we need to prove $u_{kk'} > 0$ for all $k$ and $k'\neq k$ to complete the proof, i.e., to prove $e(C_k, C_{k'}) > K^2q - K\sqrt{\log{n}}$. The fourth condition of $(c_1, c_2, c_3, c_4)$-concentration says that there is a constant $c_4$ such that $e(C_k, C_{k'}) \geq K^2q - (3/4)K\sqrt{\log{n}} - c_4\log{n}$ for all $k, k'$. Since $c_4\log{n} \ll (K/4)\sqrt{\log{n}}$ for $n$ large enough, it confirms that $e(C_k, C_{k'}) > K^2q - K\sqrt{\log{n}}$, $\mathbf{B^*}$ satisfies above conditions.
  
  Apply the result of Lemma~\ref{lemma:sdp-r-optimal-condition}, the Lemma follows.
\end{proof}
  
Next we prove that if a graph is appropriately concentrated, it is also stable under the SDP relaxation.
  
\begin{lemma}
  \label{lemma:sdp-stable-r}
  When a graph $G$ is $(c_1, c_2, c_3, c_4)$-concentrated and $c_1 - c/\epsilon > 0$, it is also $\frac{c\log{n}}{\epsilon}$-stable.
\end{lemma}

\begin{proof}

  \noindent
  When a graph $G$ is $(c_1, c_2, c_3, c_4)$-concentrated, all graph $G'$ at distance at most $\frac{c\log{n}}{\epsilon}$, formally $d(G', G)\leq \frac{c\log{n}}{\epsilon}$, are $(c_1',c_2',c_3', c_4')$-concentrated with $c_1' = c_1 -c/\epsilon > 0$ as the result of Lemma~\ref{lemma:sdp-concen-r}. Hence, in the radius of $c\log{n}/\epsilon$ from $G$, all graph $G'$ has that $SDP(G')$ is unique and optimal at $\mathbf{Z^*}$ constructed by the ground truth communities with $n$ large enough.

  \noindent
  From that, for all graphs $G'$ such that $d(G, G') \leq \frac{c\log{n}}{\epsilon}$, we have $\text{SDP}(G) = \text{SDP}(G')$ and the lemma follows. 
\end{proof}

\begin{lemma}
  \label{lemma:sdp-stable-high-prob}
%
  A graph $G$ generated by an SBM is $\Omega(\frac{c\log{n}}{\epsilon})$-stable with respect to $SDP$ function with probability at least $1-\mathcal{O}(n^{-1})$ if $\sqrt{a} - \sqrt{b} > \sqrt{r}\sqrt{2 +\frac{c}{\epsilon}(1+\frac{1}{2}\log{\frac{a}{b}})})$
\end{lemma}

\begin{proof}
  \noindent
  By the result of Lemma~\ref{lemma:sdp-concen-high-prob} and Lemma~\ref{lemma:sdp-stable-r}, the lemma follows. We note that from Lemma~\ref{lemma:sdp-concen-high-prob}, we can select $c_1 = \frac{(\sqrt{a}-\sqrt{b})^2/r - 2}{1 + \frac{1}{2}\log{\frac{a}{b}}}$ and hence we need $\frac{(\sqrt{a}-\sqrt{b})^2/r - 2}{1 + \frac{1}{2}\log{\frac{a}{b}}}> \frac{c}{\epsilon}$ to satisfy the condition of Lemma~\ref{lemma:sdp-stable-r}:
  
 \begin{align}
     &\frac{(\sqrt{a}-\sqrt{b})^2/r - 2}{1 + \frac{1}{2}\log{\frac{a}{b}}} > \frac{c}{\epsilon} \\
   \iff  &\frac{(\sqrt{a}-\sqrt{b})^2}{r} > 2 +  \frac{c}{\epsilon}\left( 1 + \frac{1}{2}\log{\frac{a}{b}}\right) \\
   \iff  &\sqrt{a}-\sqrt{b} > \sqrt{r}\sqrt{2 +  \frac{c}{\epsilon}\left( 1 + \frac{1}{2}\log{\frac{a}{b}}\right)}
 \end{align} 
\end{proof}

Finally, we prove that our mechanism outputs the ground-truth community label with high probability if the SBM satisfies two conditions for the stability. The first condition allows the concentration to holds with high probability and the second condition makes the concentration to persist under edge perturbation (of up to $\Omega(\log{n})$ connections).


\begin{theorem}
%
  Given a graph $G$ is generated by an $r$-community SBM model with $\sqrt{a} - \sqrt{b} > \sqrt{r}(\sqrt{2 +\frac{t+1}{\epsilon}(1+\frac{1}{2}\log{\frac{a}{b}})})$, $\mathcal{M}^{SDP}_{\operatorname{Stability}}$ with $\delta=n^{-t}$ outputs $\mathbf{Z^*}$ constructed by the ground-truth community vector w.h.p., \\
  i.e. $\Pr[\mathcal{M}_{Stability\_SDP}(G) \neq \mathbf{Z^*}] = o(1)$. 
  
\end{theorem}

\begin{proof}
  \noindent
  
  Lemma~\ref{lemma:sdp-stable-high-prob} states that a graph $G$ generated by the $r$-community SBM is $(c\log{n}/\epsilon)$-stable under $SDP$ w.h.p.. By applying Lemma~\ref{lemma:f-stable}, substituting the generic function $f$ by $SDP$ and $c=t+1$, $\mathcal{M}^{SDP}_{\operatorname{Stability}}$ outputs $SDP(G)$ w.h.p.. Given that $\Pr[SDP(G)\neq \mathbf{Z^*}] = o(1)$\cite{hajek2016achieving}, the Theorem follows.
  
  
\end{proof}


